\documentclass[letterpaper,11pt]{article}
\usepackage[margin=1in]{geometry}

\usepackage{xpatch}

\usepackage{authblk}

\usepackage{amsmath,amssymb,amsfonts,enumerate,enumitem}
\usepackage{amsthm}

 \usepackage[stable]{footmisc}
\usepackage{subcaption}
\usepackage{xcolor}
\usepackage{floatrow}
\usepackage{bm,bbm}
\usepackage{relsize}

\usepackage{url}
\usepackage{cleveref}

\usepackage{thmtools} 
\usepackage{thm-restate}

\usepackage{mathtools}
\usepackage{commath}

\usepackage{relsize}
\usepackage{array}
\usepackage{upgreek}

\usepackage{enumitem}

\usepackage{multicol}
\usepackage{multirow}

\usepackage{microtype}

\usepackage{algorithm}
\usepackage{algpseudocode}

\usepackage[leftcaption]{sidecap}
\sidecaptionvpos{figure}{c}
\usepackage{tikz}
\usepackage{pgfplots}
\pgfplotsset{compat=newest}
\usetikzlibrary{patterns}
\usetikzlibrary{plotmarks}
\usepgfplotslibrary{fillbetween}

\usetikzlibrary{arrows}
\usetikzlibrary{arrows.meta}
\usetikzlibrary{bending}

\usepackage{nicefrac}

\usepackage{wasysym}

\newcommand{\R}{\mathbb{R}}

\newcommand{\N}{\mathbb{N}}

\newcommand{\1}{\mathbbm{1}}

\DeclareMathOperator*{\Expect}{\mathbb{E}}
\DeclareMathOperator*{\Var}{\mathbb{V}}
\DeclareMathOperator*{\Cov}{\mathbb{C}}
\DeclareMathOperator*{\Prob}{\mathbb{P}}

\DeclareMathOperator*{\ess}{ess}

\newcommand{\distributed}{\thicksim}
\DeclareMathOperator{\Support}{Support}
\newcommand{\X}{S}

\newcommand{\M}{\mathcal{M}}

\newcommand{\Cycle}{\mathcal C}

\newcommand{\PColor}{\gamma}

\DeclareMathOperator{\TVD}{TVD}

\newcommand{\Pl}{\mathcal{P}}
\newcommand{\Dmax}{d_{\rm max}}
\newcommand{\TMix}{\tau_{\rm mix}}
\newcommand{\TRel}{\tau_{\rm rel}}
\newcommand{\TfMix}{\tau_{\rm fmix}}
\newcommand{\TfRel}{\tau_{\rm frel}}
\newcommand{\TUnif}{\tau_{\rm unif}} %TODO: INFINITY?

\newcommand{\TMixBound}{\mathcal{T}_{\!\rm mix}}
\newcommand{\TRelBound}{\mathcal{T}_{\!\rm rel}}
\newcommand{\TfMixBound}{\mathcal{T}_{\!\rm fmix}}

\newcommand{\EigTwo}{\lambda}
\newcommand{\SG}{(1-\EigTwo)}

\newcommand{\EigTwoBound}{\Lambda}
\newcommand{\HPSC}{{\textsc{hiper}m}}

\newcommand{\PiMin}{\pi_{\min}}
\newcommand{\PiMinBound}{\pi_{\min}}

\newcommand{\SVar}{v_{\pi}} %Stationary variance
\newcommand{\AVar}{v_{\mathrm{asy}}} %Asymptotic variance

\newcommand{\SVarBound}{V_{\pi}} %Stationary variance
\newcommand{\AVarBound}{V_{\mathrm{asy}}} %Asymptotic variance

\newcommand{\VAsy}{\AVar}

\newcommand{\Domain}{S}

\newcommand{\ITVn}[1]{\mathrm{trv}^{(#1)}}
\newcommand{\ITVar}{\ITVn{T}}
\newcommand{\ITRelVar}{\ITVn{\TRel}}

\newcommand{\VED}{\hat{v}}

\newcommand{\AutoCov}{C}

\newcommand{\F}{{\cal F}}

\newcommand{\Favg}{f_{\rm avg}}%F average

\newcommand{\UnifD}{U}%{\mathcal{U}}
\newcommand{\BinomD}{\mathcal{B}}

\newcommand{\frange}{R}

\newcommand{\NIterations}{I}

\newcommand{\Pcolor}{\gamma}

\newcommand{\cyrus}[1]{\textcolor{green!50!black}{\textsc{Cyrus:} \emph{#1}}}
\newcommand{\shahrzad}[1]{\textcolor{blue!60!black}{\textsc{Shahrzad:} \emph{#1}}}
\newcommand{\todo}[1]{\textcolor{red!50!black}{\textsc{ToDo:} \emph{#1}}}
\newcommand{\eli}[1]{\textcolor{red!50!black}{\textsc{Eli:} \emph{#1}}}

\makeatletter
\newcommand\longleftrightarrowfill@{%
  \arrowfill@\leftarrow\relbar\rightarrow}
\makeatother

\newtheorem{defin}{Definition}[section]
\newtheorem{theorem}{Theorem}[section]
\newtheorem*{theorem*}{Theorem}

\newtheorem{remark}[theorem]{Remark}

\newtheorem{lemma}[theorem]{Lemma}
\newtheorem{coro}[theorem]{Corollary}

\newtheorem{exm}[theorem]{Example}

\newcommand{\algo}{\textsc{DynaMITE}}
\newcommand{\DMCMC}{\textsc{McmcPro}}
\newcommand{\nsalgo}{\textsc{WarmStart\algo{}}}

\newcommand{\floor}[1]{\left\lfloor #1 \right\rfloor}
\newcommand{\ceil}[1]{\left\lceil #1 \right\rceil}

\DeclareMathOperator{\SD}{SD}

\newcommand{\LandauTheta}{\Theta}

\crefname{algorithm}{alg.}{algs.}
\Crefname{algorithm}{Algorithm}{Algorithms}
\crefname{appendix}{appx.}{appcs.}
\Crefname{appendix}{Appendix}{Appendices}
\crefname{corollary}{coro.}{coros.}
\Crefname{corollary}{Corollary}{Corollaries}
\crefname{conjecture}{conjecture}{conjectures}
\Crefname{conjecture}{Conjecture}{Conjectures}
\crefformat{conjecture}{conjecture~#2#1#3}
\crefmultiformat{conjecture}{Conjectures~#2#1#3}{\ and~#2#1#3}{, #2#1#3}{\ and~#2#1#3}
\crefname{definition}{def.}{defs.}
\Crefname{definition}{Definition}{Definition}
\crefname{figure}{fig.}{figs.}
\Crefname{figure}{Figure}{Figures}
\crefname{lemma}{lemma}{lemmas}
\Crefname{lemma}{Lemma}{Lemmas}
\crefname{proposition}{prop.}{props.}
\Crefname{proposition}{Proposition}{Propositions}
\Crefname{section}{Section}{Sections}
\crefname{section}{sect.}{sects.}
\crefname{subsection}{sect.}{sects.}
\Crefname{subsection}{Section}{Sections}
\crefname{subsubsection}{sect.}{sects.}
\Crefname{subsubsection}{Section}{Sections}
\crefname{table}{table}{tables}
\Crefname{table}{Table}{Tables}
\crefname{theorem}{thm.}{thms.}
\Crefname{theorem}{Theorem}{Theorems}

\Crefname{theorem}{Thm.}{Thms.}

\crefname{section}{\S}{\S}
\crefname{subsection}{\S}{\S}
\Crefname{subsection}{\S}{\S}
\crefname{subsubsection}{\S}{\S}
\Crefname{subsubsection}{\S}{\S}

\title{Making mean-estimation more efficient using an MCMC trace variance approach: \algo{}}
\author[1]{Cyrus Cousins}
\author[2]{Shahrzad Haddadan}
\author[3]{Eli Upfal}
\affil[1,2,3]{  Department of Computer Science, Brown University }
\affil[2]{  Data Science Initiative, Brown University }
{\affil[1]{
\footnotesize{\texttt{cyrus\_cousins@brown.edu}}}
\affil[2]{
\footnotesize{\texttt{shahrzad\_haddadan@brown.edu}}}
\affil[3]{
\footnotesize{\texttt{eli\_upfal@brown.edu}}}}
\newif\ifappendix
\appendixfalse

\newcommand{\arx}[1]{}
\renewcommand{\arx}[1]{#1}
\renewcommand{\todo}[1]{}
\renewcommand{\cyrus}[1]{}
\renewcommand{\shahrzad}[1]{}
\renewcommand{\eli}[1]{}

\date{}
\begin{document}
\maketitle

%\vspace{-1.5cm}
\begin{abstract}

 We introduce a novel
statistical measure for MCMC-mean estimation, the \emph{inter-trace variance} ${\rm trv}^{(\tau{\rm rel})}({\cal M},f)$, which  depends on a Markov chain ${\cal M}$ and a function $f:\Domain\to [a,b]$. We show that the inter-trace variance 
can be efficiently estimated from observed data, and that it leads to a more efficient MCMC-mean estimator, with complexity competitive with a lower-bound obtained from the central limit theorem of Markov chains. 
Most efficient 
 MCMC mean-estimators receive, as input, upper-bounds on chain-dependent terms like \emph{mixing time} $\tau_{\rm mix}$ or \emph{relaxation time} $\tau_{\rm rel}$, and often also function-dependent terms such as the \emph{stationary variance} $v_{\pi}$, and their performance is highly dependent to the sharpness of these bounds.
 In contrast, we introduce  \textsc{DynaMITE}, which dynamically adjusts the sample size using the observed data, and 
therefore it is less sensitive to the looseness of input upper-bounds on $\tau_{\rm rel}$, 
and requires no bound on $v_{\pi}$.

Receiving only an upper-bound ${\cal T}_{\rm rel}$ on $\tau_{\rm rel}$, \textsc{DynaMITE}  estimates w.h.p.\ the mean of $f$ to within $\varepsilon$ additive error in
$\tilde{\mathcal{O}}\bigl(\smash{\frac{{\mathcal T}_{\rm rel} R}{\varepsilon}} + \frac{\tau_{\mathrm{rel}}\cdot {\rm trv}^{(\tau{\mathrm{rel}})}}{\varepsilon^{2}}\bigr)$ steps, without \emph{a priori} bounds on the \emph{stationary variance} $v_{\pi}$ or the \emph{inter-trace variance} ${\rm trv}^{(\tau \mathrm{rel})}$. Thus we obtain minimal dependency on the tightness of ${\cal T}_{\rm mix}$, since the complexity is dominated by $\tau_{\rm rel}\rm{trv}^{(\tau{\rm rel})}$ as $\varepsilon \to 0$ 
(the high precision regime), even though the values of $\tau_{\rm rel}$ and $\rm{trv}^{(\tau{\rm rel})}$ are not known to the algorithm.
 Note that bounding $\TRel$ is known to be prohibitively difficult, however, \textsc{DynaMITE} is able to reduce its principal dependence on $\TRelBound$ to $\TRel$, simply by exploiting properties of the \emph{inter-trace variance}.
To compare our method to known variance-aware bounds, we show  ${\rm trv}^{(\tau{}\mathrm{rel})}({\cal M},f) \leq v_{\pi}$. Furthermore,  we show when $f$'s image is distributed symmetrically on ${\cal M}$'s traces, we have ${\rm trv}^{({\tau{\rm rel}})}({\cal M},f)=o(v_{\pi}(f))$, thus \textsc{DynaMITE} outperforms prior methods in these cases,
 even when tight bounds on variance and mixing or relaxation times are known.
We demonstrate the advantage of our estimator through an interesting application, counting $k$-colorings of a class of graphs, e.g., graphs drawn from the planted partition model, wherein using \textsc{DynaMITE} leads to significant improvement. 
%

%\textbf{Keywords:} Mean Estimation, MCMC Methods, Concentration Bounds, %TODO: keywords.
\end{abstract}

\section{Introduction}\label{sec:1}

Given a bounded, real-valued function $f:\Domain\to [a,b]$, and a distribution $\pi$ defined on the domain $\Domain$, our goal is to estimate the average of this function, $\mu \doteq \Expect_{\pi}[f(x)]$.
Having enough independent samples from $\pi$, one can estimate $\mu$ by the \emph{empirical mean} $\hat{\mu}$, the average value of $f(x)$ over the samples.
% The accuracy of this estimate depends on the variance of $f$ and the number of independent samples.
%Unfortunately, for many interesting problems, it is NP-hard to sample from $\pi$, but 
In many important applications it is hard (or even impossible) to efficiently generate independent samples from $\pi$. Instead, one can often generate a sequence of (nonindependent) samples by traversing a Markov chain with state space $S$ and stationary distribution $\pi$~\cite{mixingAld,countingleWinklerBrightwell,wilson, danasurvey, levin2017markov}.
\emph{MCMC-mean estimation} is the process of generating an estimate of $\mu = \Expect_{\pi}[f(x)]$ from the Markov chain samples. The efficiency of an MCMC-mean estimator 
is measured by the accuracy of the estimate it generates, and its \emph{(sample) complexity} - the (expected) number of Markov chain transitions it uses for generating the estimate. The number of Markov chain transitions clearly dominates the computational complexity of the estimator.

MCMC-Mean estimation 
%is an interesting problem, 
is a well-studied problem in probability theory, statistics, and computer science, with
numerous applications in a variety of fields such as \emph{statistical physics} \cite{statisticalphyref1,statisticalphyref2,ref1,staticref}, \emph{chemistry} \cite{chemref},
\emph{computational biology} \cite{vandin2016sample,vandin2012discovery,vandin2012novo,bio2011,bioref1,bioref2}, 
%and  various subfields of \emph{computer science} e.g. 
\emph{statistical machine learning}, \emph{image processing}, etc \cite{LearningFan,learningSalatino,learningYuan,image}.
MCMC-Mean estimation is also an important component  in solutions of related  hard computational problems.
For example, Jerrum, Valiant, and Vazirani \cite{countingsamplingreduction1986jurrumvaziranivalient}, and many others \cite{gibbsvigoda,hubergibbs,Kolmogorovgibbs} 
reduced estimating the \emph{size} of a \emph{self-reducible set} (often a \#hard problem) to solving a series of \emph{mean-estimation} tasks. 
% Algorithms for web scale networks sometimes assume random walk access models  
%  \cite{flavioHaddadan, peresParam,BeraSesh,falvioww}, which rely on MCMC-mean estimation results. %e.g.,  estimating the mean of a real-valued function on vertices \cite{flaviowww,flavioHaddadan}. 

A key concept in analyzing MCMC-estimates is the Markov chains' \emph{mixing time} (see section \ref{sec:prelim} for definition). Rigorously bounding mixing times needs potent analytical techniques and sophisticated problem-specific analysis  \cite{techguruswami2000rapidly,techAnari2020SpectralII,techbubley1997path,techvazirani1991rapidly} and it is an active area of research \cite{coloringVardi2018,exglad2017,blanca2021mixing,chen2021optimal,exAGOSTINI2019}.
Unfortunately these bounds are often loose which hampers their applicability, thus, practitioners often ignore mixing time bounds, and instead run the Markov chain until a heuristic termination condition is observed \cite{diagno2autho,diag4,diag2,diag1,diag3,diag5}. \arx{In particular, one popular method is to estimate the autocovariance between adjacent samples \cite{diagno2autho}, and terminate when it is negligible.}
Unfortunately,  correctness of such approaches is not supported by mathematical guarantees and in practice their applications are generally quite error-prone. Not surprisingly, there has been significant interest in  providing mathematical tools to analyze the computation complexity of \emph{estimating  mixing time from observations with no prior knowledge of it}  \cite{MosselHardness2011,peresParam,trelestHsuPeres2019,pmlr-v99-wolfer19a,pmlr-v117-wolfer20a}. These results are often negative, as they provide large lower bounds for this problem. Thus, some, even loose, upper bound on mixing time seems to be necessary for obtaining an efficient estimate.
%\cyrus{What goes wrong without mixing / relaxation time}

Thus, we are motivated to focus on the following question:

\smallskip 

\emph{Can a dynamic adaptive algorithm significantly improve the state of the art sample complexity of MCMC-mean estimator, and, in particular, minimize the complexity dependence on loose mixing time upper bounds?}
%
%
%  \emph{Is it possible to use available loose upper-bounds on mixing times, to  
%  design a self-terminating MCMC mean-estimation algorithm whose performance is minimally deteriorated by the looseness of bounds and furthermore  its computational complexity is smaller than the state of the art algorithms? }

\smallskip 
This paper presents and analyzes a novel MCMC-estimation algorithm \algo{}, that provides an affirmative answer to the above question.
% In the rest of this introduction, after providing the preliminaries and related work we explain how we  by introducing \algo{}. 

\subsection{Preliminaries}\label{sec:prelim}
%In this section, we  introduce notation that we use throughout the paper.
Let $f:\X\to [a,b]$ be a real valued function, and $\pi$ a probability distribution on $\X$.
We denote the \emph{mean} and \emph{%stationary 
variance} of $f$ by
\[
\mu \doteq \Expect_{x\distributed\pi}[f] \doteq \int_{\X}f(x) \, \mathrm{d}\pi(x) \quad \text{and} \quad \SVar \doteq \Var_{\pi}(f) \doteq \Expect_{x \distributed \pi}[(f(x)-\mu)^2] \enspace.
\]

\noindent 
\textbf{MCMC Terminology}
All the Markov chains we discuss are assumed to be ergodic and lazy\footnote{We explicitly note when reversibility is needed.} (see \cite{levin2017markov}), with state space $S$ and stationary distribution $\pi$.
We denote by $\M$ both the Markov chain and its \emph{transition matrix}. For any $x,y\in \X$, $\M(x,y)\doteq\Prob(X_i=y\vert X_{i-1}=x)$, and  $\M^k(x,y)=\Prob(X_{i+k}=y\vert X_{i}=x)$, where $\M^k$ is the standard matrix power notation.
Let $\nu$ be a probability distribution on $\X$, by $\M(\nu)$ we mean the distribution of $X_{i+1}$ conditioned on $X_{i}\distributed \nu$, and the distribution of the chain after $k$ steps of $\M$, starting at $X_i\sim \nu$ is denoted by $X_{i+k}=\M^{k}(\nu)$.

%  A Markov chain $\M$ on sample space $\X$ is a sequence of random variables $X_{1},X_2,\dots$ that is \emph{memoryless}.
% The evolution of $\M$ is characterized by its \emph{transition matrix}; here we abuse notation, and also use $\M$ to represent its transition matrix, i.e., for any $x,y\in \X$, $\M(x,y)\doteq\Prob(X_i=y\vert X_{i-1}=x)$, and thus $\Prob(X_{i+k}=y\vert X_{i}=x)=\M^k(x,y)$, where $\M^k$ is the standard matrix powering notation.
% When $\nu$ is a probability distribution on $\X$, by $\M(\nu)$ we mean the distribution of $X_i$ when $X_{i-1}\distributed \nu$. The distribution of the chain after $k$ steps of $\M$ is denoted by $X_{i+k}=\M^{k}(X_i)$.

A $\tau$-\emph{trace} of Markov chain $\M$ is a sequence of $\tau$ consecutive states visited by the chain. We use \emph{vector notation} to denote traces of Markov chains, e.g., $\vec{X}_{1:\tau}=(X_1,X_2,\dots X_{\tau})$, (we drop the subscript $1:\tau$ when the length is clear from context). 

We measure distance between distributions by the total variation distance (TVD). 
 For a precision parameter $\epsilon>0$, the 
 $\epsilon$-\emph{mixing time},  $\TMix(\epsilon)$,
 is the minimum $\tau$ satisfying $\TVD(\M^\tau(\nu),\pi )\leq \epsilon$, for any initial starting distribution $\nu$. In particular, we denote $\TMix \doteq \TMix(\nicefrac{1}{4})$, and note that  $\TMix(\epsilon) \leq \frac{1}{4} \log(\epsilon^{-1})  \TMix$.
  Let the \emph{second largest absolute eigenvalue}  $\M$ be $\lambda$, and   denote the \emph{relaxation time}  $\TRel \doteq (1-\lambda)^{-1}$. Note that %$\TRel \doteq (1-\lambda)^{-1}$  and it is closely related to 
 $\TMix$ and $\TRel$ are closely related as $
\left(\TRel(\M)-1\right)\ln(2)\leq \TMix(\M)\leq \TRel(\M) \ \ln\bigl(\smash{\nicefrac{2}{\sqrt{\pi_{\min}}}}\bigr)$, where $\pi_{\min}\doteq \min_{x\in\X}\pi(x)$,  see \cite{levin2017markov}.
%We use $T$ to denote a known upper-bound on the relaxation time. 

Since exact parameters of the MCMC process are often unknown and hard to compute, most MCMC estimators rely on some (possibly loose) bounds given as input to the algorithm.   
We denote by capital letters  $\TMixBound,\TRelBound,\EigTwoBound$ and $V_{\pi}$ the upper bounds given to the algorithm for $\TMix, \TRel, \lambda$ and $ \SVar$, respectively. 
%We denote the cardinality of a set $S$ by $\# S$.

For an $m$-trace of $\cal M$, $\vec{X} =( X_1,X_2, \dots, X_m )$ we define its \emph{empirical mean}  as %$\hat{\mu}(\vec{X})$  as
$
    \hat{\mu}(\vec{X}) \doteq \frac{1}{m} \sum_{i=1}^m f(X_i)  $. We are interested in designing algorithms which run a  Markov chain $\cal M$ for  $m$ steps, and return $\hat{\mu}(\vec{X}_{1:m})$ as an estimate of $\mu$. We refer to  any such algorithm as an \emph{MCMC-mean estimator}, and measure their (sample) complexity using  the following definition:
    \begin{defin}\label{def:complex}
 Assume algorithm $A$ runs a Markov chain $\M$ for  $m$  steps and generates an $m$-trace $X_1,X_2,\dots, X_m$. Algorithm $A$ uses this trace to find  an estimate $\hat{\mu}$ of $\mu=\Expect_\pi[f]$. If 
there exists  $m_{A}(\M,f,\varepsilon,\delta)$ such that for any $m\geq m_{A}(\M,f,\varepsilon,\delta)$, we have
%\begin{equation*}
    $\Prob (\vert \mu -\hat{\mu}\vert \geq \varepsilon ) \leq \delta$
%\end{equation*}
(an $(\epsilon,\delta)$-approximation of $\mu$),
 we call $m_A(\M,f,\varepsilon,\delta)$ the \emph{sample complexity} of algorithm $A$.
 %\cyrus{what type of complexity?}
\end{defin}
%\shahrzad{Reviewer's comment: needs to explicitly state the quantifiers on epsilon and delta. As it is, it looks like $m_A$ depends on both, so how can it be called the sample complexity?}\eli{sample complexity is always a function of epsilon and delta.}

\subsection{Related Work}\label{sec:related}  
The complexity of any MCMC-mean estimator depends, directly or through a proxy, on two unrelated parameters: a function dependent parameter, such as $\SVar$, the stationary variance of $f$; and chain dependent parameter, such as $\TMix$ or $\TRel$, the mixing or relaxation time of the Markov chain. Since these parameters are not known in general, efficient MCMC-mean estimator starts with some upper bound on these values, often through related parameters that are easier to estimate. For example, $\SVar$ can be bounded by $R^2$, and $\TMix$ can be bounded by $\TRel$ and $\lambda$.

The simplest MCMC-mean estimator averages samples taken at least $\TMixBound \geq \TMix$ Markov chain transitions apart. It is not hard to verify that the sample complexity of this estimator is $\Theta \left (\TMixBound(\nicefrac{\epsilon}{2R})(\nicefrac{\frange^{2}}{\varepsilon^{2}}) \log \frac{1}{\delta}\right) =\Omega\left(\TRel \cdot  (\nicefrac{\frange^{2}}{\varepsilon^{2}}) \log (\frac{\frange}{\varepsilon}) \log (\frac{1}{\delta})\right)$ chain transitions.

% % We now summarize the most notable existing MCMC-mean estimators. The complexity of  these methods is in terms of a \emph{variance proxy} term, e.g., range squared or stationary variance, which is often only function-dependent, and a \emph{mixing term}\cyrus{ergodicity parameter?}, e.g., mixing or relaxation time, which is often only chain-dependent.
% % %Generally the  (often function-dependent)
% % All of  these methods receive the mixing  parameters\cyrus{term?} and variance proxies  as input, thus in practice they depend on upper bounds of them. We denote these known upper bounds using capital letters like $\TMixBound,\TRelBound,\EigTwoBound,V_{\pi}$, etc.

% Using  sampling concentration bound like Chebyshev's, Hoeffding or Bernstein  bound for i.i.d samples (see e.g., \cite{mitzenmacher2017probability}) with, say, complexity $m(f,\varepsilon,\delta )$, we can drop in a Markov chain as a black-box sampling tool. 
% Needing a union bound, each sample must be collected by running the chain $T$ steps  where 
% $T$ is  a known upper bound on 
% $
% \TMix(\frac{\varepsilon}{2\frange})$. The total sample complexity of such \emph{iterative methods} are thus $\TMixBound \cdot m(f,\varepsilon,\delta) \log(\frac{\frange}{\varepsilon})$. 
More efficient estimators compute the average  over the entire trace of the Markov chain,  and their complexity depends on a \emph{known upper-bound} $\TRelBound$ on the \emph{relaxation time} $\TRel$.
Some of these bounds are \emph{variance agnostic} \cite{LezaudChernoff,miachernof,leonhoeffdingbound2004,fan2018hoeffding}, and they imply MCMC-mean estimators with complexity $
m_{\rm Hoff}(\M,f, \varepsilon, \delta)
    %=\frac{1+\lambda}{1-\lambda} \ln(\mathsmaller{\frac{2}{\delta}}) \frac{\frange^{2}}{2\varepsilon^2}
    \in 
  \Theta\Bigl(\TRelBound\ln(\mathsmaller{\frac{1}{\delta}})\frac{\frange^{2}}{\varepsilon^{2}}\Bigr)
$. 
These bounds, also  known as Hoeffding-type bounds, are used  ubiquitously  due to their simplicity and convenience. 
Unfortunately, they are generally much looser than variance-aware bounds~\cite{bernestein2018new,BernMCMC2020,gao2014bernstein,mnih2008empirical}.
% Much ink has been spilled improving these bounds, to prove MCMC concentration bound by \emph{variance proxies} \cite{audibert2007tuning,mnih2008empirical,audibert2009exploration,maurer2009empirical,paulin2015}. \shahrzad{I think the correct citations should be: \cite{bernestein2018new,BernMCMC2020,gao2014bernstein,mnih2008empirical}}
For example, 
having a known upper bound $V_\pi \geq v_\pi$, Bernstein type bounds imply MCMC mean estimators with complexity $m_{\rm Bern}(\M,f,\varepsilon,\delta)\in\Theta\Bigl(\TRelBound\ln(\mathsmaller{\frac{1}{\delta}})\Bigl(\nicefrac{\frange}{\varepsilon} + \nicefrac{V_\pi}{\varepsilon^{2}}\Bigr)\Bigr)$. 
%\cyrus{join 2 previous sentences}

In all of the above classic bounds, dependencies on function-specific terms are obtained separately from chain-specific terms.
In contrast,  Rabinovich et al.\ \cite{functionMixingRabinovish} analyze the sample complexity of MCMC-mean estimation using what they term the \emph{function-specific mixing time}, written %which we denote by 
$\tau_{\rm fmix}(f,{\cal M},\varepsilon)$ 
%\shahrzad{$,\epsilon$?}
%(we sometimes drop the parenthesis for simplicity).
(parentheses omitted when clear from context), which always obeys $\TfMix(f, \M, \varepsilon) \leq \TMix(\M, \varepsilon)$.
They prove a Hoeffding like bound showing MCMC mean estimation can be done using $\tilde{O}(\TfMixBound\nicefrac{\frange^{2}}{\varepsilon^{2}} \log \frac{1}{\delta})$ 
%\cyrus{tilde obscures mixing time ``true cost''}
samples where $\TfMixBound \geq \TfMix$ is an a priori  known upper-bound. % on the function-specific mixing time.
%\cyrus{Cut this: This upper bound can be obtained by a full spectral decomposition of the chain's transition matrix and disregarding eigenvectors which are orthogonal to  the function $f$, and always obeys $\TfMix(f, \M, \varepsilon) \leq \TMix(\M, \varepsilon)$. %it is always smaller than the classic mixing time.
Another important bound was derived by  Paulin \cite{paulin2015}, who uses the (chain-specific) \emph{asymptotic variance}, defined as $v_{\rm asy}(f,\M)\doteq \lim_{\tau \to \infty} \mathbb{E}[\frac{1}{\tau}\sum_{i=1}^{\tau}(f(X_i)-\mu)^2]$, which, we prove (see \cref{lemma:tvar-prop}), is smaller than  $2\TRel \SVar$, and he shows sample complexity of 
$m_{\rm Paulin}(\varepsilon,\delta)\in\Theta\Bigl(\bigl(\TRelBound\cdot\frac{\frange}{\varepsilon} + \nicefrac{V_{\rm asy}}{\varepsilon^{2}}\bigr)\ln(\mathsmaller{\frac{1}{\delta}})\Bigr)$.

While both Rabinovich et al's  \cite{functionMixingRabinovish} function-specific mixing time $\TfMix$ and Paulin's \cite{paulin2015} chain-specific asymptotic variance $\AVar$ beat  prior bounds \emph{in theory}, 
because of their more sophisticated definitions, obtaining tight upper bounds for them is harder than their classic counterparts e.g., $\TMix$, $\TRel$ or $\SVar$.
Indeed, a major flaw of all above bounds is their strict dependence on known upper bounds, thus their efficiency can highly be deteriorated by the looseness of said bounds.

\smallskip 
In order to remove  dependencies on known variance upper-bounds, \emph{MCMC-variance estimators}, which estimate variance proxies by running the same chain as the mean estimator, have been used by practitioners \cite{varSong1995OptimalMB,varhobert2002,var2006Jones,varflegal2010batch,varSpecVariance2018,Vargupta2020estimating,Varchakraborty2019estimating}.
%\shahrzad{MERGE}In practice, they are often applied to finite samples, thus asymptotic analysis is not sufficient to prove their correctness. Due to their instability in practice, providing mathematical rigour to the uses of these estimators is ongoing 
Among them, the most popular are the \emph{spectral variance} (SV) \cite{varSpecVariance2018,varspectralbelomestny2020variance,varSpect2018lugsail} and \emph{batched means} (BM) \cite{varflegal2010batch,varflegal2010batch,var2006Jones,varhobert2002,varSong1995OptimalMB,Varchakraborty2019estimating} estimators.
The SV takes a weighted sum of covariances estimated from a single trace. The BM  estimates the asymptotic variance  by dividing the trace of the chain to batches, 
%of size $\tau$, 
and estimating the empirical variance of the mean of each batch.

Since most of theoretical analysis of these methods only concentrate on  asymptotic convergence, these estimators %are often error prone when using finite samples.
generally lack guarantees for finite samples which are necessary for rigorous bounds in applications.
Furthermore, all of these estimators are \emph{biased} (but \emph{consistent} or \emph{asymptotically unbiased}). To our knowledge no \emph{unbiased MCMC-variance estimator}  is known.

%, and many parameters in these estimators are picked without worrying about finite sample concentrations. e.g., batch size $\tau$ in BM is taken as $\tau \propto \sqrt[3]{m}$, where $m$ is the total sample size.
 %, and % picked to be the third root of the total sample size, 
 %mixing estimation is done using heuristic methods, etc.\  
 %And the 

%\cyrus{CLT: \cite{kipnis1986central}}

%We belive this is due to the fact that they all collect their samples from a single trace of the Markov chain.

\smallskip 

%Motivated by applications and to address the above shortcomings, 
To circumvent dependency on known bounds on mixing parameters, %a new direction in the study  of the Markov chain Monte Carlo has emerged recently which seeks to
some authors estimate the mixing or relaxation time of a Markov chain from a trace \cite{Hsu2015MixingTE,trelestHsuPeres2019,pmlr-v99-wolfer19a,pmlr-v117-wolfer20a,levin2016estimating,peresParam}. 
%Ben-Hamou et al \cite{peresParam} use \emph{random walks with restart} to estimate the mixing time of a random walk on a graph in $\tilde{O}\bigl(\sqrt{m}\ \tau_{\rm unif}^{3/4}\bigr)$ steps, where $m$ is the number of graph edges and $\tau_{\rm unif}$ is the uniform mixing time (another mixing notion). 
Hsu et al \cite{Hsu2015MixingTE,levin2016estimating,trelestHsuPeres2019} show how to estimate the relaxation time of a reversible  Markov chain using $\tilde{O}\bigl(\TRel/\pi_{\rm min} \bigr)$ steps, and show that $\tilde{\Omega}\bigl(\TRel\cdot \abs{\Domain}\bigr)$ is necessary.  Wolfer and Kontorovich \cite{pmlr-v117-wolfer20a} and  Wolfer \cite{pmlr-v99-wolfer19a}  obtain similar results while removing the reversibility condition.  These lower bounds eliminate the possibility of estimating $\TRel$ with \emph{no prior knowledge of it} in any number of Markov chain steps that is sub-linear in $\abs{\X}$, which is generally prohibitive. % for many problems. %size of the sample space (which is desirable for large $S$). \cyrus{notation}

\smallskip

Our work complements all these results. \algo{} does not require function dependent parameter as input. It uses a novel parameter $\ITVn{\cdot}(\M,f)$, \emph{the inter-trace variance} (see \cref{def:trv}), which is estimated efficiently from the data. \algo{} requires a chain dependent parameter $\TRelBound\geq \TRel$, an upper bound on the relaxation time, but the complexity of the algorithm is less dependent on the looseness of this bound compared to previously know algorithms.

\begin{table}[t]
{\footnotesize
\begin{tabular}{|c|c|c|}
\hline 
method& complexity, leading term  for small $\varepsilon$
%& leading term as $\varepsilon \rightarrow 0$
& input parameters\\
&is highlighted in blue &\\
\hline
This paper &${\displaystyle} %\log\left(\frac{\log \mathsmaller{\nicefrac{\frange}{\varepsilon}}}{\delta}\right) 
\tilde{O}\left( \frac{\TRelBound\frange}{\varepsilon} +{\color{blue} \TRel\ITRelVar}\cdot \frac{1}{\varepsilon^{2}} \right)$
%&$\tilde{O}\left(\frac{\AVar}{\varepsilon^{2}} \right)$
&$\TRelBound$\\
\hline 
Variance agnostic methods:&&\\
Classic bounds &$\tilde{O}\bigl({\color{blue}\TRelBound\frange^{2}}\cdot \frac{1}{\varepsilon^{2}}\bigr)
%\log \frac{1}{\delta}
$
%&$\tilde{O}\left(\frac{\TRelBound\frange^{2}}{\varepsilon^2}\right)$
&$\TRelBound$\\ 
Rabonich et al's bound \cite{functionMixingRabinovish}&$\tilde{O}\bigl({\color{blue}\TfMixBound\frange^{2}}\cdot \frac{1}{\varepsilon^{2}}\bigr)$&$\TfMixBound$\\
\hline 
Variance aware methods: &&\\ Classic bounds&$\tilde{O}\bigl(\TRelBound\frac{\frange}{\varepsilon} + {\color{blue}{V_{\pi}}\TRelBound}\cdot \frac{1}{\varepsilon^{2}}\bigr)$
%\log \frac{1}{\delta}$
%&$\tilde{O}\left(\frac{V_{\pi}}{\varepsilon^{2}} \right)$
&$\TRelBound$, $V_{\pi}$\\ 
Paulin's \cite{paulin2015} &$\tilde{O}\bigl(\frac{\TRelBound\frange}{\varepsilon} + {\color{blue} V_{\mathrm{asy}}}\cdot \frac{1}{\varepsilon^{2}}\bigr)$
%\log \frac{1}{\delta}$
%&$\tilde{O}\left(\frac{V_{\rm asy}}{\varepsilon^{2}} \right)$
&$\TRelBound$, $V_{\rm asy}$\\ 
\hline 
\end{tabular}
\caption{\small
Comparison to state of the art MCMC Mean Estimators.
Note that $\AVar=O(\TRel \SVar)$, $\tau_{\rm fmix}=O(\TMix)$. In terms of hardness of computation:  $\AVar\gg  \SVar\TRel$ ($\AVar$ requires asymptotic calculations of an \emph{infinite trace}), $\tau_{\rm fmix}\gg \TMix$ ($\tau_{\rm fmix}$ requires \emph{full} spectral decomposition) %while $\TRel$ or $\TMix$ can be bounded only by bounding the second largest eigen-value.) 
%Comparison with \cite{trelestHsuPeres} and \cite{functionMixingRabinovish}
}\label{table:2}
}
\end{table}

\subsection{Our Contributions}
\begin{enumerate}[topsep=0pt,partopsep=0pt,itemsep=0pt,wide,labelwidth=0pt, labelindent=12pt]
\item 
We introduce a novel statistical measure for MCMC-mean estimator, the \emph{inter-trace variance} $\ITRelVar(\M,f)$ (\cref{def:trv}). This tool yields improved estimator performance and analysis. We show that $\ITRelVar(\M,f)\leq \SVar(f)$, and that $\ITRelVar(\M,f)$ can be estimated efficiently from the data.

\item
%Since MCMC samples are not independence,we cannot employ Bessel's correction.  
We devise a new \emph{unbiased} MCMC variance estimator $\VED$ and show finite-sample bounds on
its error. We use this estimator, together with other concepts that we develop here, to rigorously bound $\ITVn{\TRelBound}$. (Section \ref{sec:varest})

\item Leveraging $\ITRelVar(\M,f)$, we design \algo{}, a dynamic MCMC-mean estimator that adapts to the observed data.
Given an upper bound on the relaxation time, $\TRelBound\geq \TRel$, and a bound $R$ on the range of $f$, we prove that the complexity (measured in Markov chain steps) of 
\algo{} is $\tilde{\mathcal{O}} \bigl( \frac{\TRelBound\frange }{\varepsilon} + \frac{\TRel(\M)\cdot \ITRelVar(\M,f)}{\varepsilon^{2}} \bigr)$,
without %Note that \algo{} does not need 
\emph{a priori} knowledge of $\TRel$ or $\ITRelVar$. (\Cref{sec:algo})

\item
For small $\varepsilon$ (high-precision case), the complexity of \algo{} is dominated by the term
$\frac{\TRel(\M)\cdot \ITRelVar(\M,f)}{\varepsilon^{2}}$, thus, in the high precision regime, the complexity of our algorithm is less dependent on a (possibly loose) bound $\TRelBound$ compared to previously known algorithms. %(\Cref{sec:algo}).

\item 
We show that $\TRel(\M) \ITRelVar(\M,f)$ (but not $\TRel(\M) \SVar(f)$) leads to bounds  competitive with the lower bounds obtained from the central limit theorem of Markov chains \cite{gordin1978central}. (Section \ref{sec:compare})

\item We show that when $f$'s image is distributed symmetrically on ${\cal M}$'s traces, we have ${\rm trv}^{({\tau{\rm rel}})}({\cal M},f)=o(v_{\pi}(f))$. In those cases  \textsc{DynaMITE} outperforms prior methods,
 even when tight bounds on variance and mixing or relaxation times are known. We demonstrate this improvement in two applications: mean estimate on a cycle, and counting $k$-colorings, (\Cref{sec:tracevar}, and \Cref{sec:5}).

\item 
%For th\e latter we show by using 
We use
\algo{}  as a \emph{mean-estimation gadget} in the
FPRAS 
%\emph{counting-sampling} reduction 
of Jerrum, Valiant and Vazirani \cite{countingsamplingreduction1986jurrumvaziranivalient}, %.
%In particular, we prove this improves 
 and prove the sample complexity of \emph{counting proper colorings} in \emph{planted partitions} of $r$ communities and $\abs{E}=\cal E$ edges with our method
 is $\tilde{\mathcal O}\bigl(\TRel\cdot \smash{\nicefrac{{\cal E}^{3}}{r}} \bigr)$, %as compared to
 improving the bound $\tilde{\mathcal{O}}\bigl(\TRelBound \cdot {\cal E}^3 \bigr)$ of \cite{jurrumSamplingiscounting}, while assuming no knowledge of the planted-partition structure of the graph.

 We chose the JVV FPRAS for simplicity; our methods apply, mutatis mutandis, to other MCMC-based telescoping mean algorithms, e.g., \cite{gibbsvigoda,hubergibbs,Kolmogorovgibbs,harris2020parameter}. (\Cref{sec:5})
 
 \item Our work bridges theory and practice by developing solid mathematical foundations for a combination of heuristics, e.g., batching, MCMC variance estimators and mixing diagnosis. These techniques are used by practitioners without proven convergence criteria, giving unreliable results.

\end{enumerate}
%\vspace{-0.4 cm}
\section{The Inter-Trace Variance}\label{sec:tracevar}
In this section we define \emph{the inter-trace variance} (\cref{def:trv}) and  prove its properties. 
In order to estimate the inter trace variance with no prior knowledge of it, we  introduce  an \emph{unbiased}  MCMC variance estimator (\cref{sec:varest}).
The concept of the \emph{inter-trace variance} is similar to \emph{batched mean variance estimators} used by practitioners. Here we develop new mathematical concepts and prove rigorous finite-sample guarantees. %aid us to prove rigorous theorem and obtain  finite sample guarantees.
\cyrus{Reversibility question: are $\M^{(T)}$ eigenvalues real?}

In its general form, the Inter-Trace variance, written $\ITVn{\tau}$,  is defined with respect to a \emph{length} parameter $\tau$. In our analysis  we use $\ITRelVar$,  calling it the \emph{relaxed} Inter-Trace variance.

\smallskip
%Consider  function $f:\X\rightarrow \mathbb{R}$, and $\M$ a Markov chain on state space $\X$.
For a parameter $\tau$, consider 
the set of all $\tau$-traces, and  the  probability distribution $\pi^{(\tau)}$  on this set defined as $\smash{\pi^{(\tau)}(\vec{X}_{1:\tau})\doteq \pi(X_1)\smash{\prod_{i=1}^{\tau-1}}\M(X_i,X_{i+1})} $, we call a $\tau$-trace sampled from $\pi^{(\tau)}$ a \emph{stationary trace}.
For an arbitrary  $\F$ defined on $\X^{\tau}$, by $\Expect_{\vec{X}\distributed \pi^{({\mathlarger{\tau}})}}[\F(\smash{\vec{X}})]$, we mean the expectation of $\F$ on any $\tau$-trace of the Markov chain when $X_1 \distributed \pi$, and other consecutive $X_i$s follow the transition of the %Markov 
chain.
Having   $\vec{X}_{1:\tau}=(X_1,X_2,\dots , X_\tau)$,
%and an upper bound on $\M$'s mixing time $\TMixBound$, 
we define $\Favg^{(\tau)}$ 
%for any $\vec{X}$'s contiguous  $\tau$-subtrace  ${\cal X}=(X_{j+1},X_{j+2},\dots, X_{j+\tau})$
%(a contiguous block of $\vec{X}$) 
as follows: $\Favg^{(\tau)}(\vec{X}_{1:\tau}) \doteq (\frac{1}{\tau})\sum_{i=1}^{\tau}f(X_{i})$. 
%Note that $\Expect[\Favg^{(\tau)}]=\Expect[f]$. %Furthermore each ${\cal X}_i$ is one step of the \emph{trace chain}  $\M^{(\TMixBound)}$ , i.e., we have a $k$-trace of $\M^{(\TMixBound)}$. 
%and  the variance of $\Favg^{(\tau)}$ over $\M^{(\tau)}$ traces  is $\ITVn{\tau}$. 

\begin{defin}[\textbf{Inter-Trace Variance}]\label{def:trv}
Consider $f$, $\M$ and $\pi$ as before. 
%Consider  function $f:\X\rightarrow \mathbb{R}$, and $\M$ a Markov chain on state space $\X$ with stationary distribution $\pi$.
%For a parameter $\tau$, consider the set of all $\tau$-traces, and  the  probability distribution $\pi^{(\tau)}$  as $\smash{\pi^{(\tau)}(\vec{X}_{1:\tau})\doteq \pi(X_1)\smash{\prod_{i=1}^{\tau-1}}\M(X_i,X_{i+1})} $.For an arbitrary  $\F$ defined on $\Omega^{\tau}$, by $\Expect_{\vec{X}\distributed \pi^{({\mathlarger{\tau}})}}[F(\smash{\vec{X}})]$, we mean the expectation of $\F$ on any $\tau$-trace of the Markov chain when $X_1 \distributed \pi$, and other consecutive $X_i$s follow the transition of the Markov chain. 
For arbitrary $\tau$ and a $\tau$-trace of $\M$, $\vec{X}_{1:\tau}$, let  $\Favg^{(\tau)}(\vec{X}_{1:\tau})$ be as defined above. 
We denote the \emph{the inter-trace variance} of $f$  by $\ITVn{\tau}(\M,f)$, and define it as,
$$\ITVn{\tau}(\M,f)\doteq\mathbb{V}_{\vec{X}\distributed \pi^{(\mathlarger{\tau})}}[\Favg^{(\tau)}(\smash{\vec{X}})]~.$$
Henceforward, we simply use $\ITVn{\tau}$ when removing $\M$ and $f$ does not create ambiguity. 
\end{defin}

{\bf Remark.}  Note that by linearity of expectation,  $\Expect_{\smash{\vec{X}}\distributed\pi^{(\mathlarger{\tau})}}[\Favg(\smash{\vec{X})}] = \Expect_{X\distributed\pi}[f(X)] = \mu$, thus, $\ITVn{\tau}(\M,f)=\Expect_{\vec{X}\distributed \pi^{(\mathlarger{\tau})}}[(\Favg(\smash{\vec{X}})-\mu)^2]$.

The following lemma (proved in \cref{proofs:trace}),  shows that $\tau \ITVn{\tau}$ is %nondecreasing, but 
bounded above by  $2\TRel\SVar$. Which implies always equal or better sample complexity than classic variance aware bounds. 
\begin{restatable}{lemma}{lemmatvarprop}
\label{lemma:tvar-prop}
Suppose a lazy reversible Markov chain $\M$.
The inter-trace variance obeys % is nonincreasing as
\begin{equation}
\label{eq:tvdec}
\SVar \geq \ITVn{2} \geq \ITVn{3} \geq \cdots  \enspace,
\end{equation}
and $\tau \ITVn{\tau}$ is nondecreasing and bounded as
\begin{equation}
\label{eq:tvtinc}
%\SVar \leq 2v_{2} \leq 3v_{3} \leq \dots \leq \AVar \leq \frac{1 + \EigTwo(f)}{1 - \EigTwo(f)}\SVar \leq \frac{1 + \EigTwo}{1 - \EigTwo}\SVar \enspace.
\SVar \leq 2\ITVn{2} \leq 3\ITVn{3}\leq \dots  \leq \lim_{i\to \infty} i\ITVn{i}\leq (2\TRel - 1)\SVar \enspace.
\end{equation}

Furthermore, there exists some absolute constant $c > 0$, such that for any $\M$, $f$, we have
\begin{equation}
\label{eq:avar-lb}
\tau \ITVn{\tau} \geq c \tau' \ITVn{\tau'}\geq c \lim_{i \to \infty} i\ITVn{i} \enspace, \quad \text{ for any $\tau' \geq \tau \geq \TRel(\M)$\enspace.}
\end{equation}
%Furthermore for $\tau\geq \TRel(\M)$, we have
%\begin{equation}\label{eq:2}\Theta(\TRel \ITRelVar) = \Theta(\tau\ITVn{\tau}) = \Theta\left( \lim_{\tau \to \infty} \tau\ITVn{\tau} \right) \enspace.\end{equation}
\end{restatable}

\noindent
{\bf Remark. } 
 Paulin's asymptotic variance can be expressed using the inter-trace variance as  $\AVar=\lim_{i\rightarrow \infty} i\ITVn{i}$, moreover 
 he proves  
$\ITVn{\tau}\leq (\nicefrac{2\TRel}{\tau})\SVar$ (see Thm.~3.1 of \cite{paulin2015}). Note that \cref{eq:tvtinc} improves and extends this result. Furthermore,
 in the appnedix we 
show \cref{coro:fmix}  which bounds the trace variance using, $\TfRel$, the function specific relaxation time of Rabinovich et al \cite{functionMixingRabinovish}. 
%In particular we prove \cref{eq:tvtinc} and \cref{eq:avar-lb} hold when replacing  $\TRel(\M)$ by $\tau_{\rm frel}(\M,f)$, where $\tau_{\rm frel}$ is Rabinovich et al's term \cite{functionMixingRabinovish} 
%(see 
%\cref{def:trefmix}, ).

\smallskip

An important consequence of \cref{lemma:tvar-prop} is that 
 for $T \geq \TRel$, we have
\begin{equation}
\label{eq:itvar-trel-decay}
\ITVn{T} \in \Theta\left( \frac{\TRel}{T} \ITRelVar \right) \enspace,
\end{equation}
which should remind the reader of $\Var[\frac{1}{T}\sum_{i=1}^T f(X_i)]=\frac{\SVar}{T}$
which holds when $X_i$s are sampled independently from $\pi$.
In fact, $\TRel\ITRelVar$ captures the behavior of  averages of Markovian random variables just as $\SVar$ does for independent random variables.
Furthermore the asymptotic term $ \lim_{i\to \infty} i \ITVn{i}$  appears in the central limit theorem (see \cref{sec:compare}), where this average is Gaussian, and in both the finite and infinite cases, such variances are intimately tied to sample complexity.
Thus, we express the sample complexity of \algo{} with respect to $\TRel\ITRelVar$.

\vspace*{-0.2 cm}
\paragraph{The Trace Chain}
The inter-trace variance can be thought of as the variance of $\Favg^{(\tau)}$ over stationary traces of length $\tau$, i.e., $\pi^{(\tau)}$.
Letting $S^{(\tau)}$ be the space of all $\tau$-traces, we now define the $\tau$-\emph{trace chain}, which naturally groups the output of $\M$ into $\tau$-traces. %, and who's stationary distribution is $\pi^{(\tau)}$ (see~\cref{lem:cycle}).
%Defining the trace chain, we formalize this idea.
%drawn from the Markov chain $\M$.
%Alternatively, we can think of the inter-trace variance as the stationary variance of a Markov chain $\M^{(T)}$, which has state space $\X^{T}$ (i.e., length-$\tau$) traces over $\X$, where the stationary distribution matches that of length $\tau$ traces from $\M$.
 %length Tower traces

%An alternative to define the inter-trace variance which is instructive for algorithm design purposes is through the \emph{trace chain} defined as follows:

\begin{defin}[{\bf Trace chain}]\label{def:2.1}
For a Markov chain $\M$ on state space $\X$,
we define the trace chain $\M^{(\tau)}$ on state space $\X^{(\tau)}$ as follows:
given $\vec{a}=(a_1, a_2 \dots, a_\tau)$ and $\vec{b}=(b_1,b_2,\dots , b_\tau)$ in $\X^{(\tau)}$, the probability of going from $\vec{a}$ to $\vec{b}$ is
$\M^{(\tau)}(\vec{a},\vec{b})\doteq\M(a_\tau,b_1)\prod_{i=1}^{\tau-1}\M(b_i,b_{i+1})$.
\end{defin}
\arx{The definition of the trace chain in terms of transition probabilities is perhaps slightly %counterintuitive 
unintuitive, but we} Note that an $m$-trace $\vec{X}_{1:m}$ drawn from $\M$ is equivalently distributed to a $k$-trace $\vec{\cal X}_{1:k}=({\cal X}_1,
%{\cal X}_2,
\dots, {\cal X}_k)$  drawn from $\M^{(\tau)}$, where  $k=m/\tau$ and each ${\cal X}_i$ is a \emph{contiguous disjoint sub-trace} of $\vec{X}$, i.e., ${\cal X}_1=({X}_{1},\dots , X_{\tau}),\ {\cal X}_2=({X}_{\tau+1},\dots,  X_{2\tau}),\dots  , {\cal X}_k=({X}_{(k-1)\tau+1}\dots X_m) $.
We show in \cref{lemma:tracemixing} that $\pi^{(\tau)}$ is the stationary distribution of $\M^{(\tau)}$, i.e., 
%$\Expect_{\pi^{(\tau)}}[\Favg^{(\tau)}] = \Expect_{\pi}[f]$ and  
$\ITVn{\tau}(f,\M)$  is the  variance of $\Favg^{\tau}$ on stationary distribution of  $\M^{(\tau)}$.
\paragraph{When is ${\bf trv}^{\TRel}=\mathbf{o(\SVar)}$?} Plugging in $\tau=\TRel$ in Equation \ref{eq:tvtinc} we obtain $\TRel\ITRelVar\leq 2\TRel\SVar$,  i.e.,   $\ITRelVar=O(\SVar)$.
%, thus  inter-trace variance  in worst case results in equivalent sample complexity as stationary variance. 
Using projection chains (see \cref{def:proj}), we can show that the inter-trace variance becomes smaller when $f$ is symmetrically projected on traces of $\M$.
\Cref{exm:cycle}  illustrates this observation.
 In section \ref{sec:5} we present an application to a counting problem and  show \emph{approximate-symmetry} suffices. 
%the \emph{equivalence relationship} $a \simeq b$ iff $f(a) = f(b)$
%$f$ partitions the state space of ${\M}$ in such a way that the \emph{projection chain} (defined below in \cref{def:proj}) has a much smaller relaxation time ($\ITRelVar=o(\SVar)$). 
%In this scenario, even for tight bounds on $\TRel$, i.e., $T=\TRel$, we have $\ITRelVar=o(\SVar)$.
\begin{defin}[Projection chain \cite{levin2017markov}]
\label{def:proj}
Having an equivalence relationship $\simeq$ on $\X$ and classes $\tilde{\X}=\{[x]; x\in \X\}$ such that if $x\simeq x'$, then ${\M}(x,[y])={\M}(x',[y])$,
we call the Markov chain $\tilde{\M}$ with state space $\tilde{\X}$ and transition probabilities $\tilde{\M}([x],[y])={\M}(x,[y]) $   a \emph{projection chain} (see section 2.3 of \cite{levin2017markov} for full discussion).
\end{defin}\label{defin:proj}
We now introduce a class of functions having equal \emph{means} and \emph{stationary variances}, but projecting differently on traces of a fixed chain. While they all have equal $\SVar$, in \cref{lem:cycle} (proof in \cref{sec:cycle}), we prove that under appropriate parameterization, $\ITRelVar$  takes any arbitrary value.

\sidecaptionvpos{caption}{t}

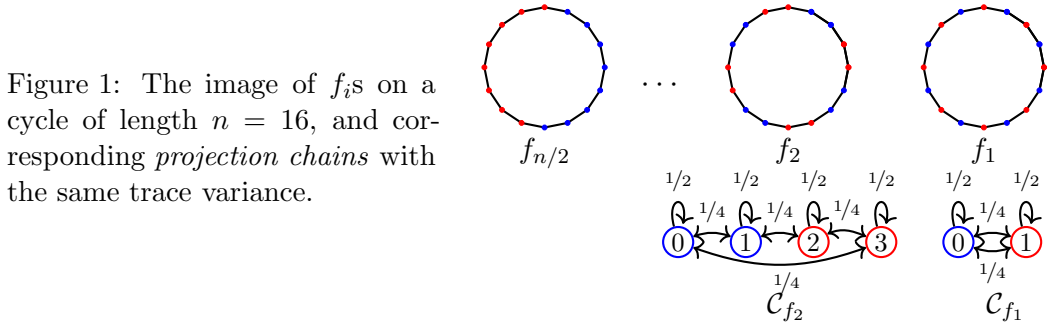
\begin{SCfigure}[0.6][hb]
%\floatbox[{\capbeside\thisfloatsetup{capbesideposition={right},capbesidewidth=16cm}}]{figure}[\FBwidth]{
\caption{%\footnotesize
The image of $f_i$s on a cycle of length $n = 16$, and corresponding \emph{projection chains} with the same trace variance.}\label{fig:cycle}
%}

\ \begin{minipage}{0.56\textwidth}

\def\circsize{0.04cm}

\begin{tikzpicture}[scale=0.8]

\begin{scope}[xshift=0.45cm]

\foreach \i[evaluate=\i as \j using {\i + 1}] in {1,2,...,16} {
 \draw[thick] ({sin(\i / 16 * 360)},{cos(\i / 16 * 360)}) --  ({sin(\j / 16 * 360)},{cos(\j / 16 * 360)});
}

\node[anchor=north] (label) at (0, -1) {$f_{n/2}$};

\foreach \i in {1,2,...,8} {
\filldraw[blue]  ({sin(\i / 16 * 360)},{cos(\i / 16 * 360)}) circle (\circsize);
}
\foreach \i in {9,10,...,16} {

\filldraw[red]  ({sin(\i / 16 * 360)},{cos(\i / 16 * 360)}) circle (\circsize);
}
\end{scope}

\begin{scope}[xshift=0.4cm]
\node[anchor=north] (label) at (2.0, 0) {\large$\cdots$};
\end{scope}

\begin{scope}[xshift=4.5cm]

\foreach \i[evaluate=\i as \j using {\i + 1}] in {1,2,...,20} {
 \draw[thick]  ({sin(\i / 16 * 360)},{cos(\i / 16 * 360)}) --  ({sin(\j / 16 * 360)},{cos(\j / 16 * 360)});
}

\node[anchor=north] (label) at (0, -1) {$f_2$};
\foreach \i in {1,5,...,20} {
\filldraw[blue]  ({sin(\i / 16 * 360)},{cos(\i / 16 * 360)}) circle (\circsize);
}

\foreach \i in {3,7,...,20} {
\filldraw[red]  ({sin(\i / 16 * 360)},{cos(\i / 16 * 360)}) circle (\circsize);
}
\foreach \i in {2,6,...,20} {
\filldraw[blue]  ({sin(\i / 16 * 360)},{cos(\i / 16 * 360)}) circle (\circsize);
}
\foreach \i in {4,8,...,20} {
\filldraw[red]  ({sin(\i / 16 * 360)},{cos(\i / 16 * 360)}) circle (\circsize);
}
\end{scope}
\begin{scope}[xshift=7.75 cm]

\foreach \i[evaluate=\i as \j using {\i + 1}] in {1,2,...,20} {
\draw[thick]  ({sin(\i / 16 * 360)},{cos(\i / 16 * 360)}) --  ({sin(\j / 16 * 360)},{cos(\j / 16 * 360)});

}

\node[anchor=north] (label) at (0, -1) {$f_1$};

\foreach \i in {1,3,...,20} {

\filldraw[blue]  ({sin(\i / 16 * 360)},{cos(\i / 16 * 360)}) circle (\circsize);
}
\foreach \i in {2,4,...,20} {

\filldraw[red]  ({sin(\i / 16 * 360)},{cos(\i / 16 * 360)}) circle (\circsize);}

\end{scope}
\end{tikzpicture}

\hspace{-0.5cm}\begin{tikzpicture}[
    scale=0.9,
    bnode/.style={thick,circle,draw=blue,minimum size=0.4cm,inner sep=0pt},
    rnode/.style={thick,circle,draw=red,minimum size=0.4cm,inner sep=0pt},
    gedge/.style={thick,<->},%, -{Stealth[length=3mm, open]}},
    align=right,
]

\draw[draw=none,use as bounding box] (0, -0.5) rectangle (1, 0.9); 

\begin{scope}[xshift=3.46cm]

\node[bnode] (a) at (0, 0) {\small$0$};
\node[bnode] (b) at (1, 0) {\small$1$};
\node[rnode] (c) at (2, 0) {\small$2$};
\node[rnode] (d) at (3, 0) {\small$3$};
%Self-Edge
\draw[gedge] (a) edge[loop above]node{$\mathsmaller{\nicefrac{1}{2}}$} (a);
\draw[gedge] (b) edge[loop above]node{$\mathsmaller{\nicefrac{1}{2}}$} (b);
\draw[gedge] (c) edge[loop above]node{$\mathsmaller{\nicefrac{1}{2}}$} (c);
\draw[gedge] (d) edge[loop above]node{$\mathsmaller{\nicefrac{1}{2}}$} (d);

\draw[gedge] (a) edge[bend left=10,out=20,in=160]node[above]{$\mathsmaller{\nicefrac{1}{4}}$} (b);

\draw[gedge] (b) edge[bend left=10,out=20,in=160]node[above]{$\mathsmaller{\nicefrac{1}{4}}$} (c);

\draw[gedge] (c) edge[bend left=10,out=30,in=160]node[above]{$\mathsmaller{\nicefrac{1}{4}}$} (d);

\draw[gedge] (d) edge[bend right=10,out=20,in=160] (a);
%\draw (a) -- (b);
\node[anchor=north] (label) at (1.6, -0.3) {$\mathsmaller{\nicefrac{1}{4}}$};
\node[anchor=north] (label) at (1.6, -0.6) { ${\cal C}_{f_2}$};
\end{scope}

\begin{scope}[xshift=2.5cm]
%\hspace{2.2 cm}

\node[bnode] (a) at (5.1, 0) {\small$0$};
\node[rnode] (b) at (6.1, 0) {\small$1$};

%Self-Edge
\draw[gedge] (a) edge[loop above]node{$\mathsmaller{\nicefrac{1}{2}}$} (a);
\draw[gedge] (b) edge[loop above]node{$\mathsmaller{\nicefrac{1}{2}}$} (b);
%\draw[gedge] (a) edge[bend left=15] (b);
\draw[gedge] (a) edge[bend left=10,out=20,in=160]node[above]{$\mathsmaller{\nicefrac{1}{4}}$} (b);
\draw[gedge] (b) edge[bend right=10,out=20,in=160]node[anchor=north]{$\mathsmaller{\nicefrac{1}{4}}$} (a);

\node[anchor=north] (label) at (5.8, -0.6) { ${\cal C}_{f_1}$};
%\draw (a) -- (b);
\end{scope}
\end{tikzpicture}
\end{minipage}
\end{SCfigure}

\begin{exm}\label{exm:cycle}
 Consider the Markov chain $\cal C$, known as the {cycle}, defined on $[n]=\{1,2,\dots, n\}$ with transition probabilities: ${\cal C}(i,i)=1/2$, ${\cal C}(i,i+1)=1/4$ and ${\cal C}(i,i-1)=1/4$, where $i-1$ and $i+1$ are taken ${\rm mod}~ n$. Clearly the stationary distribution on this chain is uniform, and it is  known that the mixing time and relaxation time are both $\Theta(n^2)$. 
 The following class of functions defined on $[n]$ all satisfy  $\mathbb{E}_{\pi}(f_i)=1/2$ and $\SVar(f_i)=1/4$. However, as shown in lemma \ref{lem:cycle},
 as the image of $f_i$s distribute more evenly  on $\M$'s traces  $\ITRelVar(f_i,\M)$ becomes smaller.

For $1\leq i\leq \frac{n}{2} $, let $f_i:[n]\rightarrow \{0,1\}$ be $f_i(x)=0$ if and only if $x ~ {\rm mod}~ 2i< i$, so $f_i$'s image on the cycle is consecutive length-$i$ runs of $0$s and $1$s (see \cref{fig:cycle}).
\arx{Note that in the two extreme cases, we have, (1) $f_1:[n]\rightarrow \{0,1\}, f_1(x)=0$ if and only if $x ~ {\rm mod}~ 2=0$, and (2) $f_{n/2}:[n]\rightarrow \{0,1\}, f_{n/2}(x)=0$ if and only if $x\leq \frac{n}{2}$.}

On equivalence classes of ${\rm mod}~ 2$, we find the corresponding projection chain, denoted by ${\cal C}_{f_1}$,  which has transition probabilities:
${\cal C}_{f_1}(0,1)={\cal C}_{f_1}(1,0)=1/4$ and ${\cal C}_{f_1}(0,0)={\cal C}_{f_1}(1,1)=1/2$. Similarly, for arbitrary $i$, we denote the projection chain on equivalence classes of ${\rm mod}~2i$  by ${\cal C}_{f_i}$.
Clearly for each $i$, the trace variance of $f_i$ on $\cal C$ and on ${\cal C}_{f_i}$ are distributed identically
(see \cref{fig:cycle}). 
\end{exm}
%We now present  the following lemma whose complete proof is presented in \cref{sec:cycle}. 
\begin{restatable}{lemma}{cycle}\label{lem:cycle}
For $1\leq i\leq \frac{n}{2}$, let $f_i$ be defined as above, 
%for arbitrary $\tau$ we have
% $\ITVn{\tau}(f_i)=\Theta(\nicefrac{i^2}{\tau})$. In particular, for $\tau = \Theta(\TRel(\M)) = \Theta(n^2)$, 
we have $\ITRelVar(\Cycle,f_i)=\Theta(\nicefrac{i^2}{n^2})$ e.g.,  $\ITRelVar(\Cycle,f_{n/2}) = \Theta(1)$, and $\ITRelVar(\Cycle,f_{1})= \Theta(\nicefrac{1}{n^2})$.
\end{restatable}

The reader may rightfully remark that the above scenario, in which the image of a function  cleanly partitions the state space is rare. In Section \ref{sec:5}, we develop new concepts to identify scenarios for more complex chains, %such as Glauber dynamics  sampling uniformly proper colorings of a given graph, 
demonstrating a case of functions whose image is \emph{approximately symmetric} on traces of the chain, showing the above projection trick works with some modifications.

\subsection{An unbiased MCMC-variances estimator: A tale of two chains }\label{sec:varest}

 %\cite{varSong1995OptimalMB,varhobert2002,var2006Jones,varflegal2010batch,varSpecVariance2018,Varchakraborty2019estimating,Vargupta2020estimating}. 

%These methods can be error prune because, assuming independence between batches. These results are ambiguous in the size of the batches and often rely on non-rigorous methods to this end.

%, or an estimator for an upper bound of $\Var[\F]$.

In the following we introduce an \emph{unbiased} MCMC estimator for the variance, and prove its concentration guarantees for \emph{finite samples}.  To our knowledge, all existing  MCMC variance estimators are   \emph{biased}  but \emph{consistent} and \emph{asymptotically unbiased}, and finite sample guarantees are not provided. Having an arbitrary function $\F:\X\rightarrow \mathbb{R}$ and  a Markov chain $\M$ defined on state space $\X$ and converging to stationary distribution $\pi$, we are interested to estimate the variance $\mathbb{V}_{X\sim\pi}[\F(X)]$, using  observations drawn from $\M$.
%we only use a Bernstein bound, which results in a smaller computation cost. 

\begin{defin}[{\bf Two chain variance estimator}]\label{def:varest}
%Running  $\M$ two times to get,
Suppose $\vec{X_1}=(X_{11},X_{12},\dots , X_{1m})$ and 
$\vec{X_2}=(X_{21},X_{22},\dots , X_{2m})$ are two independent $m$-traces drawn from $\M$ at stationarity.
The estimator is defined as
\[
\VED(\F,\M)\doteq \frac{1}{2m}\sum_{i=1}^m \left(\F(X_{1,i})-\F(X_{2,i})\right )^2 \enspace,
\]
where we drop the parentheses when $\M$ and $\F$ are clear from the context. 
\iffalse 
\cyrus{To operate on this paired (tensor-product) chain, we introduce the paired summation and paired subtraction operators, where
\[
(f \oplus g)(x_{1}, x_{2}) = f(x_{1}) + g(x_{2}) \enspace, \quad \text{and} \ (f \ominus g)(x_{1}, x_{2}) = f(x_{1}) - g(x_{2}) \enspace.
\]
We then take
\[
\VED(\F,\M)\doteq \frac{1}{2m}\sum_{i=1}^m (\F \ominus \F)^2  (X_{1,i}, X_{2,i}) = \frac{1}{2m}\sum_{i=1}^m \left(\F(X_{1,i})-\F(X_{2,i})\right )^2 \enspace.
\]
}
\fi 

\end{defin}
The following lemma shows finite sample  concentration of $\VED$ and it
%, which implies that $\VED$ is an unbiased and consistent estimator of $\Var[\F]$, 
is proved in \cref{sec:varestimatorproof}.

\begin{lemma}[A \emph{Tail} of Two Chains]\label{lemma:twochain}  Let  
 $\VED$ be as defined above. It holds that  $\VED$ is an \emph{unbiased estimator} for variance, i.e.,  $\Expect[\VED(\F,\M)]=\Var[\F]$. Furthermore, letting $\frange$ be range of $\F$ we have
%
\if 0
\begin{align*}
\Expect[\VED] = \Var[\F]
\end{align*}
\fi
%
\begin{align*}
\Prob \left(\abs{\Var[\F] - \VED} > \varepsilon \right)
  \leq \delta \ \ \ \ \text{ for  } &\ \  \varepsilon \in \LandauTheta\left( \frac{\TRel \frange^{2} \ln \frac{1}{\delta}}{m} + \frange \sqrt{\frac{\TRel\VED\ln \frac{1}{\delta}}{m}} \right), \\
&  \implies m_{\hat{v}}(\M,f,\delta,\varepsilon) \in \Theta \left( \left( \frac{\frange^{2}}{\varepsilon} + \frac{\Var[\F] \frange^{2}}{\varepsilon^{2}} \right) \TRel \log \frac{1}{\delta} \right)
  \enspace.
\end{align*}
\end{lemma}

\vspace*{-0.2 cm}
\paragraph{Inter-trace variance estimation.}
In order to avoid assuming  prior knowledge of variance or inter-trace variance in \algo{}, we employ the above estimator to the \emph{the trace chain}, $\M^{(\TRelBound)}$. 
  In \cref{lemma:mt-relax} we show $\M^{(T)}$ has constant relaxation time when $T\geq \TRel$ (e.g., $ T=\TRelBound$).
Thus, we can estimate $\ITVn{\TRelBound}$ by employing lemma \ref{lemma:twochain} and using $\VED(\Favg,\M^{(\TRelBound)})$.

\section{\algo{}}\label{sec:algo}

In this section we present the \algo{}: \textsc{DYNAmic Mcmc Inter-Trace variance Estimation} method. We show that its sample complexity is dependent on the \emph{apriori unknown} relaxed trace variance $\ITRelVar$. 
This section is 
 prelude to the themes and techniques used in \algo{}, which are fully developed in  \cref{sec:proofsalgo}. 
In \cref{sec:compare} we  compare our results with prior work.
 
\vspace*{-0.2 cm} 
\paragraph{A Prelude}\Cref{alg:1} shows a pseudocode of \algo{}.
There are two main techniques constituting  \algo{}: \emph{trace averaging} and \emph{progressive sampling}. 
To simplify the presentation, we separated the pseudocode to two subroutines: \DMCMC\  employs {progressive sampling} to dynamically estimate variance (with no prior knowledge of it), thus is itself a \emph{dynamic} MCMC mean estimator.
 \algo{} calls \DMCMC\ while using {trace averaging}, which improves sample complexity dependence from $\SVar$ to $\ITRelVar$.  
While application of progressive sampling is common  
in algorithm design \cite{progressivesamp}, here we introduce trace averaging.

\begin{defin}
[{\bf Trace averaging}]\label{def:trvavg} Consider function $f$,   a trace of length $m$ of $\M$ $\vec{X}_{1:m}=(X_1,X_2,\dots , X_m)$, and an upper bound on $\M$'s relaxation time $\TRelBound$. %let $\Favg$ for
For any $\vec{X}$'s contiguous  
$\TRelBound$-subtrace  ${\cal X}=(X_{j+1},X_{j+2},\dots, X_{j+\TRelBound})$,
%(a contiguous block of $\vec{X}$) 
let  $ \Favg^{(\TRelBound)}({\cal X})= (\frac{1}{\TRelBound})\sum_{i=1}^{\TRelBound}f(X_{j+i})$. Trace averaging is then the process of estimating $\Expect_{\pi}[f]$ %the mean of $f$ 
by employing an MCMC mean estimator for $\Favg^{(\TRelBound)}$ on the trace chain $\M^{(\TRelBound)}$.
%consider $\M^(\TMixBound)$ 
%, we estimate $\Expect[f]$ as follows:
%, we divide $\vec{X}_{1:m}$ to $k=m/\TMixBound$ contiguous disjoint blocks ${\cal X}_1=({X}_{1},\dots , X_{\TMixBound}),\ {\cal X}_2=({X}_{\TMixBound+1},\dots,  X_{2\TMixBound}),\dots  , {\cal X}_k=({X}_{(k-1)\TMixBound+1}\dots X_m) $, each of length $\TMixBound$,  and for each $1\leq i\leq k$, we define $\Favg({\cal X}_i) \doteq (\frac{1}{\TMixBound})\sum_{j=(i-1)\TMixBound+1}^{i\TMixBound}f(X_{j})$.
\end{defin}
%\shahrzad{Comment out}
%Note that $\Expect[\Favg^{(\TMixBound)}]=\Expect[f]$, %Furthermore each ${\cal X}_i$ is one step of the \emph{trace chain}  $\M^{(\TMixBound)}$ , i.e., we have a $k$-trace of $\M^{(\TMixBound)}$. 
%and  the variance of $\Favg^{(\TMixBound)}$ over $\M^{(\TMixBound)}$ traces  is $\ITVn{\TMixBound}$.  
 %In \cref{lemma:tracemixing} we show $\M^{(T)}$ has constant mixing and relaxation time when $T\geq \TMix(\M)$.
We also employ \emph{progressive sampling}. Beginning with a small 
%sample obtained by running 
number of  chain transitions,  we \emph{progressively} increase the sample size until a \emph{stopping condition} is met. 
%Our algorithm employs  \emph{progressive sampling}: we begin with a small sample from running the Markov chain. 

In each round,
we calculate the unbiased variance estimator as developed in \cref{sec:varest} (see lines 8--12 in \cref{alg:1}), using it we 
  find a high probability upper-bound on $\ITVn{\TRelBound}$ (line 13 of \cref{alg:1}).
From it and by employing an MCMC Bernstein bound e.g., \cref{thm:bernstein} we obtain  a suitable stopping condition guaranteeing that the empirical mean
is sufficiently accurate
%, then we 
%, which checks quality of mean estimation 
(line 14 and 15 of \cref{alg:1}).
\cyrus{emphasize: as though var were known a priori}
Using $\TRelBound\ITVn{\TRelBound}\leq c\TRel\ITRelVar$ (for universal constant $c$) as proved in \cref{lemma:tvar-prop}, we  show our bounds  in terms of the \emph{relaxed} inter-trace variance $\ITRelVar$.

The following theorems, proved in \cref{sec:proofsalgo}, guarantee correctness and efficiency of \algo{}. Since many of the lemmas use stationary traces, in line 4 of \cref{alg:1}, we use the standard warm start trick which is running the chain from arbitrary starting points %for $\tau_{\rm unif}$ 
%steps 
until  proximity to stationarity  is reached
%through a  warm start 
(see \cref{remark:stationarity}).

\vspace*{-0.1cm}

%\cyrus{Eli: second largest absolute eigenvalue?}\shahrzad{I deleted warm start procedure fix it in this theorem}

\begin{restatable}[Correctness of \algo{}]{theorem}{correctness}
\label{thm:correctness}
Consider a Markov chain $\M$ and
%over $\X$, 
 its relaxation time upper bound $\TRelBound$, function $f$ and its range $\frange$.
%be an upper bound on its relaxation time.
%Assume we have %\emph{trace size $T \in \N$, 
%function $f: \X \to [a, b]$ with $\frange \doteq b - a$, and 
%Assume we want to estimate \emph{true mean}, 
For arbitrary $\varepsilon$ and  $\delta $, and starting points $x_0,x_1\in S$ taking $\hat{\mu}$ as either of 
%Suppose also \emph{stationary samples}  $(x_{0}, x_{1}) \distributed \pi(\M \otimes \M)$, and 
%Take \emph{mean estimate} $\hat{\mu}$ 
%If we start at stationarity, i.e. $(x_{0}, x_{1}) \distributed \pi(\M \otimes \M)$, we take \emph{mean estimate} $\hat{\mu}$ %(arbitrarily)
%as either %of
\begin{enumerate}[topsep=4pt,partopsep=0pt,itemsep=0pt,wide,labelwidth=0pt, labelindent=12pt]
\item $\hat{\mu} \gets \DMCMC((x_{0}, x_{1}), \M, \TRelBound, f, {\varepsilon, \delta})$; or %\implies \abs{\hat{\mu} - \mu} \leq \varepsilon$;
\item $\hat{\mu} \gets \algo((x_{0}, x_{1}), \M, \TRelBound, f, {\varepsilon, \delta})$ % \implies \abs{\hat{\mu} - \mu} \leq \varepsilon$ 
for lazy $\M$.
\end{enumerate}

we will have a $(\varepsilon,\delta)$ estimator for $\mu \doteq \Expect_{\pi}[f]$, i.e, $\mathbb{P}(\vert\hat{\mu}-\mu \vert\geq \varepsilon )\leq \delta$
\end{restatable}
\cyrus{Nonstationary comment here?}

\vspace*{-0.2 cm}

\begin{restatable}[Efficiency of \algo{}]{theorem}{thmefficiency}
\label{thm:efficiency}
Suppose as in \cref{thm:correctness}.
%, and let $\NIterations = \log_{2}(\frac{\frange}{2\varepsilon})$,
%Suppose \emph{mixing time upper bound} $T \geq \TMix(1/4)$, \emph{true relaxation time} $\TRel$, and \dots.
%and take $\NIterations \doteq \log_{2}(\frac{9\frange}{100\varepsilon})$. 
%Assuming $T\geq \TRel({\M})$,
With probability at least $1 - \delta$, %TODO: Can tighten to $1 - \frac{\delta}{3\NIterations}$
it holds that total sample complexity  of
\DMCMC\ obeys
\vspace*{-0.2 cm}
\begin{align*}
m_{\DMCMC}(\M,f,\varepsilon,\delta) \in \mathcal{O}\left( \TRelBound \log\left(\frac{\log(\nicefrac{\frange}{\varepsilon})}{\delta}\right)\left( \frac{\frange}{\varepsilon} + \frac{\SVar}{\varepsilon^{2}}\right)\right) \enspace,
\end{align*}
\vspace*{-0.3cm}
and that of \algo{} obeys
\vspace*{-0.2 cm}
\begin{align*}
m_{\algo{}}(\M,f,\varepsilon,\delta) \in \mathcal{O}\left( \log\left(\frac{\log(\nicefrac{\frange}{\varepsilon})}{\delta}\right)\left( \frac{\TRelBound\frange}{\varepsilon} + \frac{\TRel\ITRelVar}{\varepsilon^{2}}\right)\right) \enspace.
  %\mathcal{O}\left( T\log\left(\frac{\log(\nicefrac{\frange}{\varepsilon})}{\delta}\right)\left( \frac{\frange}{\varepsilon} + \frac{\ITVar}{\varepsilon^{2}}\right)\right) \subseteq \mathcal{O}\left( \log\left(\frac{\log(\nicefrac{\frange}{\varepsilon})}{\delta}\right)\left( \frac{T\frange}{\varepsilon} + \frac{\TRel\SVar}{\varepsilon^{2}}\right)\right) \enspace. 
\end{align*}
%
%\cyrus{Can we add the finite-sample form for the ArXiv version?}
%
\end{restatable}

%We now have all the ingredients needed to prove \cref{thm:2.1}.
%\paragraph{Sampling Schedule and Proof Sketch}

\vspace*{-0.3 cm}
\subsection{Discussion and comparison with prior work}\label{sec:compare}

We now contrast the \algo\ sample complexity bound with those of prior art.
In order to focus on the salient differences between methods, we consider the asymptotic high-precision regime (i.e., sample complexity as $\varepsilon \to 0$), and to divide out the $\frac{1}{\varepsilon^{2}}\ln \frac{1}{\delta}$ terms, which necessarily appear in all bounds of this ilk, %all exponential tail bounds and are necessary, 
and we then report (in asymptotic notation) the quantity 
\begin{equation}
\HPSC(\M,f) \doteq \lim_{\varepsilon \to 0} \smash{\frac{\varepsilon^{2}}{\log \frac{1}{\delta}}} m(\M,f,\varepsilon,\delta )~.
\end{equation}
where $m(\M,f,\varepsilon,\delta)$ is the sample complexity %of a particular method 
as used before.  %(with method-specific dependencies on $f$ and $\M$).
%Note that 

\vspace*{-0.2 cm}

\paragraph{The Central Limit Theorem for Markov chains} We first contrast our bound with bounds in terms of the \emph{asymptotic variance} of $\M$, defined as

\vspace*{-0.2 cm}
\[
\AVar \doteq \lim_{\tau \to \infty} \tau \ITVn{\tau} \enspace.
\]
This quantity is fundamental to mean estimation, as the Markov chain central limit theorem \cite{gordin1978central} states
\[
\lim_{\tau \to \infty} \frac{1}{\sqrt{\tau}} \Favg^{(\tau)}(\vec{X}) \distributed \mathcal{N}(\Expect[f], \AVar) \enspace, %\frac{1}{\sqrt{m}} \sum_{i=1}^{m}
\]
thus asymptotically, via the CLT we get mean estimation sample complexity
\[
\HPSC_{\rm CLT}(\M,f) = \Theta(\AVar) =\Theta(\TRel\ITRelVar)\enspace.
\]

%some exes asymptotically gaussian best when coupled with best we can expect best we can expect asymptotic behavior of cube.

There are two main factors which prevent finite-sample bounds from achieving $\Theta(\VAsy)$. %in practice.
The first is that, like in the i.i.d.\ case, %finite sums are not actually Gaussian. % and 
the CLT is only an asymptotic result, and for finite samples, $\Favg^{(\tau)}(\vec{X})$ may be far from Gaussian. %particularly for chains with strong dependence between samples. %far-reaching dependencies in time.
%the further from Gaussian sums can potentially be, and
The second is that  a priori knowledge of $\VAsy$ %(, these quantities (often hard to obtain due to the sophistication of the
requires highly sophisticated analysis.
%Due to these challenges
%To circumvent these challenges, 
Thus, finite-sample guarantees  are often 
%in terms %on additional terms  (e.g., 
%they are 
stated in terms of
\emph{mixing times} or the \emph{stationary variance}, instead of more fundamental quantities and they fail to match the  CLT.

\vspace*{-0.2 cm}

\paragraph{\algo{}'s high precision complexity}From \Cref{thm:efficiency}, we conclude that \vspace*{-0.2 cm}

\[\HPSC_{\algo{}}= \TRel \ITRelVar \cdot  (1+\ln\ln\bigl(\mathsmaller{\frac{\frange}{\varepsilon}})\bigr) =\Theta\bigl(\VAsy\cdot (1+\ln\ln(\mathsmaller{\frac{\frange}{\varepsilon}}))\bigr) \enspace.
\]
Note that our bound matches the CLT, %and Paulin's: 
except for the $(1 + \ln \ln \frac{\frange}{\varepsilon})$ term, which is due to the progressive sampling union bound, and is what allows us to improve dependence on $\AVarBound$ to $ \Theta(\AVar)$.

\vspace*{-0.2 cm}
\paragraph{Comparison with asymptotic variance bounds \cite{paulin2015}} Paulin \cite{paulin2015} (Thm. 3., Eq. 3.20) accounts for finite-sample approximation error with a finite sample bound of
\[
m_{\rm Paulin}(\M,f,\varepsilon,\delta) = \Theta\left( \ln \frac{1}{\delta} \left( \frac{\frange \TRelBound}{\varepsilon} + \frac{\AVarBound}{\varepsilon^{2}} \right) \right) \enspace, \text{which implies }\ \   \HPSC(\M,f) = \Theta(\AVarBound) \enspace.
\]

\vspace*{-0.1 cm}
\paragraph{Comparison with MCMC Bernstein bound (\cref{thm:bernstein})}

%In contrast to the bounds in terms of the \emph{asymptotic variance} $\AVar$, 
Most Bernstein-type bounds  depend on (loose) \emph{a priori} bounds on the stationary variance $\SVar$ of the chain \emph{and its relaxation time}. Thus they are inferior to our method which depends on $\VAsy$.
%This results in a sample complexity of
\begin{align*}
\HPSC_{\rm Bern}(\M,f) = \Theta\left(\TRelBound\SVarBound\right) 
\mathop{\implies}\limits_{\mathclap{\text{ lem. \ref{lemma:tvar-prop}}}} \HPSC_{\rm Bern}(\M,f) = \Omega\left(\VAsy\right) \enspace.
\end{align*}
%which is inferior to our methods
%via \cref{lemma:tvar-prop}, 
%because of their full dependency on (loose) \emph{a priori} bounds and that $ \TRel\ITRelVar \leq \TRel\SVar$.
%, and %furthermore these %static (non-adaptive) 
%methods depend

%on $\TRel$ and $\SVar$.

\vspace*{-0.2 cm}
\paragraph{Comparison with function specific mixing Hoeffding bound \cite{functionMixingRabinovish}}
Rabinovich et al.\ 
\cite{functionMixingRabinovish} show a bound of the form
\begin{align*}
m_{\rm Rabi}(\M,f,\varepsilon,\delta) \leq \Theta \left( \frac{\TfMixBound(f,\M \frac{\frange}{\varepsilon})\frange^{2}}{\varepsilon^{2}} \log \frac{1}{\delta} \right) &\implies \HPSC_{\rm Rabi}(\M,f) \in \Theta \left( \TfMixBound(f,\M, \mathsmaller{\frac{\frange}{\varepsilon}})\frange^{2} \right)\\
&\mathop{\implies}\limits_{\mathclap{\text{see \ref{sec:compareproofs}}}}\HPSC_{\rm Rabi}(\M,f) \in \Omega \left( \AVar \log \mathsmaller{\frac{\frange}{\varepsilon}} \right) 
\enspace.
\end{align*}

The above bound is loose because it depends on an upper-bound of $\TfMix$, rather than $\TfRel$ (see \cref{def:trefmix}), and because, as a Hoeffding like-bound, it uses %the loose upper-bound 
$\frange^2$ instead of any variance proxy. %One can ask whether obtaining a variance aware bound of the form $\Theta(\TfRelBound\SVarBound)$ is possible exploiting the machinery Rabinovich et al has developed.
%Note that using \cref{coro:fmix} we have   $\VAsy\in \Theta(\TfRel \ITVn{\TfRel}) \subseteq O(\TfRel\SVar)$. 
More importantly as prior work has shown difficulty of estimating the standard relaxation time from observations, a serious obstacle in employing it in practice is to find $\TfMixBound\geq \TfMix$ which needs full spectral decomposition of $\M$'s transition matrix.

\section{Application to a counting problem}\label{sec:5}

In section \ref{sec:tracevar}, we found  $\ITRelVar=o(\SVar)$ for  functions  whose images partition the state space of $\cal M$, producing simpler \emph{projection chains} (see \cref{exm:cycle}). 
In this section, our goal is to obtain similar results for the JVV counting to sampling reduction  when traces of a Markov chain \emph{are hardly distinguishable} from a simpler  \emph{projection chain}. Developing some mathematical tools, we show that using \algo{} as a mean estimator not only makes  the complexity of JVV algorithm less dependent on loose mixing (relaxation) time bounds, but more importantly, we show it significantly reduces computation cost in some instances like counting number of $k$-colorings in planted partitions. 
While we present results for the number of proper $k$-colorings of a graph $G=(V,E)$, similar techniques can be used for other similar counting problems like \emph{independent sets}.

\emph{Exact} counting of the number of proper $k$-colorings is known to be \#P-hard \cite{ValiantHardness}, thus research has been conducted to find a \emph{fully polynomial time randomized approximation scheme} (FPRAS) for this problem. 

Jerrum 
\footnote{This result was improved by many, most importantly \cite{coloringvigodasmainresult}. See e.g.,\cite{coloringsurvey} for a survey and recent result e.g.,  \cite{Chen2021RapidMF}. }
proved that a Glauber dynamics chain on $k$-colorings (see \cref{def:coloringchain}) is rapidly mixing when $k>2\Dmax$
. Furthermore, based on the seminal work of Jerrum, Valient and Vazirani  (JVV) on counting self-reducible structures  \cite{countingsamplingreduction1986jurrumvaziranivalient}, he showed  a FPRAS for counting $k$-colorings  using a telescoping sum of  MCMC-mean estimation sub-problems (see Section \ref{sec:jvvred} or \cite{jurrumSamplingiscounting} for more details). We plug-in \algo{} as a mean estimation subroutine in Jerrum's FPRAS.  
 \begin{figure}[h]
\colorlet{rcol}{red!90!black}
\colorlet{bcol}{blue!95!black}
\colorlet{ccol}{cyan!80!black}
\scalebox{1.1}{
\null\hspace{-0.25cm}\begin{tikzpicture}[
        scale=1.13,
        region/.style={thick},
        rregion/.style={region,draw=rcol,fill=rcol,fill opacity=0.02},
        bregion/.style={region,draw=bcol,fill=bcol,fill opacity=0.02},
        cregion/.style={region,draw=ccol,fill=ccol,fill opacity=0.02},
        uncutedge/.style={line width=0.18mm,opacity=0.45},%draw=uncotcol,
        cutedge/.style={line width=0.22mm,line cap=round,{dash pattern=on 0pt off 1.75\pgflinewidth},opacity=0.65},%draw=cutcol,
        testedge/.style={line width=0.26mm,dash pattern={on 3\pgflinewidth off 2\pgflinewidth},opacity=0.91},%draw=testcol,
        %Generic vertices
        vtx/.style={fill=black,opacity=0.7},
    ]

\begin{scope}

%3 Regions
\draw[rregion] plot [smooth cycle] coordinates {(0,0) (1,1) (1.8,1) (1,0) (1,-1)};
\draw[bregion] plot [smooth cycle] coordinates {(2.2,0) (3,1) (3.2,0.8) (3.1,1.5) (2.5,1.5)  (2.2,1.5) };
\draw[cregion] plot [smooth cycle] coordinates {  (2,-1)  (2,-0.5)  (1.5,0) (2,0)  (2.5,-1)  };

%cut edges
\draw[cutedge] (1.6,0)--(0.8,0); 
\draw[cutedge] (2,-0.2)--(2.2,0.3);
\draw[cutedge] (1.6,1)--(2.3,1); 
\draw[cutedge] (1.6,0.8)--(2.3,0.6); 

%gray cut  vertices 
\draw[vtx]  (1.6,1) circle (0.02 cm);
\draw[vtx]  (2.3,1) circle (0.02 cm);
\draw[vtx]  (1.6,0.8) circle (0.02 cm);
\draw[vtx]  (2.3,0.6) circle (0.02 cm);
\draw[vtx]  (0.8,0) circle (0.02 cm);
\draw[vtx] (2.2,0.3) circle (0.02 cm);
\draw[vtx]  (1.6,0) circle (0.02 cm);
\draw[vtx]  (2,-0.2) circle (0.02 cm);

%more vertices 
\draw[vtx]  (1,0.9) circle (0.02 cm);
\draw[uncutedge] (1,0.9)-- (1.3,0.5);
\draw[vtx]  (1.3,0.5) circle (0.02 cm);
\draw[vtx]  (1,0.7) circle (0.02 cm);
\draw[vtx]  (0.7,0.6) circle (0.02 cm);
\draw[uncutedge](0.7,0.6)--(0.5,0.5);
\draw[vtx]  (0.5,0.5) circle (0.02 cm);
\draw[vtx]  (0.9,0.2) circle (0.02 cm);
\draw[vtx]  (0.2,0.1) circle (0.02 cm);
\draw[vtx]  (0.2,-0.2) circle (0.02 cm);
\draw[vtx]  (0.9,-0.4) circle (0.02 cm);
\draw[vtx]  (0.9,0.5) circle (0.02 cm);

\draw[vtx]  (2.7,1.4) circle (0.02 cm);
\draw[vtx]  (2.8,1.2) circle (0.02 cm);
\draw[uncutedge](2.8,1.2)-- (2.7,1.4);
\draw[uncutedge](2.8,1.2)--  (2.6,1);
\draw[vtx]  (2.6,1) circle (0.02 cm);
\draw[vtx]  (2.4,1.3) circle (0.02 cm);
\draw[vtx]  (2.4,0.4) circle (0.02 cm);
\draw[vtx]  (2.1,-0.5) circle (0.02 cm);
\draw[vtx]  (2.2,-0.8) circle (0.02 cm);
\draw[vtx]  (2.8,1.2) circle (0.02 cm);
\draw[vtx]  (2.5,0.6) circle (0.02 cm);

\draw[uncutedge]  (0.5,0)--(0.9,0.5);
\draw[uncutedge]  (0.5,0)--(0.9,-0.4);

\filldraw[black]  (0.5,0) circle (0.035 cm);
\node[anchor=east,inner sep=2.0pt] (labelu) at (0.5, 0) {\small $u$};
\filldraw[black]  (0.63,-0.5) circle (0.035 cm);
\node[anchor=east,inner sep=1.9pt] (labelv) at (0.63,-0.5) {\small $v$};

%Test edge (used by $f$)
\draw[testedge] (0.5, 0) -- (0.63,-0.5);

%\node[anchor=north] (label) at (1.6, -1.1) {\footnotesize $G$: loosely connected components $V_r$, $V_b$, and $V_c$};
\node[anchor=north] (label1) at (1.6, -1.1) {\small $G$: loosely connected components};
\node[anchor=north] (label2) at (1.6, -1.45) {\small $\color{rcol}V_r$, $\color{bcol}V_b$, and $\color{ccol}V_c$};
\end{scope}

\node[anchor=north] (label) at (4.0, 0.5) {$\xleftrightarrow{\text{Coupling }}$};

\begin{scope}[xshift=5.0cm]

%3 Regions
\draw[rregion] plot [smooth cycle] coordinates {(0,0) (1,1) (1.8,1) (1,0) (1,-1)};
\draw[bregion] plot [smooth cycle] coordinates {(2.2,0) (3,1) (3.2,0.8) (3.1,1.5) (2.5,1.5)  (2.2,1.5) };
\draw[cregion] plot [smooth cycle] coordinates {  (2,-1)  (2,-0.5)  (1.5,0) (2,0)  (2.5,-1)  };

\filldraw[black]  (0.5,0) circle (0.035 cm);
\node[anchor=east,inner sep=2.0pt] (labelu) at (0.5, 0) {\small $u$};
\filldraw[black]  (0.63,-0.5) circle (0.035 cm);
\node[anchor=east,inner sep=1.9pt] (labelv) at (0.63,-0.5) {\small $v$};

%Test edge (used by $f$)
\draw[testedge] (0.5, 0) -- (0.63,-0.5);

%more vertices 
\draw[vtx]  (0.8,0) circle (0.02 cm);
\draw[vtx] (2.2,0.3) circle (0.02 cm);

\draw[vtx]  (1.6,1) circle (0.02 cm);
\draw[vtx]  (2.3,1) circle (0.02 cm);
\draw[vtx]  (1.6,0.8) circle (0.02 cm);
\draw[vtx]  (2.3,0.6) circle (0.02 cm);
\draw[vtx]  (1,0.9) circle (0.02 cm);
\draw[uncutedge] (1,0.9)-- (1.3,0.5);
\draw[vtx]  (1.3,0.5) circle (0.02 cm);
\draw[vtx]  (1,0.7) circle (0.02 cm);
\draw[vtx]  (0.7,0.6) circle (0.02 cm);
\draw[uncutedge](0.7,0.6)--(0.5,0.5);
\draw[vtx]  (0.5,0.5) circle (0.02 cm);
\draw[vtx]  (0.9,0.2) circle (0.02 cm);
\draw[vtx]  (0.2,0.1) circle (0.02 cm);
\draw[vtx]  (0.2,-0.2) circle (0.02 cm);
\draw[vtx]  (0.9,-0.4) circle (0.02 cm);
\draw[vtx]  (0.9,0.5) circle (0.02 cm);

\draw[vtx]  (2.7,1.4) circle (0.02 cm);
\draw[vtx]  (2.8,1.2) circle (0.02 cm);
\draw[uncutedge](2.8,1.2)-- (2.7,1.4);
\draw[uncutedge](2.8,1.2)--  (2.6,1);
\draw[vtx]  (2.6,1) circle (0.02 cm);
\draw[vtx]  (2.4,1.3) circle (0.02 cm);
\draw[vtx]  (2.4,0.4) circle (0.02 cm);
\draw[vtx]  (2.1,-0.5) circle (0.02 cm);
\draw[vtx]  (2.2,-0.8) circle (0.02 cm);
\draw[vtx]  (2.8,1.2) circle (0.02 cm);
\draw[vtx]  (2.5,0.6) circle (0.02 cm);

\draw[uncutedge]  (0.5,0)--(0.9,0.5);
\draw[uncutedge]  (0.5,0)--(0.9,-0.4);
\draw[vtx]  (1.6,0) circle (0.02 cm);
\draw[vtx]  (2,-0.2) circle (0.02 cm);

%Label:
%\node[anchor=north] (label) at (1.6, -1.1) {\footnotesize $G'$: disconnected components $V_r$, $V_b$, and $V_c$};

\node[anchor=north] (label1) at (1.6, -1.1) {\small $G'$: disconnected components};
\node[anchor=north] (label2) at (1.6, -1.45) {\small $\color{rcol}V_r$, $\color{bcol}V_b$, and $\color{ccol}V_c$};

\end{scope}

\node[anchor=north] (label) at (9.0, 0.5) {$\xrightarrow{\text{Projection}}$};
\begin{scope}[xshift=10.0cm]
\begin{scope}[yshift=-0.3cm,scale=1.2,rotate=18]

%1 Region
\draw[rregion] plot [smooth cycle] coordinates {(0,0) (1,1) (1.8,1) (1,0) (1,-1)};

%UV vertices
\filldraw[black] (0.5,0) circle (0.035 cm);
\node[anchor=north east,inner sep=2.0pt,yshift=0.8pt,xshift=-0.5pt] (labelu) at (0.5, 0) {\small $u$};
\filldraw[black] (0.63,-0.5) circle (0.035 cm);
\node[anchor=north east,inner sep=1.9pt,yshift=0.8pt,xshift=-0.5pt] (labelv) at (0.63,-0.5) {\small $v$};

%Test edge (used by $f$)
\draw[testedge] (0.5, 0) -- (0.63,-0.5);

%more vertices 
\draw[vtx]  (1.6,1) circle (0.02 cm);

\draw[vtx]  (1,0.9) circle (0.02 cm);
\draw[uncutedge] (1,0.9)-- (1.3,0.5);
\draw[vtx]  (1.3,0.5) circle (0.02 cm);
\draw[vtx]  (1,0.7) circle (0.02 cm);
\draw[vtx]  (0.7,0.6) circle (0.02 cm);
\draw[uncutedge](0.7,0.6)--(0.5,0.5);
\draw[vtx]  (0.5,0.5) circle (0.02 cm);
\draw[vtx]  (0.9,0.2) circle (0.02 cm);
\draw[vtx]  (0.2,0.1) circle (0.02 cm);
\draw[vtx]  (0.2,-0.2) circle (0.02 cm);
\draw[vtx]  (0.9,-0.4) circle (0.02 cm);
\draw[vtx]  (0.9,0.5) circle (0.02 cm);

\draw[uncutedge]  (0.5,0)--(0.9,0.5);
\draw[uncutedge]  (0.5,0)--(0.9,-0.4);

\end{scope}

%Label:
\node[anchor=north] (label) at (1, -1.1) {\small $\color{rcol}V_r$};

\end{scope}

\end{tikzpicture}
}
\caption{
 Consider a loosely connected graph $G=(V_r\cup V_b\cup V_c,E)$, and for any $e=(u,v)\in E$, let  $f_e$ be as defined in  the JVV reduction, and $u,v\in V_r$.
Let $G'$ be the graph obtained by removing the edges connecting $V_r$ to the rest of the graph. 
We show that only a negligible mass of stationary traces of $\M_{G'}$ have zero probability in $\M_{G}$, and the remaining traces of $\M_{G'}$ can be perfectly coupled to identical traces in $\M_{G}$.
We then use projection chain of $\M_{G'}$ onto equivalence classes defined by the image of $f_e$. Note that the transition probabilities of this projection chain
%of  $f_e$, through $\M_{G'}$ 
are  identical to projection of $(\frac{\abs{V_r}}{n}){\M}_{V_r}$ onto equivalence classes defined by the image of $f_e$, where $(\frac{\abs{V_r}}{n}){\M}_{V_r}$ is a Markov chain looping w.p. $(1-(\frac{\abs{V_r}}{n}))$ and otherwise transitioning through  ${\M}_{V_r}$. Using this observation, we show $\ITRelVar(\M_G ,f_e)\leq \left(\frac{\TRel(\M_{G'})}{\TRel(\M_G)}\right)\ITRelVar(\M_{G'},f_e)=o(\SVar)$. 
}
 \label{fig:loosefig}
\end{figure}
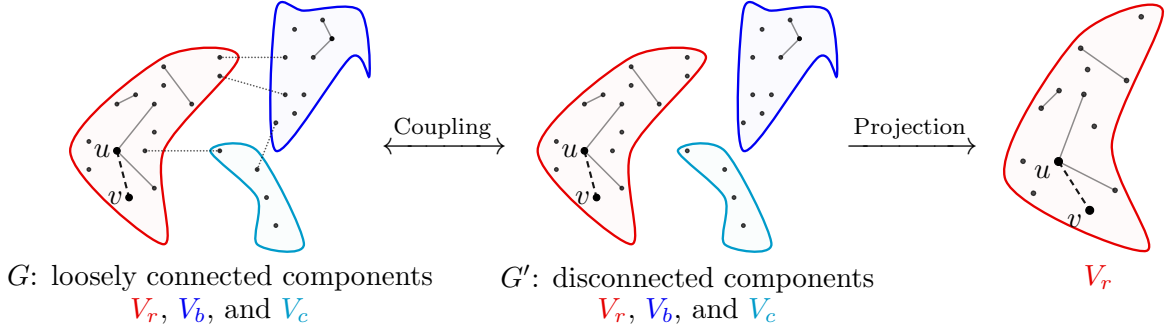
For an arbitrary graph $G$ and fixed $k$, let $\M_G$ be the  Glauber dynamics chain on $k$-colorings, and $\{f_i\}_{i=1}^{\cal E}$ be the functions appearing in Jerrum's reduction whose mean is to be estimated, where ${\cal E}=\abs{E}$, the number of edges in $G$.  We identify a class of graphs, which we call \emph{loosely connected} (see \cref{def:looselyconnected}), and we prove for any $G$ in this class that $\ITRelVar(\M_G,f_i)=o(\SVar)$ (see \cref{lemm:loosely}).

Using these results, we then prove that for $G$ sampled from the planted partition model with $r$ communities we have: $\ITRelVar=O(\SVar)/r$. Since JVV reduction needs high precision mean estimations, as it needs a union bound over all intermediate steps, using \algo{} is indeed impactful. In particular it reduces the complexity of the $k$-coloring fpras of \cite{jurrumSamplingiscounting} for planted partition from $\tilde{O}\left(n^2{\cal E}^3/\varepsilon^2\right)$ down to  $\tilde{O}\left(n^2\left({\cal E}^2(\frac{1}{\varepsilon})+({\cal E}^3/r)\frac{1}{\varepsilon^2}\right)\right)$, where $n=\abs{V}$ and ${\cal E}=\abs{E}$ and $r$ is the number of communities (see \cref{thm:plantedpartition}) .

Figure 2 %\ref{fig:loosefig}
illustrates the main ideas constituting the proof presented in detail in  \cref{sec:plantedap}.

\newpage

\newcommand{\algrule}[1][.2pt]{\par\vskip.5\baselineskip\hrule height #1\par\vskip.5\baselineskip}

%%%%%%%%%%%%%%%%%%%%%%%%%%%%%%%
%WARNING: TERRIBLE CODE AHEAD%%
%TAKE ALL STANDARD PRECAUTIONS%

\newcommand{\VBoundVarUpper}{\bm{u}}
\newcommand{\VBoundRange}{\bm{r}}
\newcommand{\VBoundMean}{\bm{\mu}}
\newcommand{\VBoundVar}{\bm{v}}

\newcommand{\VERawMoment}{\bm{a}}
\newcommand{\VBlockAvg}{\bm{b}}
\newcommand{\VEDRMean}{\bm{U}}
\newcommand{\VEDRVar}{\bm{V}}
\newcommand{\VEVar}{\hat{\bm{v}}}
\newcommand{\VEMean}{\hat{\bm{\mu}}}
\begin{algorithm}[htbp]
\algrenewcommand\algorithmicindent{1.1em}
%\scalebox{0.9}[0.75]{
%\scalebox{0.96}[0.95]{
%\begin{minipage}{1.04\textwidth}
%\scalebox{1}[1]{
%\begin{minipage}{1\textwidth}
\begin{algorithmic}[1]
\small
%Dynamite on the given chain

\Procedure{\DMCMC{}}{%
  $(x_{0}, x_{1})$, $\M$, $\EigTwoBound$, 
  $f$, {$\varepsilon$, $\delta$}} $\null \mapsto \hat{\mu}$ % (\VEMean, \VBoundMean)$

\State \textbf{Input:}
arbitrary initial states $x_{0}, x_{1} \in \X$ ,
 \emph{Markov chain} $\M$ over $\X$, \emph{second absolute eigenvalue upper-bound} $\EigTwoBound$, %\emph{relaxation time bound} $T$ s.t. $\TRel(\M) \leq T$,
\emph{function} $f: \X \to [a,b]$ with $\frange \doteq b - a$,
\emph{confidence interval radius} $\varepsilon$,
and
\emph{failure probability} $\delta \in (0, 1)$.
\State \textbf{Output:}
Additive $\varepsilon$, $\delta$ approximation $\hat{\mu}$ of $\mu=\Expect_{\pi}[f]$.

%\algrule[0.6pt]
\hrulefill\smallskip

%\State $\frange\gets b-a$ \Comment{Function range}\label{alg:frange}
%\Comment{Length of a trace}\label{alg:mixtime}

%Combining I and $\alpha$ on 1 line:

\State $\displaystyle(X_{0,1}, X_{0,2}) \distributed (\M \otimes \M)^{\TUnif}(x_{0}, x_{1})$ \Comment{Before collecting samples for mean estimation run two  independent copies of $\M$, each $\displaystyle\TUnif \gets \ceil{\frac{\ln \mathsmaller{\frac{1}{\PiMinBound}}}{\ln \mathsmaller{\frac{1}{\EigTwoBound}}}}$ steps
\footnote{
Starting at Nonstationarity. %Procedure, 
The MCMC Bernstein bound and McDiarmid bound we use here assume samples from stationary distribution. Often we can't assume even a single (perfectly) stationary sample may be efficiently drawn, and realistically can only start with an \emph{arbitrary} $x \in \Support(\pi)$.
Thus we use  standard \emph{warm start technique} (widely used in MCMC algorithms), where the chain is run from an arbitrary point for its \emph{uniform mixing time}\todo{cite}, and then \algo{} is run, applying a standard \emph{nonstationarity correction} (see~\cref{remark:stationarity}) to account for the nonstationary start.
%Note that this correction applies to both \cref{thm:hoeffding,thm:bernstein} of \cref{sec:prelim}.
}}\cyrus{$\delta$ problem}

\State
$\NIterations \gets \max\Bigl(1, \floor{\log_{2} \bigl( \smash{ %\left(
  %\frac{1+\EigTwoBound}{5(6+\EigTwoBound)} \cdot
  \frac{\frange}{2\varepsilon}
  } \bigr)} \Bigr)$; \ %\right)}$;
$\displaystyle\alpha \gets \frac{(1 + \EigTwoBound)\frange\ln \frac{3\NIterations}{\delta}}{(1 - \EigTwoBound)\varepsilon}$; \ 
%Old $\NIterations$
%$\NIterations \gets \floor{\log_{2}\left( \frac{1 + \EigTwoBound}{35.582 + 3.400\EigTwoBound} \cdot \frac{\frange}{\varepsilon} \right)}$; 
%Old $\alpha$
%$\alpha \gets \frac{(17.791 + 1.700 \EigTwoBound)\ln \frac{3\NIterations}{\delta}}{(1-\EigTwoBound)} \cdot \frac{\frange}{\varepsilon}$;
$m_{0} \gets 0$ \Comment{Initialize \emph{sampling schedule} %(Iterations $\NIterations$, optimistic sample size $\alpha$)
}\label{alg:niter}\label{alg:alpha}

%\cyrus{What happens if $\varepsilon > \frange/2$?}

\if 0 %Removed special case due to loose schedule selection

\If{$\NIterations < 1$} \Comment{Weak guarantee requested; fall back to static Hoeffding}

\State $\NIterations \gets 1; \alpha \gets \mathsmaller{\frac{1}{2}}m_{H}(\EigTwoBound, \frange, \varepsilon, \delta %\frac{\delta}{2}
)$ \Comment{Single-step schedule}

\EndIf

\fi

%\State $$ \Comment{Minimum sample size}\label{alg:alpha}

%\State $\NIterations \gets \left\lfloor \log_{2}(\frac{9\frange}{100\varepsilon}) \right\rfloor$ \Comment{Maximum iteration count}\label{alg:niter}

%\State $\alpha \gets \TRel \ln(\frac{3\NIterations}{\delta}) \frac{10\frange}{\varepsilon}$ \Comment{Minimum sample size}\label{alg:alpha}

%\State $m_{0} \gets 0$ \Comment{Initialize cumulative sample size to 0}

%\State $(X_{T-1,1}, X_{T-1,2}) \sim \pi(\M \otimes \M)$ \Comment{Start the chain \emph{at stationarity}}\label{alg:initialize}

%\todo{No tensor product notation?}
%\todo{flip chain indices?}

%\State $m\gets v^2/\varepsilon$ \Comment{Maximum number of samples \cyrus{constants and stuff needed.}}\label{alg:ubms:niter}

%\State $\NIterations \gets \lceil \log_{2}(m) \rceil$

%Log term macro:
\newcommand{\ldt}{\ln({\frac{3\NIterations}{\delta}})} 

\For{$i \in 1, 2, \dots, \NIterations $}
\State $m_i \gets \left\lceil \alpha 2^i \right\rceil$ \Comment{Total sample count at iteration $i$}\label{alg:ss} %\cyrus{Should we call these ``epochs'' instead of iterations / timesteps?}
%\State Starting\footnote{Here we assume $X_0\sim \pi$, thus, $\TMix$ initial steps are needed for to generate $X_0$. } from $X_0$, run the Markov chain for $m_i T$ number of steps, to take $m_i$ trace samples. \cyrus{Do we have notation for this?}

%\For{$j \in  T(m_{i-1} + 1), \dots, T(m_{i} + 1) - 1$}

\if 0
\State $\forall j \in %\{ 
  (m_{i-1} \! + \! 1), \dots, m_{i} %\}
  \!: \, (X_{j,1}, X_{j,2}) \sim (\M \otimes \M)(X_{j-1,1}, X_{j-1,2})$
%\Comment{Extend the chain}
\Comment{Run both chains up to step $m_{i}$}%
%\Comment{Run both chains for $m_{i} - m_{i-1}$ steps}%\todo{get this on one line}
%\Comment{Run the chain from the previous state}
\label{alg:run-chain}
\fi

\For{$j \in (m_{i-1} + 1), \dots, m_{i}$}

\State $(X_{j,1}, X_{j,2}) \sim (\M \otimes \M)(X_{j-1,1}, X_{j-1,2})$
\Comment{Run both chains up to step $m_{i}$}
\EndFor

%\EndFor

%Old: Trace averaging loop:
\if 0
\For{$k \in m_{i-1}+1, \dots, m_{i}$}

\vspace{-0.08cm}

\State $\displaystyle %\forall k \in \{m_{i-1}+1, \dots, m_{i}\}: \, 
\bm{s}_{k} \gets {\frac{1}{2T} \sum_{j=1}^{T} (f \oplus f)(X_{Tk+j,1},X_{Tk+j,2})}$ \Comment{Compute trace-average samples} \label{alg:tavg}

\vspace{-0.18cm}

\State $\displaystyle %\forall k \in \{m_{i-1}+1, \dots, m_{i}\}: \, 
\bm{r}_{k} \gets { \frac{1}{2}\biggl( \frac{1}{T} \sum_{j=1}^{T} (f \ominus f)(X_{Tk+j,1}, X_{Tk+j,2})\biggr)^{2} }$ \Comment{Compute trace-variance samples} \label{alg:tavg}

\vspace{-0.25cm}

\EndFor

\vspace{-0.1cm}
\fi

%%%%%%%%%%%%%%%%%%%%%%%%%%%%%%%%%%
%Compute block means and variances

%Combined onto a single line

\State $\displaystyle\VEMean_{i} \gets \mathsmaller{\frac{1}{2{m_i}}} \sum_{j=1}^{m_i} \bigl( f(X_{j,1}) + %\oplus 
f(X_{j,2})\bigr)$ \Comment{Empirical mean}\label{alg:emean} %(averaged over both tensor product chains)
%\todo{Make sure ominus and oplus clear}

\State $\displaystyle\VEVar_{i} \gets \mathsmaller{\frac{{1}}{2{m_i}}} \sum_{j=1}^{m_i} \bigl(f(X_{j,1}) - %\ominus
f(X_{j,2})\bigr)^{2}$ \Comment{Empirical variance}\label{alg:evar} %\Comment{Empirical inter-trace variance}\label{alg:evar}

%\State $\VEMean_{i} \gets \mathsmaller{\frac{1}{{m_i}}} \sum_{k=1}^{m_i}\bm{s}_k$ \Comment{Empirical  mean}\label{alg:emean}

%\State $ \VEVar_{i} \gets \mathsmaller{\frac{{1}}{{m_i}}} \sum_{k=1}^{m_i}(\bm{s}_k-\VEMean_i)^2$ \Comment{Inter-trace variance}\label{alg:evar}
%\vspace{-0.05cm}

\State $\displaystyle \VBoundVarUpper_{i} \gets \VEVar_{i} + \frac{(11 + \sqrt{21})(1 + 
  %\mathsmaller{\frac{\EigTwoBound}{\sqrt{21}}}
  \nicefrac{\EigTwoBound}{\sqrt{21}}
) \frange^{2} \ln \frac{3\NIterations}{\delta}}{(1-\EigTwoBound)m_{i}} + \sqrt{\frac{(1 + \EigTwoBound)\frange^{2} \VEVar_{i}\ln \frac{3\NIterations}{\delta}}{(1 - \EigTwoBound)m_{i}}}
$\Comment{Variance upper bound}\label{alg:mcd-var}\label{alg:2chain-var}

\State $\displaystyle \hat{\bm{\epsilon}}_{i} \gets \frac{10\frange\ln \frac{3\NIterations}{\delta}}{(1-\EigTwoBound)m_{i}} + \sqrt{\frac{(1 + \EigTwoBound) \VBoundVarUpper_{i}\ln \frac{3\NIterations}{\delta}}{(1 - \EigTwoBound)m_{i}}}$ \Comment{Apply Bernstein bound}\label{alg:bern}

\vspace{0.1cm}

\If{$(i = \NIterations) \vee (\hat{\bm{\epsilon}}_{i} \leq \epsilon)$} \Comment{Terminate if accuracy guarantee is met}\label{alg:tc} %\comment{Check termination conditions: $i = \NIterations$: Hoeffding, $\hat{\bm{\epsilon}}_{i} \leq \epsilon$: Bernstein}\label{alg:tc}
%\vspace{-0.1cm}
\State \Return $\VEMean_{i}$ \label{alg:return}
\EndIf

\EndFor
\EndProcedure

%\hrulefill
%\linegoal

%\end{algorithmic}
%\begin{algorithmic}[1]

%Dynamite on a trace chain

%\Procedure{\textsc{StationaryTrace\algo{}}}{$\M$, $\EigTwoBound$, $f$, {$\varepsilon$, $\delta$}} $\null \mapsto \hat{\mu}$

\algrule[0.6pt]
%\hrulefill

\Procedure{\textsc{\algo{}}}{$(x_{0}, x_{1})$, $\M$, $\EigTwoBound$\shahrzad{check}, $f$, {$\varepsilon$, $\delta$}} $\null \mapsto \hat{\mu}$%\todo{Have this first?}

%\algrule[0.6pt]

\State \textbf{Input:}
%Initial state $(x_{0}, x_{1}) \distributed \pi(\M) \times \pi(\M)$\todo{Not always (w/ nonstat start)},
Initial state $(x_{0}, x_{1}) \in \X \times \X$, %\todo{Not always (w/ nonstat start)},
\emph{lazy Markov chain} $\M$ over $\X$,\todo{need lazy? for exponentially decreasing $\EigTwoBound$?} \emph{second absolute eigenvalue upper-bound} $\EigTwoBound$, %(i.e., $\EigTwo(\M) \leq \EigTwoBound$)\todo{absolute, better star notation?}, %s.t. $\TRel(\M) \leq T$,
\emph{function} $f: \X \to [a,b]$, % with $\frange \doteq b - a$,
\emph{confidence interval radius} $\varepsilon$,
and
\emph{failure probability} $\delta \in (0, 1)$.
 \State \textbf{Output:}
Additive $\varepsilon$, $\delta$ approximation $\hat{\mu}$ of $\mu=\Expect_{\pi}[f]$.

%\algrule[0.6pt]
%\medskip
\hrulefill\medskip

\State $T \gets %\ceil{\ln(2)\left(\frac{1}{1 - \EigTwoBound} - \frac{1}{2}\right)} = 
\ceil{\frac{1 + \EigTwoBound}{1 - \EigTwoBound} \ln\sqrt{2} }$ \Comment{Select $T$ s.t.\ $\TRel(\M^{(T)}) \leq 2$}

%\State $x_{0}, x_{1} \distributed \pi(\M)$ \Comment{Sample initial states from stationarity} \todo{Option for this?}

%\State  $(X_{0}, X_{1}) \gets \bigl( \bigl( \underbrace{0, 0, \dots, 0}_{T-1 \textsc{ Zeros}}, x_{0} \bigr), \bigl( \underbrace{0, 0, \dots, 0}_{T-1 \textsc{ Zeros}}, x_{1} \bigr) \bigr)$ \cyrus{0 doesn't make sense: is this better?} \Comment{Initial trace chain to extend to stationarity}\label{alg:dyna:init-trace}

%{\color{red}
%\State $(X_{0}, X_{1}) \gets \biggl( \bigl( \underbrace{x_{0}, x_{0}, \dots, x_{0}}_{T \textsc{ Copies}} \bigr), \underbrace{s \in \X, s \in \X, s \in \X, \dots, s \in \X}_{T-1 \textsc{ Arbitrary Values}}, x_{0} \bigr), \bigl( \underbrace{s, s, s, \dots, s}_{T-1 \textsc{ Arbitrary \ensuremath{s} Values}}, x_{1} \bigr) \biggr)$ \cyrus{Which is better?} \Comment{Initial trace chain to extend to stationarity}\label{alg:dyna:init-trace}
%}

%{\color{blue}
%\State $(X_{0}, X_{1}) \gets \Bigl( \bigl( s_{1} \in \X, s_{2} \in \X, \dots, s_{T-1} \in \X, x_{0} \bigr), \bigl( s_{1} \in \X, s_{2} \in \X, \dots, s_{T-1} \in \X, x_{1} \bigr), \bigl( \underbrace{s, s, s, \dots, s}_{T-1 \textsc{ Arbitrary \ensuremath{s} Values}}, x_{1} \bigr) \biggr)$ \cyrus{Which is better?} \Comment{Initial trace chain to extend to stationarity}\label{alg:dyna:init-trace}

\State $(X_{0}, X_{1}) \gets \bigl( ( \underbrace{s_{1}, s_{2}, s_{3}, \dots, s_{T-1}}_{\mathclap{\textsc{Arbitrary } s_{1:T-1} \in \S^{T-1}}}, x_{0} ), ( \underbrace{s_{1}, s_{2}, s_{3}, \dots, s_{T-1}}_{\mathclap{\textsc{Arbitrary } s_{1:T-1} \in \S^{T-1}}}, x_{1} ) \bigr)$ \Comment{Initialize trace chain state}\label{alg:dyna:init-trace}
%}

\if 0
\For{$i \in 2, \dots, T$}

\State $x'_{0,i} \gets \M(x_{0,i-1})$; $x'_{1,i} \gets \M(x_{1,i-1})$ \Comment{Extend initial samples to stationary trace chain}\cyrus{Doesn't quite work}

\EndFor

\State $(X_{0}, X_{1}) \gets \Bigl( \bigl( x_{0}, x'_{0,2}, x'_{0,3}, \dots, x'_{0,T} \bigr), \bigl( x_{1}, x'_{1,2}, x'_{1,3}, \dots, x'_{1,T} \bigr) \Bigr)$ \Comment{Initialize stationary trace chain}\label{alg:dyna:init-trace}
\fi

%\State \cyrus{Draw out to stationarity.}

\State \Return $\DMCMC\bigl((X_{0}, X_{1}), \M^{(T)}, \EigTwoBound^{T}, \Favg^{(T)}, {\varepsilon, \delta}\bigr)$ \Comment{Run \DMCMC{} on trace chain}\label{alg:dyna:return}

\EndProcedure
\end{algorithmic}
%\end{minipage}
%Note that ...
%}
\caption{\DMCMC{} and \algo{} routines}
%\cyrus{Note: lazy + reversibility required for warm start analysis?}

\label{alg:1}
\end{algorithm}

\appendixtrue

%Remove labels from restatables
\xapptocmd{\beginrestatable}{\renewcommand{\label}[1]{}}{}{}

%\newpage
\newpage 

\appendix

%%%%%%%%%%%%%%%%
%Required below:

\providecommand{\VBoundVarUpper}{\bm{u}}
\providecommand{\VBoundRange}{\bm{r}}
\providecommand{\VBoundMean}{\bm{\mu}}
\providecommand{\VBoundVar}{\bm{v}}

\providecommand{\VERawMoment}{\bm{a}}
\providecommand{\VBlockAvg}{\bm{b}}
\providecommand{\VEDRMean}{\bm{U}}
\providecommand{\VEDRVar}{\bm{V}}
\providecommand{\VEVar}{\hat{\bm{v}}}
\providecommand{\VEMean}{\hat{\bm{\mu}}}

%%%%%%%%%%%%%%%%%%%%%%%%%%%%%%%%

\section{A Compendium of Complementary Material}
\label{sec:missingproofs}

We now present all the details backing our results including proofs and rigours definitions for newly developed concepts.
Some of these proofs are involved and are broken to a proof sketch describing ideas and intuitions accompanied by a full proof with  rigorous details. 

In \cref{sec:thms} we state the bounds, theorems, and algorithms which exist in the literature and are used in our proofs.

We have provided a table of contents at the end of this paper.

\subsection{Intra-Trace variance properties}\label{proofs:trace}

\begin{defin}
Consider $\vec{X}_{1:T}:X_1, X_2,\dots X_T$.
We define $C_i$  to be  the lag $i$-autocovariace, i.e., the autocovariace of two steps of $\M$ (at stationarity) being $i$ apart. i.e., $C_{i} \doteq \Cov(X_{1}, X_{1+i})$.
\end{defin}

Note that for i.i.d samples the autocovariance of any pair of samples $X_i$ and $X_j$ is zero.
We do not have independence here  nevertheless  for reversible  Makrov chains (See for example Equation 12.9 from \cite{levin2017markov}) we have $C_{i} \leq \EigTwo^{i}\sqrt{\Var_{\pi}[f]\Var_{\pi}[f]}$.
The following lemmas will be used throughout:

\begin{lemma}\label{lem:autocov}
Suppose $\vec{X}_{1:T}:X_1, X_2,\dots X_T$ is a trace of length $T$ of $\M$. using the above definition for $C_i$, the trace variance are related as %follows:  
\begin{equation}
\label{eq:autocov}
 \ITVar = \frac{1}{T} \SVar+ \frac{2}{T^2}\sum_{i=1}^{T-1}(T - i) C_{i} \enspace.
\end{equation}
\end{lemma}

\begin{proof}

The trace variance is
\begin{align*}
\ITVar =
   \Expect\left[\left(\frac{1}{T}\sum_{i=1}^{T} (f(X_{i}) - \mu) \right) ^2\right] &= \Expect\left[\frac{1}{T^{2}}\sum_{i=1}^{T}\sum_{j=1}^{T} (f(X_{i}) - \mu)(f(X_{j}) - \mu) \right]\\
   &= \frac{1}{T} \SVar+ \frac{2}{T^2}\sum_{i=1}^{T-1}(T - i) C_{i}. 
\end{align*}

\end{proof}

\shahrzad{restatable: Labels being lost?}
We now show \cref{lemma:tvar-prop}.
{\renewcommand{\label}[1]{}
\lemmatvarprop*
}
\begin{proof}
We first show \eqref{eq:tvdec}.
First note that
\[
T\ITVn{T} = \SVar + 2\sum_{i=1}^{T-1} \frac{T-i}{T} C_{i} \enspace.
\]
Laziness implies nonnegativity of each covariance term $C_{i}$, and the weight $\max(0, \frac{T-i}{T}) = \max(0, 1 - \frac{i}{T})$ assigned to each $C_{i}$ in $T\ITVn{T}$ is monotonically increasing in $T$.

\bigskip

We now show \eqref{eq:tvtinc}.
%\cyrus{
%I think the last 2 are implied by laziness (no negative eigenvalues)?  I think laziness implies 
We first show that laziness and reversibility implies $C_{i}$ is monotonically decreasing $\forall i \geq 0$.
To see this, consider again the spectral decomposition
\[
C_{t} = \int \lambda_{i}^{\lvert t \rvert} \ \mathrm{d} E_{f}(\lambda_{i}) \enspace,
\]
and note that each $\lambda_{i}^{\abs{t}}$ in the integral is \emph{decreasing} in $t$ (as \emph{laziness} implies $\lambda_{i} \in [0, 1]$, and $\forall x \in [0, 1]: x^{t+1} \leq x^{t}$), thus by linearity of the integral, $C_{i}$ is monotonically decreasing.

Now, as we have shown that $C_{i}$ is monotonically decreasing, and for all $t$, the cumulative weight of $C_{1:t}$ in $\ITVn{T}$ always exceeds the corresponding weight in $\ITVn{T+1}$, we may conclude
\[
\ITVn{T+1} - \ITVn{T} \leq 0 \enspace,
\]
which implies \eqref{eq:tvtinc}.

\bigskip

We now show \eqref{eq:avar-lb}.

The Bernstein's inequality applied to  $f_{\rm avg}$ on blocks of size $\TRel$, %implies a sample complexity of
imply that (see \cref{thm:bernstein})
\[
m \leq 
m_{1}%(???)
= \ceil{\TRel(\M^{(T)})}T\left(\frac{10\frange}{\varepsilon} + \frac{4\ITVar}{\varepsilon^{2}}\right) \ln \frac{2}{\delta} 
 \leq C_{1}T\left(\frac{10\frange}{\varepsilon} + \frac{4\ITVar}{\varepsilon^{2}}\right) \ln \frac{2}{\delta} \enspace, %\enspace.
\]
where $C_{1}$ is the (constant) upper-bound to %the relaxation time of $\M^{(T)}$,
$\ceil{\TRel(\M^{(T)})}$,
are \emph{sufficient} to $\varepsilon$-$\delta$ estimate $\Expect[f]$ with the empirical mean.

The Markov-chain central limit theorem, together with standard Gaussian anticoncentration inequalities, imply that 
\[
m \geq m_{2} %(\dots)
=
\frac{C_{2}\AVar}{\varepsilon^{2}} \ln \frac{1}{\delta} \enspace,
\]
where $C_{2} \geq 1$ %$0 < C_{2} \leq 1$ 
is an absolute constant, %(for sufficiently small $\delta$?),
are \emph{necessary} (asymptotically) to $\varepsilon$-$\delta$ estimate $\Expect[f]$ with the empirical mean.

We thus have that, asymptotically, $m_{1} \geq m \geq m_{2}$, from which we derive
\[
1 \leq \lim_{\varepsilon \to 0} \lim_{\delta \to 0} \frac{m_{1}}{m_{2}} = \frac{4C_{1}T\ITVar}{C_{2}\AVar} \implies \frac{C_{2}}{4C_{1}} \AVar \leq T\ITVar \enspace,
\]
which concludes the proof.

\end{proof}

Rabinovich et al. \cite{functionMixingRabinovish} define function specific mixing and relaxation times using the spectral decomposition of transition matrix of $\M$. We restate this definition here and using it we extend \cref{lemma:tvar-prop}.
\begin{defin}[Function specific relaxation time, $\TfRel(\M,f)$]\label{def:trefmix}
Given  $f$ and Markov chain $\M$ defined on $\X$, Let $P$ be the transition matrix of $\M$ and consider all eigenvectors of $P$ which are orthogonal to $f$. Let $\lambda_1=1,\lambda_2,\lambda_3,\dots ,\lambda_k$ be the eigenvalues corresponding to the remaining eigen-vectors. We define $\TfRel(\M,f)=\frac{1}{1-\lambda_*}$, where $\lambda_*=\max_{i=2:k}\{\abs{\lambda_i}\}$.
%\shahrzad{CHeck this}
\end{defin}

It is now difficult to see the following corollary, which extends \cref{lemma:tvar-prop} using the function specific relaxation time.

%\shahrzad{Check with Cyrus}
\begin{coro}\label{coro:fmix}
Suppose as in \cref{lemma:tvar-prop}.
Then
%Suppose a lazy reversible Markov chain $\M$,  the inter-trace variance is nonincreasing as
\begin{equation*}
\SVar \leq 2\ITVn{2} \leq 3\ITVn{3}\leq \dots  \leq \lim_{i\to \infty} i\ITVn{i}\leq (2\TfRel - 1)\SVar \enspace.
\end{equation*}
%
%and $\tau \ITVn{\tau}$ is nondecreasing and bounded as
%

Furthermore, there exists some absolute constant $c > 0$, such that for any $\M$, $f$, we have
\begin{equation*}
\tau \ITVn{\tau} \geq c \tau' \ITVn{\tau'}\geq c \lim_{i \to \infty} i\ITVn{i} \enspace, \quad \text{ for any $\tau' \geq \tau \geq \TfRel(\M)$\enspace.}
\end{equation*}

\if 0 %Not sure about this one
Furthermore, there exists some absolute constant $c > 0$, such that for any $\M$, $f$, we have
\begin{equation*}
\label{eq:avar-lb}
\tau \ITVn{\tau} \geq c \tau' \ITVn{\tau'}\geq c \lim_{i \to \infty} i\ITVn{i} \enspace, \quad \text{ for any $\tau' \geq \tau \geq \TRel(\M)$\enspace.}
\end{equation*}
%Furthermore for $\tau\geq \TRel(\M)$, we have
%\begin{equation}\label{eq:2}\Theta(\TRel \ITRelVar) = \Theta(\tau\ITVn{\tau}) = \Theta\left( \lim_{\tau \to \infty} \tau\ITVn{\tau} \right) \enspace.\end{equation}
\fi
\end{coro}
\begin{proof}
Proof of the first two inequalities is essentially identical to that of \cref{lemma:tvar-prop}, where we use the fact that in the spectral decomposition, eigenvalues orthogonal to $f$ have 0 weight in the integral.
We thus ignore all eigenvalues orthogonal to $f$.

In particular, for reversible $\M$, the spectral decomposition allows us to bound the autocovariance as
\[
\abs{\AutoCov_{t}(f)}
  = \abs{\int \lambda^{|t|} \mathrm{d}E_{f}(\lambda)}
  \leq \EigTwo^{|t|}(f) \SVar
  %\leq \EigTwo^{|t|} \SVar
  \enspace,
\]
where $\EigTwo^{|t|}(f)$ is the second largest absolute eigenvalue not orthogonal to $f$.
From here, all bounds on inter-trace variances follow as before, now using this inequality in place of the standard spectral decomposition inequality for covariances.
\end{proof}

\subsubsection{Bounding trace variance using projection chains: Cycle}\label{sec:cycle}

We now seek to show \cref{lem:cycle}, restated below.
Recall here that for $1\leq i\leq \frac{n}{2} $, we take $f_i:[n]\rightarrow \{0,1\}$ be $f_i(x)=0$ if and only if $x ~ {\rm mod}~ 2i< i$, so $f_i$'s image on the cycle is consecutive length-$i$ runs of $0$s and $1$s (see \cref{fig:cycle}).
\cycle*

%In order to prove \cref{lem:cycle}, we prove the following lemma, and employ  \cref{lemma:tvar-prop} \eqref{eq:tvtinc}.

%\begin{lemma}For $1\leq i \leq \frac{n}{2}$, let $f_i$ be defined as in \cref{lem:cycle}, we have $\ITVar(f_i)\leq \Theta(\nicefrac{i^2}{2})$ \cyrus{$\Theta(\nicefrac{i^2}{T})$?}. %Moreover, for  $T\leq \Theta(i^2)$ we have $\ITVar\geq \Theta(1) $.
%\cyrus{Why this lemma?}\end{lemma}

\begin{proof}[Proof]
We first prove that $\ITVn{\tau}(f_i)\leq \Theta(\nicefrac{i^2}{\tau})$. Note that  $x \mapsto x$ mod $2i$ partitions the set $[n]$ into $2i$ partitions\cyrus{$n / 2i$?} (see \cref{fig:cycle}), and the relaxation time of the projection chain ${\cal C}_{f_i}$ is $\Theta(i^2)$.
Since $f_i$'s image on traces of $f_i$ on $\cal C$ and ${\cal C}_{f_i}$ is distributed identically, for $\tau=\TRel(\Cycle)$ by applying  \cref{lemma:tvar-prop}, we get $\ITVn{\tau}(\Cycle_{f_i},f_{i}) \leq 
\frac{\TRel(\Cycle_{f_i})}{\tau}\SVar$. Replacing the known we get: 
$\ITVn{\tau}(\Cycle,f_{i})=
\ITVn{\tau}(\Cycle_{f_i},f_{i})\leq 
\Theta(\nicefrac{i^2}{n^2})$.
We now proceed to prove that $\ITVn{\tau}(\Cycle,f_{i})=
\ITVn{\tau}(\Cycle_{f_i},f_{i})\geq \Theta(\nicefrac{i^2}{n^2})$\cyrus{Sophia's confusion. Use doubling inequalities?}.

\emph{A Sketch.} Let $S_0=\{i/3+1, i/3+2 ,\dots , 2i/3 \}\subseteq [n]$.  we show that  any trace of length at most $i^2/9$ starting at $S_0$, or any other middle point of other \emph{monochromatic} regions, have low probability of escaping from it 
(see Figure \ref{fig:2cycle}). 
%for visualization of $f_{i}$ on $\mathcal{C}$, and the region $S$.

Using this we show that the $\nicefrac{i^2}{9}$-trace variance conditioned on starting at one of these middle sub-regions is $\Theta(1)$. 
Having a bound on $\ITVn{\nicefrac{i^2}{9}}$, we use \cref{lemma:tvar-prop}  and conclude the premise.  The next paragraph presents proof details. 

%for any $0\leq k\leq \frac{n}{2}i $ we  define  $S_k$  to be $S_k\dot{=}\{k(2i)+(i/3+1), k(2i)+(i/3+2),\dots , k(2i)+2i/3\}$\cyrus{Why $k2i$ instead of $ki$?}, and using a similar argument we can show that any walk starting at $S=\bigcup_{k=0}^{n/2i}S_k$ has trace variance at least $0.03822$.
%Since the stationary probability of $S$ is $\frac{1}{3}$, we conclude $\ITVn{i^2}(\Cycle, f_i)\geq \Theta(1)$. 

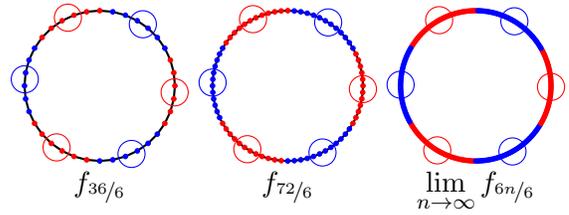
\begin{SCfigure}[1.9][ht]

\centering 
\newcommand{\plotcycle}[4]{
\def\k{#1}
\def\kh{#2}
\def\jn{#3}
\def\n{#4}
\foreach \i[evaluate=\i as \j using {\i + 1}] in {1,2,...,\n} {
 \draw[thick] ({sin(\i / \n * 360)},{cos(\i / \n * 360)}) --  ({sin(\j / \n * 360)},{cos(\j / \n * 360)});
}
\foreach \i in {1,2,...,\kh} {
  \foreach \j in {1,2,...,\jn} {
    \filldraw[blue] ({(sin(((\i * 2 + 0) * \jn + \j) / \n * 360)},{(cos(((\i * 2 + 0) * \jn + \j) / \n * 360)}) circle (0.03cm);
    \filldraw[red]  ({(sin(((\i * 2 + 1) * \jn + \j) / \n * 360)},{(cos(((\i * 2 + 1) * \jn + \j) / \n * 360)}) circle (0.03cm);
  }
  \draw[color=blue]%[rotate={(\i * 2 + 1.5) / \k * 360}] 
  ({(sin((\i * 2 + 0 + 0.5 * (1 + 1 / \jn)) / \k * 360)},{(cos((\i * 2 + 0 + 0.5 * (1 + 1 / \jn)) / \k * 360)}) circle (0.18cm); 
  \draw[color=red]%[rotate={(\i * 2 + 1.5) / \k * 360}] 
  ({(sin((\i * 2 + 1 + 0.5 * (1 + 1 / \jn)) / \k * 360)},{(cos((\i * 2 + 1 + 0.5 * (1 + 1 / \jn)) / \k * 360)}) circle (0.18cm); 

}
}

\begin{tikzpicture}[scale=1.0]

\plotcycle{6}{3}{6}{36}
\node[anchor=north] (label) at (0, -1) {$f_{\nicefrac{36}{6}}$};

\begin{scope}[xshift=2.5cm]

\plotcycle{6}{3}{12}{72}
\node[anchor=north] (label) at (0, -1) {$f_{\nicefrac{72}{6}}$};

\end{scope}

\begin{scope}[xshift=5cm]

\plotcycle{6}{3}{36}{216}
\node[anchor=north] (label) at (0, -1) {$\lim\limits_{n \to \infty} f_{\nicefrac{6n}{6}}$};

\end{scope}

\end{tikzpicture}

\caption{
Image of $f_{\nicefrac{n}{6}}$ on a cycle for various values of $n$, zero values are colored in red, and one values in blue.
The trace variance for trace length shorter than $(n/54)^2$, conditioned starting at $S=\bigcup_{k=1}^6 S_k$ (circled regions), is at least constant.
}\label{fig:2cycle}
\end{SCfigure}

\newcommand{\CTS}{\nicefrac{i^{2}}{9}}

\emph{Proving $  \ITVn{\CTS}(\Cycle, f_i)\geq \Theta(1)$}.
We assume WLOG that $i$ divides $3$: this is so we may analyze an convenient integer-length trace, but the remaining cases hold with constant-factor differences.
Note that $\TRel \in \LandauTheta(n^{2})$, and our proof operates by analyzing $I$-traces, for $I \doteq \CTS$, and then applying trace-variance inequalities to draw the desired conclusions.

Assume $X_1,X_2,\dots , X_{i^2}$ is a trace of the unbiased walk on the cycle, and define $Y_k=X_{k+1}-X_k$.
Thus, $X_k=\sum_{j=1}^{k-1}Y_j+s_0$ where $s_0$ %\in S_0$ 
is the starting point. 
%We define $D$ to be the maximum number that has been covered by a walk of at most length $n^2$, i.e. $D \doteq \max_k\{X_k\}$\cyrus{Don't we need $D \doteq \max_k\abs{X_k}$?}. Note that $D-s_0= \max_k\{\sum_{j=1}^{k-1} Y_j\} $. 
Note that $Y_j$ is a symmetric random variable (thus has equal mean and median), thus, we shall use to apply L\'evy's inequality. %, thus the conditions of 
The proof strategy here is to lower-bound $\ITVn{\CTS}$ by showing that a constant fraction of stationary traces in $\Cycle^{(\CTS)}$ see only one color.
We call such traces \emph{homogeneous}, and note that for such traces $\vec{X}$, we have $(\Favg(\vec{X}) - \mu)^{2} = \frac{1}{4}$ (as $\Favg(\vec{X}) \in \{0, 1\}$, and $\mu = \frac{1}{2}$), and from there we bound trace variances as appropriate.

We begin with a key step in deriving a lower-bound on the proportion of such homogeneous traces.
In particular, take $S_0 \doteq \{ \nicefrac{i}{3}+1, \nicefrac{i}{3}+2, \dots, \nicefrac{2i}{3} \} \subseteq [n]$, and similarly take $S_{k} \doteq \{ ki + \nicefrac{i}{3}+1, ki +  \nicefrac{i}{3}+2, \dots, ki + \nicefrac{2i}{3} \} \subseteq [n]$, in other words each $S_{k}$ is the \emph{middle third} of the $k$th contiguous color region.
Finally, take $S \doteq \bigcup_{k=1}^{n/i} S_k$.
These regions are depicted graphically in \cref{fig:2cycle}.

%TODO: DUPLICATE FIGURE!

Let $\displaystyle \SD(\vec{Y}) \doteq \max_{k \in 1, \dots, \CTS} \abs{ \sum_{j=1}^{I} Y_{k}}$, i.e., $\displaystyle \SD(\vec{Y}) = \max_{k \in 1, \dots, \CTS} \Delta_{\circ}(X_{1}, X_{k})$, for $\Delta_{\circ}$ the shortest-path distance on the 
cycle, and observe
\begin{align*}
\Prob \left( \SD(\vec{Y}) \geq \frac{i}{3} \right) &\leq 2\Prob \left( \abs{ \sum_{j=1}^{\CTS} Y_j } \geq \frac{i}{3} \right)  & \textsc{L\'evy's inequality} \\
    &\leq 4\exp\left(\frac{-2 \left(\frac{i}{3}\right)^{2}}{ {i^{2}}/{9} }\right) & \textsc{Hoeffding's inequality} \\
    &= 4\exp(-2) \enspace. & \\
\end{align*}

%Now let $\displaystyle \SD(\vec{Y}) \doteq \max_{k \in 1, \dots, \CTS} \abs{ \sum_{j=1}^{I} Y_{k}}$, i.e., $\displaystyle \SD(\vec{Y}) = \max_{k \in 1, \dots, \CTS} \Delta_{\circ}(X_{1}, X_{k})$, for $\Delta_{\circ}$ the shortest-path distance on the cycle.

From here, we decompose the trace variance and bound it as
\begin{align*}
\ITVn{I} &= \Expect_{\vec{Y}}
            \left( \frac{1}{\CTS}\sum_{j=1}^{\CTS} Y_{j} - \mu \right)^{2} & \textsc{By Definition } \\
 %&\geq \Prob(X_{1} \in S) \Expect_{\vec{Y} \mid X_{1} \in S} \left( \frac{1}{\CTS}\sum_{j=1}^{\CTS} Y_{j} - \mu \right)^{2} & \textsc{Law of Total Expectation} \\
 &\geq \Prob(X_{1} \in S) \Prob\!\left( \SD(\vec{Y}) \leq \frac{i}{3} \middle| X_{1} \in S \right) %  f(X_{1}) = f(X_{2}) = \dots = f(X_{\CTS}) \mid X_{1} \in S\right) 
 \Expect_{\vec{Y} \mid \SD(\vec{Y}) \leq \frac{i}{3}} \left( \frac{1}{\CTS}\sum_{j=1}^{\CTS} Y_{j} - \mu \right)^{2} & \textsc{Law of Total Expectation} \\
 &\geq \frac{1}{3} \bigl(1 - 4 \exp(-2) \bigr) \frac{1}{4} > 0.03822 \enspace. & \begin{tabular}{r} $\Prob(X_{1} \in S) = \frac{1}{3}$ \\ \textsc{See Above} \\ $(0 - \mathsmaller{\frac{1}{2}})^{2} = (1 - \mathsmaller{\frac{1}{2}})^{2} = \mathsmaller{\frac{1}{4}}$ \\ \end{tabular} \!\!\! \\
\end{align*}

We thus conclude $\ITVn{I} \in \Theta(1)$, therefore via \cref{lemma:tvar-prop}
%via \cref{lem:varratios},
$\ITRelVar \in \LandauTheta(\nicefrac{i^{2}}{n^2})$.

\end{proof}

\subsection{An unbiased variance estimator}\label{sec:varestimatorproof}

 Recall from \cref{def:varest} that 

\[
\VED(\F,\M)\doteq \frac{1}{2m}\sum_{i=1}^m \left(\F(X_{1,i})-\F(X_{2,i})\right)^2  \enspace,
\]

\begin{proof}[Proof of lemma \ref{lemma:twochain}]
We first note that by the tensor product chain rule (see Ex.~12.6 of \cite{levin2017markov}), the \emph{spectral gap} $\SG$ of $\M \otimes \M$ equals that of $\M$.
We now define the function $g: \X \times \X \mapsto [0, \mathsmaller{\frac{1}{2}}\frange^{2}]$ as $g(x_1, x_2) \doteq \frac{1}{2}(x_{1} - x_{2})^{2}$.

We first show that $g$ is an \emph{unbiased estimator} of the variance of $f$, i.e., $\Expect_{\M \times \M}[g] = \mathbb{V}_{\pi}[\F]$:
\begin{align*}
\Expect[g] &= \Expect\left[\frac{1}{2}(X_{1} - X_{2})^{2}\right] & \textsc{Definition of $g$} \\
 %&= \frac{1}{2m}\sum_{i=1}^{m} \Expect\left[(X_{1,i} - X_{2,i})^{2}\right] & \textsc{Linearity} \\
%  &= \frac{1}{2} \Expect\left[(X_{1} - X_{2})^{2}\right] & \textsc{Stationarity} \\
 &= \frac{1}{2} \left(\Expect[X_{1}^{2}] + \Expect[X_{2}^{2}] - 2\Expect[X_{1}X_{2}] \right) & \textsc{Linearity} \\
 &= \Expect[X_{1}^2] - (\Expect[X_{1}])^{2} & \textsc{Independence} \\
 &= \Var[X_{1}] = \Var_{\pi}[\F] & \textsc{Variance Properties} \\
\end{align*}

Note that $\VED(\vec{X}_1,\vec{X}_2)=\frac{1}{m}\sum_{i=1}^m\left(g(X_{1,i},X_{2,i})\right)$, thus it is immediate that $\Expect[\VED(\vec{X}_1,\vec{X}_2)]=\SVar$.

\smallskip

We now seek to apply the Bernstein bound to $g$ on $\M \otimes \M$.
Note that the \emph{range} of $g$ is $[0,\mathsmaller{\frac{1}{2}}\frange^{2}]$.
We require also a bound on the \emph{variance of the variance} $\Var[g]$.
We break the infinite regress here by noting that
\begin{align*}
\Var[g] &= \Expect[g^{2}] - (\Expect[g])^{2} & \\
 &\leq \Expect[g \cdot g] & \\
 &\leq \Expect[g \mathsmaller{\frac{1}{2}}\frange^{2}] \\
 &= \mathsmaller{\frac{1}{2}}\frange^{2}\Expect[g] \enspace. \\
 &= \mathsmaller{\frac{1}{2}}\frange^{2} \Var[\F]\enspace. \\
\end{align*}
%\cyrus{Can improve to $(\mathsmaller{\frac{1}{2}}\frange^{2} - \SVar)\SVar$?}
Note that this is effectively the argument of Bernstein's inequality that removes higher moments (beyond the second) from the exponential-sum decomposition in the Chernoff-MGF bound.

Applying the Bernstein inequality (\cref{thm:bernstein}) %\cite{jiang2018bernstein}
then yields

\[
\Prob\left(\abs{ \Var[\F] -\VED }\geq  \varepsilon \right) \leq \delta \enspace,
\]
for
\begin{align*}
 \varepsilon &\leq \frac{10\frange^{2}\ln \frac{1}{\delta}}{2(1-\EigTwo)m} + \sqrt{\frac{2(1+\EigTwo) \Var[g]\ln \frac{1}{\delta}}{(1-\EigTwo)m}} & \operatorname{Range}(g) = \mathsmaller{\frac{1}{2}}\frange^{2} \\
 &\leq \frac{5\frange^{2}\ln \frac{1}{\delta}}{(1-\EigTwo)m} + \sqrt{\frac{(1+\EigTwo) \frange^{2}\Var[\F]\ln \frac{1}{\delta}}{(1-\EigTwo)m}} & \Var[g] \leq \mathsmaller{\frac{1}{2}}\frange^{2}\Var[\F] \\
\end{align*}
We now seek a form that depends only on the \emph{empirical variance}, which we derive via the quadratic formula.
\begin{align*}
 \varepsilon %&= & \textsc{Let \dots} \\
 %&= & \textsc{Quadratic Formuala} \\
 &\leq \frac{(11 + \EigTwo)\frange^{2} \ln \frac{1}{\delta}}{2(1 - \EigTwo)m} + \sqrt{ \!\! \left(\frac{(1 + \EigTwo)\frange^{2}\ln \frac{1}{\delta}}{2(1 - \EigTwo)m}\right)^{2} \!\! + \frac{5\frange^{2}\ln \frac{1}{\delta}}{(1-\EigTwo)m} \cdot \frac{(1 + \EigTwo)\frange^{2}\ln \frac{1}{\delta}}{(1 - \EigTwo)m} + \frac{(1 + \EigTwo)\frange^{2} \VED\ln \frac{1}{\delta}}{(1 - \EigTwo)m}} \hspace{-1.76cm} & \begin{tabular}{r}\textsc{Quadratic} \\ \textsc{Formula} \\ \end{tabular} \\ % \textsc{Substitution} \\
 &= \frac{(11 + \EigTwo)\frange^{2} \ln \frac{1}{\delta}}{2(1 - \EigTwo)m} + \sqrt{ \!\! \left(\frac{(1 + \EigTwo)\frange^{2}\ln \frac{1}{\delta}}{2(1 - \EigTwo)m}\right)^{2} \!\! + \left(\frac{\sqrt{5(1+\EigTwo)}\frange^{2} \ln \frac{1}{\delta}}{(1 - \EigTwo)m}\right)^{2} + \frac{(1 + \EigTwo)\frange^{2} \VED\ln \frac{1}{\delta}}{(1 - \EigTwo)m}}  \hspace{-2cm} & \textsc{Algebra} \\
 &= \frac{(11 + \EigTwo)\frange^{2} \ln \frac{1}{\delta}}{2(1 - \EigTwo)m} + \sqrt{((2\sqrt{5})^{2} + (1 + \EigTwo))\! \left(\!\frac{\sqrt{1 + \EigTwo}\frange^{2}\ln \frac{1}{\delta}}{2(1 - \EigTwo)m}\!\right)^{\!2} \!\! + \frac{(1 + \EigTwo)\frange^{2} \VED\ln \frac{1}{\delta}}{(1 - \EigTwo)m}}  \hspace{-1.5cm} & ((2\sqrt{5})^{2} + (1 + \EigTwo)) = 21 + \EigTwo \\
 &\leq \frac{(11 + \EigTwo)\frange^{2} \ln \frac{1}{\delta}}{2(1 - \EigTwo)m} + \sqrt{ \left(\frac{(\sqrt{21} + \frac{11}{\sqrt{21}}\EigTwo)\frange^{2}\ln \frac{1}{\delta}}{2(1 - \EigTwo)m}\right)^{2} + \frac{(1 + \EigTwo)\frange^{2} \VED\ln \frac{1}{\delta}}{(1 - \EigTwo)m}}  \hspace{-1.5cm} & \hspace{-1cm} \sqrt{21 + \EigTwo}\sqrt{1 + \EigTwo} \leq \sqrt{21} + \frac{11}{\sqrt{21}}\EigTwo & \\
 &\leq \frac{(11 + \sqrt{21} + (1 + \mathsmaller{\frac{11}{\sqrt{21}}}) \EigTwo)\frange^{2} \ln \frac{1}{\delta}}{2(1 - \EigTwo)m} + \sqrt{\frac{(1 + \EigTwo)\frange^{2} \VED\ln \frac{1}{\delta}}{(1 - \EigTwo)m}}  \hspace{-1.5cm} & \sqrt{a^{2} + b} \leq a + \sqrt{b} \\
 &= \frac{((11 + \sqrt{21})(1 + \mathsmaller{\frac{\EigTwo}{\sqrt{21}}}) \frange^{2} \ln \frac{1}{\delta}}{2(1 - \EigTwo)m} + \sqrt{\frac{(1 + \EigTwo)\frange^{2} \VED\ln \frac{1}{\delta}}{(1 - \EigTwo)m}} \enspace. \hspace{-1.5cm} & \textsc{Algebra} \\
 %&< \frac{(11+\sqrt{21})(1 + \EigTwo) \frange^{2} \ln \frac{1}{\delta}}{2(1 - \EigTwo)m} + \sqrt{\frac{(1 + \EigTwo)\frange^{2} \hat{v}\ln \frac{1}{\delta}}{(1 - \EigTwo)m}}  \hspace{-1.5cm} & \sqrt{(2\sqrt{5})^{2} + (1 + \EigTwo)}\sqrt{1 + \EigTwo} \leq \sqrt{21} + \frac{11}{\sqrt{21}}\EigTwo \enspace& \sqrt{a^{2} + b} \leq a + \sqrt{b} \\
 %Correct but messy?
 %&\leq \frac{(6 + \sqrt{5(1+\EigTwo)} + \lambda)\frange^{2}\ln \frac{1}{\delta}}{(1-\EigTwo)m} + \sqrt{\frac{2(1 + \EigTwo)\frange^{2} \hat{v}\ln \frac{1}{\delta}}{2(1 - \EigTwo)m}} & \sqrt{a^{2} + b} \leq a + \sqrt{b} \\ %& \textsc{More Algebra} \\
 %&< \frac{(7 + \sqrt{10})\frange^{2}\ln \frac{1}{\delta}}{(1-\EigTwo)m} + \sqrt{\frac{(1 + \EigTwo)\frange^{2} \hat{v}\ln \frac{1}{\delta}}{(1 - \EigTwo)m}} \enspace. & \EigTwo \in [0, 1]
 \end{align*}

From which we derive

\[
\Prob\left( \abs{\Var[\F]-\VED} \geq  \frac{(11 + \sqrt{21})(1 +  \mathsmaller{\frac{\EigTwo}{\sqrt{21}}}) \frange^{2} \ln \frac{1}{\delta}}{(1-\EigTwo)m} + \sqrt{\frac{(1 + \EigTwo)\frange^{2} \VED\ln \frac{1}{\delta}}{(1 - \EigTwo)m}} \right) \leq \delta \enspace.
\]
\end{proof}

\subsection{Analysis  of  \algo{}}\label{sec:proofsalgo}

\paragraph{Prelude to a Proof }
We first note that correctness of \algo{} follows from that of \DMCMC{}, as they simply apply \DMCMC{} to $\Favg$ and $\M^{(T)}$, the trace chain.

Our proof is divided to the following parts: proving trace chain properties, trace variance estimation, and progressive sampling.  

\smallskip 
%and/or apply a \emph{nonstationarity correction}.\todo{ref nonstat correction.}

\emph{The Trace Chain.} 
In \cref{sec:tracechain} we first obtain the stationary distribution of the trace chain and then we bound its mixing and relaxation time in term of the original chain $\M$.

\smallskip 

\emph{Estimating the Inter-Trace Variance.}
In order to estimate the trace variance we use $\VED$ whose correctness is proved in \cref{sec:varestimatorproof}. We employ this estimator to $\M^{(T)}$, which we proved has constant relaxation time for $T\geq \TMixBound$. The termination condition which is used in the progressive sampling is based on the variance estimation and using the Bernstein  bound (\cref{thm:bernstein}) to ensure its accuracy. 

\smallskip 

\emph{Sketch of the Sampling Schedule.} The key to \DMCMC{} is an \emph{a priori} fixed \emph{sampling schedule}, which determines the sizes of \emph{progressively larger samples}. To ensure correctness, a sequence of tail bounds must all hold simultaneously w.h.p.\ by union bound.
The main difficulty %in deriving an efficient schedule is that its
is that the schedule length $\NIterations$ and the probability concentration bounds are codependent.
%Note that there is inherent conflict here:
A shorter schedule is more \emph{statistically efficient}, but can \emph{overshoot} the sufficient sample size, and the opposite holds for longer schedules. %, whereas denser schedules won't overshoot by much, but they weaken all tail bounds due to the union bound.
%We thus seek to derive a short schedule (for statistical efficiency) that nevertheless does not contain large multiplicative gaps.
We explain how to resolve this cyclic dependence
in the next paragraph.

Over a run of \algo{}, we take (up to) $3\NIterations$ probability concentration bounds ($\NIterations$ bounds for variance, line~\ref{alg:mcd-var}, and $\NIterations$ bounds each for upper and lower mean bounds, line~\ref{alg:bern}). %, and all hold simultaneously via union bound.
We first establish the \emph{worst-case} Hoeffding sample complexity $m_{\NIterations}$
(\cref{thm:hoeffding}) and \emph{best-case} Bernstein sample complexity $\alpha$ (\cref{thm:bernstein}) 
by taking $v=0$,
%by taking $\hat{v} = 0$, respectively, of
%\vspace{-0.2cm}
\[
\scalebox{1}{$\displaystyle
m_{I}
  \doteq m_{H}(\EigTwoBound, \frange, \varepsilon, \mathsmaller{\frac{2\delta}{3\NIterations}})
  %= 2(1+\EigTwoBound)\TRel\ln(\mathsmaller{\frac{3\NIterations}{\delta}}) \mathsmaller{\frac{\frange^{2}/4}{\varepsilon^{2}}}
  = \frac{(1 + \EigTwoBound) \frange^{2} \ln \frac{3\NIterations}{\delta}}{2(1 - \EigTwoBound)\varepsilon^{2}}
  \enspace, \ \ \& \ \
 \alpha \doteq \frac{(1 + \EigTwoBound)\frange\ln \frac{3\NIterations}{\delta}}{(1-\EigTwoBound)\varepsilon} \approx m_{B}(\EigTwoBound, \frange, 0, \varepsilon, \mathsmaller{\frac{2\delta}{3\NIterations}})
 \if 0 %OLD:
 \alpha \doteq \frac{5(6 + \EigTwoBound)\frange\ln \frac{3\NIterations}{\delta}}{2(1-\EigTwoBound)\varepsilon} \approx m_{B}(\M, f, v : \hat{v} = 0, \varepsilon, \mathsmaller{\frac{2\delta}{3\NIterations}})
 \fi
 \if 0 %OLDER:
 \alpha \doteq m_{B}(\M, f, v = 0, \varepsilon, \mathsmaller{\frac{2\delta}{3\NIterations}})
 %=2(1+\EigTwoBound)\TRel \ln(\mathsmaller{\frac{3\NIterations}{\delta}})\left(0 + \frac{5\frange}{(1+\EigTwoBound)\varepsilon}\right)
 %=2(1+\EigTwoBound)\TRel \ln(\mathsmaller{\frac{3\NIterations}{\delta}})\mathsmaller{\frac{5\frange}{(1+\EigTwoBound)\varepsilon}}
 = \frac{(17.791 + 1.700 \EigTwoBound)\ln \frac{3\NIterations}{\delta}}{(1-\EigTwoBound)}  \cdot \frac{\frange}{\varepsilon}
 \fi
 \enspace.$}
\]

\cyrus{Is it cleaner to define $m_{H}$ $m_B$ as 1-tail bounds?}
%\vspace{-0.4cm}

%\todo{Should these be equalities?  Is $m_H$ $m_B$ notation clear? I THINK IT IS CLEAR}

We know from these bounds that once a sample of size $m_{\NIterations}$ is drawn, the desired guarantee has been met, %(by Hoeffding),
and similarly, before a sample of size $\null \approx \alpha$ is drawn, the desired Bernstein bound \emph{can not} be met.
%

%
%Note that we have not yet chosen $\NIterations$; we now seek to 
We now select the minimal $\NIterations$ such that each sample size $m_{i}$ obeys $m_{i} \leq 2m_{i-1}$.
Observe that in the ratio

\[
\frac{m_{\NIterations}}{\alpha} = \frac{(1+\EigTwoBound)\frange}{4 \cdot 5 \varepsilon} = \frac{9\frange}{100\varepsilon} \enspace,
\]

all dependence on $\NIterations$ is divided out ($\ln\frac{3\NIterations}{\delta}$ terms cancel).
%in which the $\ln(\frac{3\NIterations}{\delta})$ terms conveniently divide out.
We initially run the chain for $m_{1}$ to be $\lceil2\alpha\rceil$, and in iteration $i$, we run up to $\ceil{2\alpha^{i}}$ steps, thus we conclude 
%$\NIterations = \left\lceil \log_{2}\left( \frac{1 + \EigTwoBound}{35.582 + 3.400\EigTwoBound} \cdot \frac{\frange}{\varepsilon} \right) \right\rceil - 1 \leq \log_{2}\left( \frac{1 + \EigTwoBound}{35.582 + 3.400\EigTwoBound} \cdot \frac{\frange}{\varepsilon} \right)$
%$\NIterations \doteq \floor{\log_{2}\left(\frac{1+\EigTwoBound}{5(6+\EigTwoBound)} \cdot \frac{\frange}{\varepsilon}\right)} \geq  \ceil{\log_{2}\left(\frac{1+\EigTwoBound}{5(6+\EigTwoBound)} \cdot \frac{\frange}{\varepsilon}\right)} - 1$
$\NIterations \doteq \lfloor{\log_{2}(\frac{\frange}{2\varepsilon})\rfloor} \geq  \lceil{\log_{2}( \frac{\frange}{2\varepsilon})\rceil} - 1$
(doubling) iterations are sufficient.\todo{s.t. ... = $m_{H}(\dots)$}
%Most of the details of the schedule selection are not strictly necessary to prove \emph{correctness}, but rather they illustrate the reasoning behind the algorithm, and will prove invaluable in showing the \emph{efficiency} of the schedule.
%We are now ready to show our main results.
These computations are repeated verbatim by \DMCMC{} (lines~\ref{alg:niter} \& \ref{alg:ss}) to compute the sampling schedule.
\todo{Comment on the max with $1$? Also why we use an approximation of $m_{B}$?}

\smallskip 
Full details of the schedule and details of putting the above pieces together are presented in 
\cref{thm:efficiency,thm:correctness}.

\subsubsection{Stationarity, Relaxation, and Mixing Times of the Trace Chain}\label{sec:tracechain}

\begin{lemma}[Trace Chain Mixing]
\label{lemma:tracemixing}
Having a Markov chain $\M$ with stationary distribution  $\pi$ and  mixing time $\TMix$, the trace chain $\M^{(T)}$'s  stationary distribution is $\pi^{(T)}$ where $\pi^{(T)}(\vec{a})\doteq \pi(a_1)\prod_{i=1}^{k-1}\M(a_i,a_{i+1})$ and its mixing time is bounded as $\TMix(\M^{(T)}) \leq 1 + \frac{1}{T}\TMix( \M)$.
In particular for $T\geq\TMix(\M)$, we have that the second largest eigenvalue of $\M^{(T)}$ is at most $\nicefrac{4}{5}$.
\cyrus{Do we use this lemma? And I still think it's missing a ceil.}
\end{lemma}

\begin{proof}[Proof of \Cref{lemma:tracemixing}]

We first remark that
it is not difficult to prove by induction that  \[\left(\M^{(T)}\right)^j(\vec{a},\vec{b})=\M^{(j-1)T+1}(a_T,b_1)\prod_{i=1}^{T-1}\M(b_i,b_{i+1}).\]

It is not difficult to see that with the above definitions, for any $\vec{a}$ a trace of length $T$ we have: $\pi^{(T)}(\vec{a})=\sum_{\vec{b}\in\Omega^{T}}\pi^{(T)}(\vec{a},\vec{b})\M^{(T)}(\vec{a},\vec{b})$. \cyrus{Eli: this is very confusing.}

\medskip 

\cyrus{Type mismatch: $\M^{(T)}$ is a matrix.}
We now show that $\M^{(T)}$ is close to $\pi^{(T)}$ after $(\tau(\epsilon)/T)+1$ steps \cyrus{shouldn't this give $\lceil 1 + \frac{1}{T}\TMix(\M)\rceil$?  We are lacking the ceil.}. Let $X_0, X_1, X_2 \dots$, be the trace of the original random walk,  thus each $X_i$ is  a random variable having distribution $X_i\sim  \M^{i}\mu$.
We partition this trace into blocks of length $T$ as follows: $B^j=(X_j,X_{j+1},\dots X_{j+{T-1}})$. Thus the trace of $\M^{(T)}$ will be $B^0, B^1, B^2,\dots$.

\cyrus{$\tau(\epsilon)$ undefined?}

\medskip

Since the mixing time of $\M$ is $\TMix$, we have: for any starting distribution $\nu$ over $\Omega$, and $\tau \geq \TMix \log(\epsilon)$, $\TVD(\M^{\tau}(\nu),\pi) \leq \epsilon$. Assume $\M^{(T)}$ has started at initial distribution $\nu'$, looking at the distribution of $B_{\frac{\tau(\epsilon)}{T}+1}$ we will have: 

\begin{align*}
  \TVD(B_{\frac{\tau}{T}+1},\pi^{(T)})
  & = \frac{1}{2}\sum_{a\in \Omega^T} \left\lvert  (\M^{(T)})^{\frac{\tau}{T}+1} \nu' (a)-\pi^{(T)}(a) \right\rvert \\
  & = \frac{1}{2}\sum_{a\in \Omega^T} \left\lvert (\M^{\tau})(\nu'_T,a_1)\prod_{i=1}^{T-1}\M(a_i,a_{i+1}) -\pi^{(T)}(a) \right\rvert\\
  &=\frac{1}{2} \sum_{a\in \Omega^T} \left\lvert (\M^{\tau})(\nu'_T,a_1)\prod_{i=1}^{T-1}\M(a_i,a_{i+1}) - \pi(a_1)\prod_{i=1}^{T-1}\M(a_i,a_{i+1})\right\rvert \\
      &\leq \epsilon \prod_{i=1}^{T-1}\M(a_i,a_{i+1})\leq \epsilon \enspace.
\end{align*}
The bound on $\EigTwo$ then follows since $
\left(\TRel(\M)-1\right)\ln(2)\leq \TMix(\M)\leq \TRel(\M) \ \ln\Bigl(\smash{\frac{2}{\sqrt{\pi_{\min}}}}\Bigr) \enspace.$
%Thus, using the formula relating mixing time and relaxation time we get the result.
%
\if 0
\cyrus{This is from Eq.~\ref{eq:mix-relax}?  This seems pretty loose, I get the following:}
{\color{green!50!black}
\[
\left(\TRel(\M)-1\right)\ln(2)\leq \TMix(\M)\leq \TRel(\M) \ \ln(\frac{2}{\sqrt{\pi_{\min}}}) \enspace,
\]

\[
\left(\TRel(\M)-1\right) = \frac{1}{\SG} - 1 = \frac{\EigTwo}{1-\EigTwo}
\]

\[
\frac{\EigTwo}{1-\EigTwo} \leq \TMix/\ln(2) \implies \EigTwo \leq \frac{\TMix/\ln(2)}{1 + \TMix/\ln(2)} = \frac{\TMix}{\ln(2) + \TMix} = \frac{2}{\ln(2) + 2} < 0.743 < \frac{3}{4}
\]

Using $\EigTwo = 3/4$ gets us an $\frac{m}{m-7}$ Bessel correction.

Furthermore, with $\EigTwo \leq \frac{\TMix}{\ln(2) + \TMix}$, we get
\[
\frac{1+\EigTwo}{1-\EigTwo} \leq \frac{1 + \frac{\TMix}{\ln(2) + \TMix}}{1 - \frac{\TMix}{\ln(2) + \TMix}} = \frac{\ln(2) + 2\TMix}{\ln(2)} = 1 + \frac{4}{\ln(2)} < 6.771 \enspace.
\]
}
\fi
%
\end{proof}

\begin{lemma}[Relaxation Times of Trace Chains]
\label{lemma:mt-relax}
Suppose a reversible chain $\M$.
Then for any $T \in \N$, we have\cyrus{If we drop supremum, we basically show that all eigenvalues are powered}
\[
\EigTwo(\M^{(T)}) = \EigTwo(\M^{T}) = \EigTwo^{T}(\M) \Leftrightarrow \TRel(\M^{(T)}) = \TRel(\M)^{T} \enspace.
\]
\end{lemma}
\begin{proof}
%$T$\textordmasculine or $T$\textsuperscript{th}?
We show this result in three movements.
We first introduce the $T$\textsuperscript{th}-root $\sqrt[T]{\M^{(T)}}$ of the trace chain $\M^{(T)}$.
We then relate eigenvalues of $\sqrt[T]{\M^{(T)}}$ to eigenvalues of $\M^{(T)}$, which illuminates the proof structure.
Finally, we show that the second eigenvalue of $\sqrt[T]{\M^{(T)}}$ matches that of $\M$.

%eigenvalues (in particular the second eigenvalue) of $\sqrt[T]{\M^{(T)}}$ match those of $\M$.

\bigskip
\textbf{First Movement:} We introduce the $T$\textsuperscript{th}-root of the $T$-trace chain, i.e., $\sqrt[T]{\M^{(T)}}$, which corresponds to 1-step \emph{overlapping windows} of length $T$ on the chain $\M$ (as opposed to the nonoverlapping windows of $\M^{(T)}$).
In particular, $\sqrt[T]{\M^{(T)}}$ has transition matrix
\[
\sqrt[T]{\M^{(T)}}(\vec{a}, \vec{b}) = 
%\begin{cases}
\left\{ \begin{array}{ccc}
\vec{a}_{2:T} = \vec{b}_{1:T-1} & : & \M(a_{T}, b_{T}) \\
\text{otherwise} & : & 0 %\\
\end{array}
\right.
%\end{cases}
\enspace.
% = \1_{\vec{a}_{2:T}}(\vec{b}_{1:T-1}) \M(a_{T}, b_{T})
\]
Since $b_{T-1} = a_{T}$ for all \emph{nonzero transition probabilities}, we may alternatively express these transition probabilities as
\[
\sqrt[T]{\M^{(T)}}(\vec{a}, \vec{b}) = \1_{\vec{a}_{2:T}}(\vec{b}_{1:T-1}) \M(b_{T-1}, b_{T}) = \1_{\vec{a}_{2:T}}(\vec{b}_{1:T-1}) \M(a_{T}, b_{T})
\]
Here transitions between \emph{incompatible} $\vec{a}$ and $\vec{b}$, i.e., those that don't overlap in a $T-1$ window, thus $\vec{a}_{2:T} \neq \vec{b}_{1:T-1}$, \emph{can not} occur.
Therefore, the transition matrix is always extremely sparse; only $\abs{\X}$ of the total $\abs{\vphantom{\X}\smash{\X^{T}}}$ possible states are ever 1-step reachable.

Note that, in general, matrix roots may be non-unique, and stochastic matrices may not have roots that are also stochastic matrices.
However, the $T$\textsuperscript{th}-root of $\M^{(T)}$ described above always exists and is convenient for the analysis, though we do not claim it is unique.

\bigskip

\textbf{Second Movement:} We now relate the eigenvalues of $\sqrt[T]{\M^{(T)}}$ and $\M^{(T)}$.

First, note that both $\sqrt[T]{\M^{(T)}}$ and $\M^{(T)}$ are chains over state space $\X^{T}$.
Now, note that from the transition matrix, it is clear that $(\sqrt[T]{\M^{(T)}})^{T} = \M^{(T)}$, and thus $\sqrt[T]{\M^{(T)}}$ has the same stationary distribution as $\M^{(T)}$, %i.e, $\pi^{(T)}$,
i.e., $\pi(\sqrt[T]{\M^{(T)}}) = \pi(\M^{(T)}) = \pi^{(T)}$.

We now observe that $\EigTwo(\M^{(T)}) = \EigTwo^{T}(\sqrt[T]{\M^{(T)}})$, as is the case with any chain derived via transition-matrix powering (since powering is \emph{repeated multiplication}, it preserves eigenvectors, and powers eigenvalues).
It thus suffices to show that $\EigTwo(\sqrt[T]{\M^{(T)}}) = \EigTwo(\M)$, as this immediately implies $\EigTwo(\M^{(T)}) = \EigTwo(\M^{T}) = \EigTwo^{T}(\M)$.

\bigskip

\textbf{Third Movement:} Here we show that $\EigTwo(\sqrt[T]{\M^{(T)}}) = \EigTwo(\M)$.
This movement is rather more subtle than the prior two, and it is shown via direct computation of the second absolute eigenvalue.

In particular, we decompose the second absolute eigenvalue of $\sqrt[T]{\M^{(T)}}$, i.e.,  $\EigTwo(\sqrt[T]{\M^{(T)}})$, into a supremum over vector-matrix-vector products (the \emph{quadratic form} eigenvalue characterization).
%and the identities of the transition matrix determine that it equals that of $\M$.
As we %are interested in the \emph{second} eigenvalue
need only consider the \emph{second eigenvector}, and the first is always the stationary distribution, we consider in the supremum only the unit eigenvectors that are \emph{orthogonal to} $\pi^{(T)}$, written $\vec{x} \bot \pi^{(T)}$.

Note that here we take $\vec{x} \in \R^{\abs{\vphantom{\X}\smash{\X^{T}}}}$ to denote a \emph{vector} over states in $\X^{T}$, and 
%$s \in 1, \dots, \abs{\X^{T}}$
$s \in \X^{T}$ denotes a \emph{single state} in the state space of $\sqrt[T]{\M^{(T)}}$ (i.e., $T$-traces over $\X$), which is used to index into $x$.
Likewise, w.r.t. $\M$ and $\X$, we use $\vec{x} \in \R^{\abs{\X}}$ and $s \in \X$. %$s \in 1, \dots, \abs{\X}$.
Observe now that
%Using "s \in \X^{T}" and not "s \in 1, \dots, \abs{\X^{T}}":
\newcommand{\transpose}{\top}
\newcommand{\supargtracechain}{\vec{x} \bot \pi^{(T)} : \, \norm{\vec{x}}_{2} = 1}
\newcommand{\supargbasechain}{\vec{x} \bot \pi, \norm{x}_{2} = 1}
\begin{align*}
\hspace{-0.2cm}\EigTwo\bigl(\!\sqrt[T]{\!\M^{(T)}}\bigr) %&= \sup_{\vec{x} \in \R^{\abs{\vphantom{\X}\smash{\X^{T}}}} : \vec{x} \bot \pi^{(T)}} \frac{x \sqrt[T]{\M^{(T)}}x}{x \cdot x} & \\
    %\hspace{-1cm} & \hspace{1cm}
    &= \hspace{-0.3cm} \sup_{\supargtracechain} \vec{x} \sqrt[T]{\M^{(T)}} \vec{x}^{\transpose} & \textsc{Eigenvalue Characterization} \\
    %&\leq \sup_{\supargtracechain} x \sqrt[T]{\M^{(T)}}x & (\pi^{(T)})_{T} = \pi(\M) \\
    %&\leq \sup_{\supargtracechain} x \sqrt[T]{\M^{(T)}} x & \norm{x_{T}}_{2} \leq \norm{x}_{2} \\
    &= \hspace{-0.3cm} \sup_{\supargtracechain} \sum_{s \in \X^{T}} x_{s} \bigl(\sqrt[T]{\M^{(T)}}(\cdot, s) \cdot \vec{x} \bigr) & \textsc{Matrix Multiplication} \\
    &= \hspace{-0.3cm} \sup_{\supargtracechain} \sum_{s \in \X^{T}} x_{s} \sum_{s' \in \X^{T}} x_{s'} \sqrt[T]{\M^{(T)}}(s', s) & \textsc{Dot Product} \\
    %&= \sup_{\supargtracechain} \sum_{s \in \X^{T}} x_{s} \sum_{s' \in \X^{T} | s'_{2:T} = s_{1:T-1}} x_{s'} \sqrt[T]{\M^{(T)}}(s', s) & \textsc{$\sqrt[T]{\M^{(T)}}$ Properties} \\
    &= \hspace{-0.3cm} \sup_{\supargtracechain} \sum_{s \in \X^{T}} x_{s} \hspace{-0.2cm} \sum_{s' \in \X^{T} | s'_{2:T} = s_{1:T-1}} \hspace{-0.65cm} x_{s'} \M(s'_{T}, s_{T}) & \textsc{$\sqrt[T]{\!\M^{(T)}}(\vec{a}, \vec{b}) = \1_{\vec{a}_{2:T}}(\vec{b}_{1:T-1}) \M(a_{T}, b_{T})$} \\
    &= \hspace{-0.3cm} \ \, \sup_{\supargbasechain} \sum_{s \in \X} x_{s} \sum_{s' \in \X} x_{s'} \M(s', s) & \textsc{Summation Consolidation} \\
    &= \lambda(\M) \enspace. & \textsc{Eigenvalue Characterization} \\
\end{align*}
\end{proof}

\subsubsection{Correctness and efficiency proof}

\todo{Use $\EigTwoBound$.}

\correctness*
\begin{proof}
We first show (1), i.e., correctness of \DMCMC, and then show that correctness of \nsalgo{} easily follows from that of \algo{}, which follows from that of \DMCMC.

\todo{Simpler schedule without Hoeffding: see URL %\url{https://www.wolframalpha.com/input/?i=%2810+R%2Fx+%2B+2+*+R%5E2%2F4+%2F+x%5E2%29+%2F+%2810R%2Fx%29+}. 
Get something like $\log_{2}(\frac{\frange}{20\varepsilon} + 1)$.}
%\todo{Consider $I=1$ special case.}

%\cyrus{Mention that applying to trace chain, where $\TRel$, $\EigTwo$ are constant.}

We now show claim (1).
First note that in initialization (independent of any sampling), \DMCMC\  computes \emph{iteration count} $\NIterations$ and \emph{initial sample size} $\alpha$ (line \ref{alg:niter}), which determine the \emph{schedule} of \emph{sample sizes} and \emph{probabilistic bounds}.
Over the course of the algorithm, at each of $\NIterations$ timesteps, a $1$-tail (upper) bound on \emph{variance} $\VBoundVarUpper_{i}$ is computed (line~\ref{alg:2chain-var}), and a $2$-tail bound on \emph{mean} $\VBoundMean_{i}$ is computed (line~\ref{alg:bern}), for a total of $3\NIterations$ tail bounds.
Each tail is bounded with probability $1 - \frac{\delta}{3\NIterations}$, thus by union bound, all hold simultaneously with probability $1 - \delta$.
We assume henceforth that app tail bounds hold, thus all conclusions are thus qualified as holding \emph{with probability $\geq 1 - \delta$}.
Now, note that termination (line \ref{alg:tc}) occurs for one of two reasons: either $i = \NIterations$, or $\hat{\bm{\epsilon}}_{i} \leq \epsilon$.  We analyze these cases separately.

In case 1, we have termination at $i < \NIterations$, i.e., $\hat{\bm{\epsilon}}_{i} \leq \varepsilon$.
As assumed above, all tail bounds at each iteration hold by union bound.
Thus for the sample drawn at iteration $i$, in particular the \emph{variance bounds} (line \ref{alg:mcd-var}) hold, as $\hat{\bm{v}}_{i}$ is an \emph{unbiased estimate} of $\ITVar$, and by \cref{lemma:twochain}, w.h.p., $\ITVar \leq \VBoundVarUpper_{i}$.
Similarly, the \emph{mean bounds} (\ref{alg:bern}) hold via \cref{thm:bernstein} (noting that, by \emph{averaging over} a pair of \emph{independent chains}, the variance proxy of interest is $\mathsmaller{\frac{1}{2}}\ITVar$).
Note that while both tail bounds are taken over the tensor-product chain $\M \otimes \M$, it holds that $\EigTwo(\M \otimes \M) = \EigTwo(\M)$\todo{cite}, so the bound remains valid.

Consequently, when it holds that $\hat{\bm{\epsilon}}_{i} \leq \varepsilon$, the algorithm returns the estimate $\VEMean_{i}$ (line~\ref{alg:return}), which by the above is sufficiently accurate to satisfy the stated guarantees.

We now consider case 2, wherein we have termination at step $i = \NIterations$.
When this occurs, it holds that
\[
m_{i} = m_{\NIterations}
 = \left\lceil \alpha 2^{\NIterations} \right\rceil
 %TODO: fix constants in intermediaries
 %\geq \left\lceil 2\TRel \ln(\mathsmaller{\frac{3\NIterations}{\delta}})\frac{5\frange}{\varepsilon} 2^{\log_{2}(\frac{9\frange}{100\varepsilon})} \right\rceil
 %= \left\lceil 2\TRel \ln(\mathsmaller{\frac{3\NIterations}{\delta}})\frac{9\frange^{2}}{4\varepsilon^{2}} \right\rceil
 \geq m_{H}(\EigTwo, \frange, \varepsilon, \mathsmaller{\frac{\delta}{3I}}) \enspace,
\]
and thus
by Hoeffding's inequality for mixing processes (\cref{thm:hoeffding}), we have $\lvert \mu - \VEMean_{\NIterations} \rvert \leq \varepsilon$, (i.e., the schedule was selected exactly to ensure $m_{\NIterations}$ samples would be sufficient, regardless of early termination and variance.

\medskip

We now proceed to show claim (2).
To see this result, note that
\[
\EigTwo(\M^{(T)}) \leq \EigTwo^{T}(\M) \enspace,\todo{Where do we show this?}
\]
and thus claim (2) follows directly from claim (1).

\medskip

Finally, note that claim (3) follows from claim (2), paired with Eq.~\ref{eq:nonstationary}.
In particular, note that the uniform mixing time is selected such that we have
\todo{
\[
\sup_{\omega} \abs{ \norm{  \frac{\M^{\TUnif}(\omega)}{{\partial \pi}} }_{\pi,\infty} - 1} \leq \frac{\EigTwoBound^{\TUnif}}{\sqrt{\pi(\omega)\PiMin}} \leq \frac{\EigTwoBound^{\TUnif}}{\PiMin}
\]
TODO: HOW GET TO $\pi_{\min}$?
12.13 of \cite{levin2017markov}
}
\cyrus{Notation $\partial \pi(\omega)$?}
\[
\forall \omega: \, \left\lVert \frac{\partial \M^{\TUnif}(\omega)}{\partial \pi} \right\rVert_{\pi, \infty} \leq 2 \enspace,
\]
i.e., $\M$ is uniformly mixed.
%This further implies $\M^{(T)}(X_{0})$ is uniformly mixed (as the probabilities or densities of the initial state and $\PiMin$ decay at the same rate)\cyrus{Need to confirm this.  Where is the derivation?}.
This further implies that the trace chain is uniformly mixed, as both $\PiMin$ and $\EigTwoBound$ are exponential in the trace-length $T$, which cancels out in the uniform-mixing bound (see Eq.~12.13 of \cite{levin2017markov}).
However, we still apply nonstationarity correction $\frac{\delta}{4}$ instead of $\frac{\delta}{2}$, to account for the fact that the \emph{tensor product chain} $(\M \otimes \M)$ \emph{uniformly mixes} slightly slower than $\M$, even though it \emph{relaxes} and \emph{mixes} as quickly as $\M$.
%
%TODO {Proof for trace chains and nonstationarity correction.}
%
\end{proof}

We now present and prove an extended statement of \cref{thm:efficiency}, which provides finite-sample and asymptotic sample complexity bounds to \DMCMC{}, \algo{}.
%, and \nsalgo{}.

\begin{theorem}[Efficiency of \algo]

Suppose as in \cref{thm:correctness},
%Suppose \emph{mixing time upper bound} $T \geq \TMix(1/4)$, \emph{true relaxation time} $\TRel$, and \dots.
and take
$\NIterations \doteq \log_{2}\left( \frac{\frange}{2\varepsilon} \right)$ and
$T \doteq \ceil{\frac{1 + \EigTwoBound}{1 - \EigTwoBound} \ln \sqrt{2}}$.
Then with probability at least $1 - \frac{\delta}{3\NIterations}$, each mean-estimation algorithm runs for no more than $\hat{m}$ steps (individually), where $\hat{m}$ is\cyrus{Use $T$ or $\frac{1}{1 - \EigTwoBound}$ in the asymptotic forms?}
\begin{enumerate}
\item for \DMCMC{}:\hfill

\begin{minipage}[t]{0.5\textwidth}
%Minipage to get equations level with enumerate number.
\vspace{-0.5cm}
\begin{flalign*}
%\text{for \DMCMC{}: } \ 
\hat{m} &\leq %\ln \frac{3\NIterations}{\delta} \left(\frac{130\frange}{2(1-\EigTwoBound)\varepsilon} + \frac{12(1+\EigTwoBound)\ITVar}{(1-\EigTwoBound)\varepsilon^{2}} \right) & \\
   4T\ln \frac{3\NIterations}{\delta} \left(\frac{5(6 + \EigTwoBound)\frange}{2(1-\EigTwoBound)\varepsilon} + \frac{(1+\EigTwoBound)\ITVar}{(1-\EigTwoBound)\varepsilon^{2}} \right) & \\
 &\in \mathcal{O}\left( \frac{1}{1 - \EigTwoBound} \log\left(\frac{\log(\nicefrac{\frange}{\varepsilon})}{\delta}\right)\left( \frac{\frange}{\varepsilon} + \frac{\SVar}{\varepsilon^{2}}\right)\right) \enspace; & \\%[-0.5cm]
 & \null \hspace{\textwidth} \null & \null \\[-0.75cm]
\end{flalign*}
\end{minipage}
\item for \algo{}:\hfill

\begin{minipage}[t]{0.5\textwidth}
\begin{flalign*}
%\text{for \algo{}: } \ 
\hat{m} &\leq %T\ln \frac{3\NIterations}{\delta} \left(\frac{130\frange}{2(1-\EigTwoBound^{T})\varepsilon} + \frac{12(1+\EigTwoBound^{T})\ITVar}{(1-\EigTwoBound^{T})\varepsilon^{2}} \right) & \\
   2T \ln \frac{3\NIterations}{\delta} \left(\frac{65\frange}{\varepsilon} + \frac{12\ITVar}{\varepsilon^{2}} \right) & \\
 &\in \mathcal{O}\left( T\log\left(\frac{\log(\nicefrac{\frange}{\varepsilon})}{\delta}\right)\left( \frac{\frange}{\varepsilon} + \frac{\ITVar}{\varepsilon^{2}}\right)\right) & \\%[-0.5cm]
 &= \mathcal{O}\left( \log\left(\frac{\log(\nicefrac{\frange}{\varepsilon})}{\delta}\right)\left( \frac{\frange}{(1-\EigTwoBound)\varepsilon} + \frac{\TRel\ITRelVar}{\varepsilon^{2}}\right)\right) \enspace; \ \& & \\
 & \null \hspace{\textwidth} \null & \null \\[-0.75cm]
\end{flalign*}
\end{minipage}
\item Adding also warm start complexity (line 4):\hfill

\begin{minipage}[t]{0.5\textwidth}
\vspace{-0.5cm}
\begin{flalign*}
%\text{for \nsalgo{}: } \ 
\hat{m} &\leq %T\ln \frac{3\NIterations}{\delta} \left(\frac{130\frange}{2(1-\EigTwoBound^{T})\varepsilon} + \frac{12(1+\EigTwoBound^{T})\ITVar}{(1-\EigTwoBound^{T})\varepsilon^{2}} \right) & \\
 2\ceil{\frac{\ln \mathsmaller{\frac{1}{\PiMinBound}}}{\ln \mathsmaller{\frac{1}{\EigTwoBound}}}} +
   2T \ln \frac{12\NIterations}{\delta} \left(\frac{65\frange}{\varepsilon} + \frac{12\ITVar}{\varepsilon^{2}} \right) & \\
 &\in \mathcal{O}\left( \frac{\ln \mathsmaller{\frac{1}{\PiMinBound}}}{\ln \mathsmaller{\frac{1}{\EigTwoBound}}} + T \log\left(\frac{\log(\nicefrac{\frange}{\varepsilon})}{\delta}\right)\left( \frac{\frange}{\varepsilon} + \frac{\ITVar}{\varepsilon^{2}}\right)\right) & \\%[-0.5cm]
 &= \mathcal{O}\left( \frac{\ln \mathsmaller{\frac{1}{\PiMinBound}}}{\ln \mathsmaller{\frac{1}{\EigTwoBound}}} + \log\left(\frac{\log(\nicefrac{\frange}{\varepsilon})}{\delta}\right)\left( \frac{\frange}{(1-\EigTwoBound)\varepsilon} + \frac{\TRel\ITRelVar}{\varepsilon^{2}}\right)\right) \enspace. & \\
 & \hspace{\textwidth} & \\[-0.5cm]
\end{flalign*}
\end{minipage}
\end{enumerate}
\end{theorem}
\begin{proof}%[Proof of Theorem \ref{thm:efficiency}]
\newcommand{\ltd}{\eta}%
The strategy here is to derive a sample size $m'$, dependent on $\frange, \ITVar, \varepsilon, \delta$, s.t. w.h.p., each algorithm will terminate after drawing a sample of at least $m'$ traces.
We then bound the total number of samples drawn over the course of this process, and make some substitutions to derive the result.
For brevity, throughout this result we take $\ltd \doteq {\ln \! \frac{3\NIterations}{\delta}}$.

We show the result for \DMCMC{}, using second absolute eigenvalue bound $\EigTwo$, as it immediately implies the corresponding results for \algo{} and \nsalgo{}.

\if 0
\[
\Prob\left( \SVar \geq \hat{v} + \frac{5\frange^{2}\ln \frac{1}{\delta}}{(1-\EigTwoBound)m} + \sqrt{\frac{(1 + \EigTwoBound)\frange^{2} \SVar \ln \frac{1}{\delta}}{(1 - \EigTwoBound)m}} \right) \leq \delta \enspace.
\]

\[
\Prob\left( \SVar \geq \hat{v} + \frac{(11 + \sqrt{21})(1 +  \mathsmaller{\frac{\EigTwoBound}{\sqrt{21}}}) \frange^{2} \ln \frac{1}{\delta}}{(1-\EigTwoBound)m} + \sqrt{\frac{(1 + \EigTwoBound)\frange^{2} \hat{v}\ln \frac{1}{\delta}}{(1 - \EigTwoBound)m}} \right) \leq \delta \enspace.
\]
\fi

We first show that, with high probability, the \emph{empirical variance} is not much larger than the true variance, and thus with high probability, the variance-bounds used by \DMCMC{} are not loose.
Let
\[
\varepsilon_{v,1} \doteq \frac{5\frange^{2}\ltd}{(1-\EigTwoBound)m}, \ \ \varepsilon_{v,2} \doteq \sqrt{\frac{(1 + \EigTwoBound)\frange^{2} \SVar \ltd}{(1 - \EigTwoBound)m}} \enspace.
\]
Now, note that by \cref{lemma:twochain},\footnote{Note that here we take a lower-tail bound, rather than an upper-tail bound; the constants are identical and the result similarly follows from the Bernstein inequality.} we have for any sample size $m$ that
\[
\Prob\left(\VEVar_{i} \geq \SVar + \varepsilon_{v,1} + \varepsilon_{v,2} \right) \leq \frac{\delta}{3\NIterations} \enspace.
\]
%which is a lower-tail analogue of \cref{lemma:tracemixing} (note that Bernstein upper and lower tail bounds are symmetric).

\smallskip

\todo{Don't use $\VBoundVarUpper_{i}$?  Different symbol?}

%Consequently, at
We now consider the first iteration $i$ such that $m_{i} \geq m'$, letting $m \doteq m_{i}$ and $\hat{v} \doteq \VEVar$.
On line~\ref{alg:mcd-var} of \DMCMC{}, we have (w.h.p.)
\begin{align*}
\VBoundVarUpper_{i} &\leq \hat{v} + \frac{(11 + \sqrt{21} + (1 +  \mathsmaller{\frac{11}{\sqrt{21}}})\EigTwoBound) \frange^{2} \ltd{}}{(1-\EigTwoBound)m} + \sqrt{\frac{(1 + \EigTwoBound)\frange^{2} \hat{v}\ltd{}}{(1 - \EigTwoBound)m}} & \\
 %&\leq \SVar + \frac{5\frange^{2}\ltd{}}{(1-\EigTwoBound)m} + \sqrt{\frac{(1 + \EigTwoBound)\frange^{2} \SVar \ltd{}}{(1 - \EigTwoBound)m}} + \frac{(11 + \sqrt{21} + (1 +  \mathsmaller{\frac{11}{\sqrt{21}}})\EigTwoBound) \frange^{2} \ltd{}}{(1-\EigTwoBound)m} + \sqrt{\frac{(1 + \EigTwoBound)\frange^{2} (\SVar + \frac{5\frange^{2}\ltd{}}{(1-\EigTwoBound)m} + \sqrt{\frac{(1 + \EigTwoBound)\frange^{2} \SVar \ltd{}}{(1 - \EigTwoBound)m}}) \ltd{}}{(1 - \EigTwoBound)m}}  & \\
 &\leq \SVar + \varepsilon_{v,1} + \varepsilon_{v,2} + \frac{(11 \! + \! \sqrt{21} \! + \! (1 \! + \! \mathsmaller{\frac{11}{\sqrt{21}}})\EigTwoBound) \frange^{2} \ltd{}}{(1-\EigTwoBound)m} \! + \! \sqrt{\!\frac{(1 \! + \! \EigTwoBound)\!\frange^{2} (\SVar \! + \! \varepsilon_{v,1} \! + \! \varepsilon_{v,2}) \ltd{}}{(1 - \EigTwoBound)m}\!} \!\!\! & \hspace{-0.25cm} \hat{v} \! \leq \! \SVar \! + \! \varepsilon_{v,1} \! + \! \varepsilon_{v,2} \  \text{(w.h.p.)} \\
 &= \SVar + \varepsilon_{v,2} + \frac{(16 + \sqrt{21} + (1 +  \mathsmaller{\frac{11}{\sqrt{21}}})\EigTwoBound) \frange^{2} \ltd{}}{(1-\EigTwoBound)m} + \sqrt{\frac{(1 + \EigTwoBound)\frange^{2} (\SVar + \varepsilon_{v,1} + \varepsilon_{v,2}) \ltd{}}{(1 - \EigTwoBound)m}} & \\
 &= \SVar + \varepsilon_{v,2} + \frac{(16 + \sqrt{21} + (1 + \mathsmaller{\frac{11}{\sqrt{21}}})\EigTwoBound) \frange^{2} \ltd{}}{(1-\EigTwoBound)m} + \sqrt{\frac{(1 + \EigTwoBound)\frange^{2} \ltd{}}{(1 - \EigTwoBound)m}}\sqrt{\SVar + \varepsilon_{v,1} + \varepsilon_{v,2}} & \textsc{Algebra} \\
 &< \SVar + \varepsilon_{v,2} + \frac{(16 + \sqrt{21} + (1 + \mathsmaller{\frac{11}{\sqrt{21}}})\EigTwoBound) \frange^{2} \ltd{}}{(1-\EigTwoBound)m} + \sqrt{\frac{(1 + \EigTwoBound)\frange^{2} \ltd{}}{(1 - \EigTwoBound)m}} \bigl(\sqrt{\SVar} \! + \! \sqrt{\varepsilon_{v,1}}\bigr) & \hspace{-0.25cm} \sqrt{a + b + 2\sqrt{ab}} = \sqrt{a} + \sqrt{b} \\
 &\leq \SVar + \varepsilon_{v,2} + \frac{(16 + \sqrt{5} + \sqrt{21} + (1 + \mathsmaller{\frac{\sqrt{5}}{2}} + \mathsmaller{\frac{11}{\sqrt{21}}})\EigTwoBound) \frange^{2} \ltd{}}{(1-\EigTwoBound)m} + \sqrt{\frac{(1 + \EigTwoBound)\frange^{2}\SVar \ltd{}}{(1 - \EigTwoBound)m}} & \\ %\hspace{-0.9cm} \sqrt{a \! + \! b \! + \! 2\sqrt{ab}} = \sqrt{5(1+\EigTwoBound)} < \sqrt{5}(1 \! + \! \mathsmaller{\frac{\EigTwoBound}{2}}) \\
 &= \SVar + \underbrace{\frac{(16 + \sqrt{5} + \sqrt{21} + (1 + \mathsmaller{\frac{\sqrt{5}}{2}} + \mathsmaller{\frac{11}{\sqrt{21}}})\EigTwoBound) \frange^{2} \ltd{}}{(1-\EigTwoBound)m}}_{\doteq\varepsilon_{v,3}} + \underbrace{2\sqrt{\frac{(1 + \EigTwoBound)\frange^{2}\SVar \ltd{}}{(1 - \EigTwoBound)m}}}_{\doteq\varepsilon_{v,4}} & \\
 \if 0
 & OLD & \\
 &\leq \SVar + \varepsilon_{v,2} + \frac{(21 + \sqrt{21} + (1 +  \mathsmaller{\frac{11}{\sqrt{21}}})\EigTwoBound) \frange^{2} \ltd{}}{(1-\EigTwoBound)m} + \sqrt{\frac{(1 + \EigTwoBound)\frange^{2} (\SVar + \varepsilon_{v,2}) \ltd{}}{(1 - \EigTwoBound)m}}  & \\
 &\leq \SVar + \varepsilon_{v,2} + \frac{(21 + \sqrt{21} + (1 +  \mathsmaller{\frac{11}{\sqrt{21}}})\EigTwoBound) \frange^{2} \ltd{}}{(1-\EigTwoBound)m} + 2\sqrt{\frac{(1 + \EigTwoBound)\frange^{2} \SVar \ltd{}}{(1 - \EigTwoBound)m}}  & \\
 &= \SVar + \frac{(21 + \sqrt{21} + (1 +  \mathsmaller{\frac{11}{\sqrt{21}}})\EigTwoBound) \frange^{2} \ltd{}}{(1-\EigTwoBound)m} + 3\sqrt{\frac{(1 + \EigTwoBound)\frange^{2} \SVar \ltd{}}{(1 - \EigTwoBound)m}}  & \\
 \fi
\end{align*}

Substitution into the Bernstein bound (line~\ref{alg:bern}), and similar algebra, gives us
\begin{align*}
\hat{\bm{\epsilon}}_{i} &= \frac{10\frange\ltd{}}{(1-\EigTwoBound)m_{i}} + \sqrt{\frac{(1 + \EigTwoBound) \VBoundVarUpper_{i}\ltd{}}{(1 - \EigTwoBound)m_{i}}} \\
 &\leq \frac{10\frange\ltd{}}{(1-\EigTwoBound)m_{i}} + \sqrt{\frac{(1 + \EigTwoBound) (\SVar + \varepsilon_{v,3} + \varepsilon_{v,4})\ltd{}}{(1 - \EigTwoBound)m_{i}}} & \textsc{See Above} \\
 &< \frac{5(6 + \EigTwoBound)\frange\ltd{}}{2(1-\EigTwoBound)m_{i}} + \sqrt{\frac{(1 + \EigTwoBound) \SVar \ltd{}}{(1 - \EigTwoBound)m_{i}}} & \begin{tabular}{r} $\sqrt{a + b + 2\sqrt{ab}} = \sqrt{a} + \sqrt{b}$ \\ \textsc{Algebraic Bounds} \\ \end{tabular}\!\!\! \\ %\textsc{Algebra} \\
 %True but loose:
 %&\leq \frac{18\frange\ltd{}}{(1-\EigTwoBound)m_{i}} + \sqrt{\frac{(1 + \EigTwoBound) \SVar \ltd{}}{(1 - \EigTwoBound)m_{i}}} & \textsc{Algebra} \\ %sqrt(16+sqrt(5) + sqrt(21) + (1 + sqrt(5)÷2 + 11÷sqrt(21))×.1) × sqrt(1.5) \\
 %OLD:
 %&< \frac{20\frange\ltd{}}{(1-\EigTwoBound)m_{i}} + 4\sqrt{\frac{(1 + \EigTwoBound) \SVar \ltd{}}{(1 - \EigTwoBound)m_{i}}} & \\
\end{align*}
%\cyrus{$\frac{m}{(m-9)^2} = \frac{m}{m^2 - 18m + 81} = \frac{1}{m-18+81/m}$}

\todo{More detail}
\todo{Constant is < $7.982$ for no assumption on $\EigTwoBound$: see sqrt(16+sqrt(5) + sqrt(21) + (1 + sqrt(5)/2 + 11/sqrt(21))*2) * sqrt(2) \ \ \ AND \ \ \  sqrt(16+sqrt(5) + sqrt(21) + (1 + sqrt(5)/2 + 11/sqrt(21))*.1) * sqrt(1.5)}

\todo{Split double step: using $\varepsilon_{v,3} * ... * (1 + \EigTwoBound) \leq 5 + 2.5\EigTwoBound$: see \url{https://www.desmos.com/calculator/2wn6mfnd34}}

We terminate when $\hat{\bm{\epsilon}}_{i} \leq \varepsilon$, thus
this implies sufficient sample size
\[
%m' \leq + \ltd{} \left(\frac{20\frange}{(1-\EigTwoBound)\varepsilon} + \frac{16(1+\EigTwoBound)\SVar}{(1-\EigTwoBound)\varepsilon^{2}} \right)
m' \leq + \ltd{} \left(\frac{5(6 + \EigTwoBound)\frange}{2(1-\EigTwoBound)\varepsilon} + \frac{(1+\EigTwoBound)\SVar}{(1-\EigTwoBound)\varepsilon^{2}} \right) \enspace.
\]

Now, due to the doubling geometric grid, we must have $m_{i} \in [m', 2m']$, thus we have\cyrus{Subtlety with start}
\[
%m_{i} < 2\ltd{} \left(\frac{20\frange}{(1-\EigTwoBound)\varepsilon} + \frac{16(1+\EigTwoBound)\SVar}{(1-\EigTwoBound)\varepsilon^{2}} \right)
m_{i} \leq 2\ltd{} \left(\frac{5(6 + \EigTwoBound)\frange}{2(1-\EigTwoBound)\varepsilon} + \frac{(1+\EigTwoBound)\SVar}{(1-\EigTwoBound)\varepsilon^{2}} \right) \enspace.
\]

Now, each step of the tensor-product chain $(\M \otimes \M)$ requires two steps of $\M$, so we conclude (1) by noting that $4m'$ samples suffice.
%and we substitute in $\NIterations = \log_{2}(\frac{\frange}{2\varepsilon})$.

\medskip

Now, to get (2), note that in \algo{}, we take $T \doteq \ceil{\frac{1+\EigTwoBound}{1-\EigTwoBound}\ln \sqrt{2}}$, and thus $\EigTwo(\M^{(T)}) \leq \EigTwo^{T} \leq \frac{1}{2}$.
The finite-sample bound then follows from (1) as \algo{} simply calls \DMCMC{} (see line~\ref{alg:dyna:return}, multiplying total sample complexity by $T$, as each step in $(\M^{(T)} \otimes \M^{(T)})$ (i.e., the tensor-product trace-chain) takes $T$ steps in $\M$ for every step in the $\M \otimes \M$ chain of \DMCMC{}.
Finally, applying \cref{lemma:tvar-prop} yields the result.

\end{proof}

\subsubsection{A Note on Nonstationarity}\label{remark:stationarity}
Note that theorems \ref{thm:hoeffding} and \ref{thm:bernstein} assume \emph{stationarity} i.e. they consider traces  $\vec{X}_{1:m} :X_1, X_2,\dots X_m$  assuming $X_1\sim \pi$. This assumption is often prohibitive, as drawing even a single such sample can be NP-hard.
We overcome this problem using the following equation for $\vec{X}_{1:m} :X_1, X_2,\dots X_m$, $X_1\sim \nu$. 
\begin{equation}
\label{eq:nonstationary}
\smash{\Prob_{\substack{\vec{X}\\X_1\sim \nu}}}\left( \lvert \hat{\mu} - \mu \rvert \geq \varepsilon \right) \leq \left\lVert \frac{\partial \nu}{\partial \pi} \right\rVert_{\pi, \infty} \smash{\Prob_{\substack{\vec{X}\\{X}_{1} \sim \pi}}}\left( \lvert \hat{\mu} - \mu \rvert \geq \varepsilon \right) \enspace;
\end{equation}
see, e.g., \cite{fan2018hoeffding}, proof of thm 2.3.
Note that by eq. \ref{eq:nonstationary}\cyrus{use \eqref{eq:nonstationary} or eq.\ \ref{eq:nonstationary}}, it is sufficient to have  $\left\rVert\frac{\partial \nu}{\partial \pi} \right\rVert_{\pi, \infty} = \ess\sup_{\omega\in \Omega}\vert \frac{\nu(x)}{\pi(x)}\vert \in \mathcal{O}(1)$. This can generally be accomplished straightforwardly with a \emph{warm-start} by selecting an arbitrary fixed $\omega \in \Omega$, taking $\nu$ to be the distribution reached after running $\M$ for $\TRel(\M)\ln(1/\pi_{\min})$ steps (see any standard MCMC text book, e.g. \cite{levin2017markov}).  With this in mind, for simplicity we assume stationarity, knowing that our proofs and algorithm generalize  with trivial modifications.\cyrus{Comment on uniform mixing time!}

\subsection{Missing proofs from \ref{sec:compare}}\label{sec:compareproofs}

The following equations compare and contrast Rabinovich et al.'s bound with the central limit theorem and this work:
\begin{align}\label{eqution8}
\frac{{\cal T}_{\rm fmix}(f,\M \frac{\frange}{\varepsilon})\frange^{2}}{\varepsilon^{2}} \log \frac{1}{\delta}
  &\geq \frac{\TfMixBound(f,\M)}{\TfMix(f,\M)} \cdot \frac{R^{2}}{\SVar} \cdot \frac{\AVar}{\varepsilon^{2}} \log(\frac{1}{\delta}) \log (\frac{\frange}{\varepsilon})
\end{align}
From which we can conclude 
\[
  \HPSC_{\rm Rabi}(\M,f) \in \Omega \left( \frac{\TfMixBound(f)}{\TMix(f)} \cdot \frac{R^{2}}{\SVar} \cdot \AVar(f,\M) \log \frac{\frange}{\varepsilon} \right) \enspace,
  \]
which implies 
 \[ \HPSC_{\rm Rabi}(\M,f) \in \Omega \left(  \AVar(f,\M) \log \frac{\frange}{\varepsilon} \right)
  \enspace.
\]

Note that the gap between our bound and the central limit theorem is $\Theta\left(1+\log\log (\frange/\varepsilon)\right)$ which is exponentially smaller compare to the $\log (\frange/\varepsilon)$ appearing here. 

\begin{proof}[Proof of \cref{eqution8}]

\begin{align*}
\TfMixBound(f,\M \frac{\varepsilon}{\frange})\frange^{2}
 &= \frac{\TMixBound(f,\M \frac{\varepsilon}{\frange})}{\TfMix(f,\M \frac{\varepsilon}{\frange})} \cdot \frac{R^{2}}{\SVar} \cdot \TMix(f,\M \frac{\varepsilon}{\frange}) \SVar \\
 %? More direct step; e.g.\ thm~2.5 of \cite{levin2017markov}??)} \\
 &\gtrsim \frac{\TfMixBound(f,\M)}{\TfMix(f,\M)} \cdot \frac{R^{2}}{\SVar} \cdot \TfRel(f,\M) \SVar \log \frac{\frange}{\varepsilon} & \textsc{eq.~(9c) \cite{functionMixingRabinovish} }
 %(gap up to $\log \frac{1}{\pi_{\min}}$?)} 
 \\
 &\geq \frac{\TfMixBound(f,\M)}{\TfMix(f,\M)} \cdot \frac{R^{2}}{\SVar} \cdot \AVar \log \frac{\frange}{\varepsilon} & \cref{lemma:tvar-prop}: \AVar \leq 2 \TRel(f) \SVar(f)  \\
\end{align*}
\end{proof}

\subsection{Application to counting k colorings}\label{sec:plantedap}

Consider a graph $G = (V, E)$,  and some \emph{number of colors} $k$. A \emph{coloring} of $G$ is a mapping $\PColor: V\to \{1,2,\dots k\}$, where $\gamma(v)$ denotes the \emph{color} of $v \in V$, and a \emph{proper coloring} is any coloring $\PColor$
 s.t.~$\forall u, v \in V: \, \Pcolor(u)=\Pcolor(v) \implies (u,v)\notin E$. 
 %satisfying: for any $u,v \in V$, if $\Pcolor(u)=\Pcolor(v)$ then $(u,v)\notin E$. 
 %For such $G$ and $k$, and a proper $k$-coloring $\Pcolor$, we denote color of any arbitrary vertex $v\in V$ by $\Pcolor (v)$.
 For a subset ${V}'\subseteq V$, by $\Pcolor({V'})$, we mean the restriction of the mapping $\Pcolor$ to domain ${V'}$, i.e., a proper coloring on the induced subgraph on ${V'}$. Furthermore we define $\Gamma(G,k)$ to be the set containing all the proper $k$-colorings of $G$. We denote the size of a set $X$ by $\# X$, and the uniform distribution on it by $\UnifD(X)$. For example $\# \Gamma(G,k)$ is the number of proper $k$-colorings of $G$, and $\UnifD(\Gamma(G,k))$ is the uniform distribution on it. 
 For any graph $G$, 
 the \emph{Glauber dynamics }\footnote{Also known as the zero temperature Pott's model, or the single site update} chain   is defined on $\Gamma(G,k)$ as follows and it known that it converges to  stationary distribution 
 $\UnifD(\Gamma(G,k))$:
 
% The following Markov chain, known in general form as \emph{Glauber dynamics}, the \emph{zero temperature Pott's model}, or the \emph{single site update} chain was initially introduced to sample uniformly  proper colorings by Jerrum:

   \begin{defin}[Glauber dynamics chain for proper $k$-colorings \cite{jurrumSamplingiscounting}]\label{def:coloringchain}
We define the Markov chain ${\cal M}_{G}$ on $\Gamma(G,k)$ as follows:
At each time step $t$, let $X_t$ be a proper $k$-coloring of $G$,
\smallskip 

\noindent
 (1)  Pick $c\in\{1,2,\dots k\}$, and $u\in V$ uniformly at random. 
 
\noindent 
 (2) If   changing $u$'s color to $c$ is still a valid proper coloring, let $X_{t+1}$ be this new coloring. 

\noindent 
 (3)  Else, $X_{t+1}=X_t$.
  
\end{defin}
 
 The planted partition, or stochastic block, model
generalizes the Erd\"os-R\'enyi model, allowing for \emph{communities} in graphs, and it has the following distribution:
 
\begin{restatable}[Planted Partition Model]{defin}{defplantedpartitions}
\label{def:plantedpartitions}
Given the following parameters:
(i)  number of vertices $n$, (ii)  a partitioning of $\{1,2,\dots,n\}$ to $r$ subsets $C_1, C_2,\dots , C_r$, and (iii) an \emph{edge placement matrix} $P \in [0, 1]^{r \times r}$.
%an $r\times r$ matrix $P$ for probabilities of edge placement. 
A graph $G$ with $\#V=n$ is generated as follows: for any two vertices $u\in C_i$ and $v\in C_j$, the edge $(u,v)$ is in $E$ with probability $P_{i,j}$.
\end{restatable}
A simplified version of the planted partition model is when each $C_i$ has size $\nicefrac{n}{r}$, the diagonal elements of $P$ are all $p$, and the other elements are $\frac{q}{r-1}$.  Since the probability of having edges inside a community is often more than having edges between two communities we assume $q/p\in o(1)$. We denote this model by $\Pl(n,r,p,q)$.
We show bounds for this simplified model, though they may easily be extended to arbitrary planted partitions graphs or similar network models having small clusters.

\begin{restatable}
[Application of \algo{} to Planted Partitions]{theorem}{planted}
\label{thm:plantedpartition}
 
Consider $G \sim \Pl(n,r,p,q)$.
Assume %$q/p\leq 1/n$
$p \geq qn$, and let $k> \nicefrac{qn^4}{r^2}$.
With probability $\geq 1-o(1)$, 
Jerrum's counting algorithm on $G$, with $\delta=\epsilon=\nicefrac{1}{4}$, equipped with \algo{} as a mean-estimation gadget, has sample complexity
\[
\tilde{O}\Biggl(
n^2\left((\# E)^{2}+\frac{(\# E)^3}{r}\right)\Biggr) \enspace.
\]
\end{restatable}

\begin{proof}[\textbf{Prelude to a Proof }]
 The proof is based on developing two new notions:  loosely connectedness (\cref{def:looselyconnected}) and restriction of a chain (\cref{def:restrictedchain}) to a subset of its vertices. We use these definitions to prove \cref{lemm:loosely}  which then constitutes the proof of \cref{thm:plantedpartition}. Here we provide a road map and intuition, and the proofs are presented in full detail in \cref{sec:proofdetialsplanted}. 
 
 Consider a subset of graph vertices, $V'\subset V$,  and let $G'$ be a graph obtained from $G$ after removing cutting the edges between $V'$ and $V\setminus V'$.  When  $\UnifD(\Gamma(G,k))$ is \emph{only negligibly different} from $\UnifD(\Gamma(G',k))$ we say $V'$ is 
\emph{loosely connected} to the rest of the graph (see \cref{def:looselyconnected}). 
In lemma \ref{lemm:loosely}, we show that when it occurs, $\ITVn{\TRel({\cal M}_G)}({\cal M}_G,f_e) =o\left( \SVar(f_e)\right)$ where $f_e$ is intermediate phase showing up in JVV reduction (see \cref{sec:jvvred}) and $\M_G$ is the Glauber dynamic chain.
The proof is based on coupling the probability spaces of the two Markov chain's traces (${\cal M}_G$ and ${\cal M}_{G'}$). Then we show  it is sufficient that a simpler and faster mixing variant of ${\cal M}_G'$  mixes (see \cref{def:restrictedchain}).

 Finally, we show that w.h.p., graphs distributed according to the \emph{planted partition model}, and their subgraphs, have  loosely connected parts.
 \end{proof}

 \smallskip

\medskip

\subsubsection{Definitions and proof sketches from \ref{sec:plantedap} }\label{sec:proofdetialsplanted}
Consider  a subset of $G$'s vertices $V'\subseteq V$ and assume that $f:\Gamma(k,G)\rightarrow[0,1]$ is a function whose value \emph{only depends} on a coloring \emph{restricted to} some $V'$, i.e. $f(\Pcolor(V))=f(\Pcolor(V'))$ for any $\Pcolor$. %i.e.,  $f(\Pcolor)=f(\Pcolor(V_j))$. 
We define ${\M}'_G$ being restriction of $\M_G$ to $V'$ as follows:

\begin{restatable}[Restricted Glauber dynamic chain]{defin}{defrestrictedchain}
\label{def:restrictedchain}
We define Markov chain ${\M'}_G$ as follows:
At each time step $t$, let $X_t$ be a proper $k$-coloring of the induced subgraph of $V'$, and then
\begin{enumerate}
\item Pick $c\in\{1,2,\dots k\}$, and $u\in V$ uniformly at random. 
\item If  $u\in V_i$ then follow ${\cal M}_G$ to get to $X_{t+1}$ (\Cref{def:coloringchain}). 
\item Else, $X_{t+1}=X_t$.
\end{enumerate}
\end{restatable}

If there is no edge between $V'$ and $V\setminus V'$, then the image of  $f(\Pcolor)$ under traces of $\M_G$ and restricted chain $\M'_G$ will have the same distribution. From this,  \cref{lemma:tvar-prop} and $\SVar \leq \nicefrac{1}{4}$, we bound the length-$T$ trace variance as $\ITRelVar\leq \smash{(1/4)(\nicefrac{\TRel(\M')}{2\TRel(\M)})=\nicefrac{\TRel(\M')}{2\TRel(\M)}}$~.

By defining \emph{loosely connectedness} our goal is to consider the case some  when the cut between $V'$ and $V$ does not make the image of $f(\Pcolor)$ in traces of $\M'$ and under $\M$ \emph{noticeably} different, i.e., we can couple the images of $f$ on a trace of $\M$ and on a trace of $\M'$ so that the coupled states are with high probability identical. Thus, we can largely ignore the connecting edges between $V'$ and $V\setminus V'$.
%Note that when $\TRel(\M_j)\ll\TRel(\M)$, \cref{alg:1}  greatly improves the sample complexity. 

%\shahrzad{Eli said this definition is confusing}

%\shahrzad{I think it is better to define looseness parameter $c$ satisfying $1/\zeta\geq \TRel({\cal M})\log(c)$}

%\shahrzad{Like: the graph is loosely connected if $1/\zeta\geq \TRel({\cal M})\log(c)$ and we call $\zeta$ the connectedness parameter and $c$ the  looseness parameter. OTHER TERMINOLOGY?}
\begin{defin}[Connectedness parameters $\zeta$ and looseness parameter $c$]\label{def:looselyconnected}

For a subset $V'$ of $V$, we define the \emph{cut set} $B({V'})\subseteq E$ as $B(V')\doteq\{(u,v) \in E : u\in V', v\in V\setminus V'\}$. Let $G' = (V, E \setminus B(V'))$. %be the graph we obtain by removing $B({V_j})$ from $E$. 
We say $\zeta$ is the connectedness parameter of $V'$ if
\vspace{-0.15cm}
\[
\Prob_{\Pcolor}\Bigl(\smash{\bigvee_{(u,v)\in B(V_i)}} \Pcolor(u) = \Pcolor(v)\Bigr)\leq \zeta \enspace,\ \text{for} \   \Pcolor \sim  \UnifD(\Gamma(G',k))\enspace.
\]

If there exists $c>1$ such that $1/\zeta\geq \TRel({\cal M}_G)\log(c)$, we say that $V'$ is loosely connected to ${\cal M}_G$ and $c$ is the looseness parameter.  
\if 0
\cyrus{Probably need:}
\[
\sup_{i}\Prob_{\Pcolor}\Bigl(\bigvee_{(u,v)\in B(V_i)} \Pcolor(u) = \Pcolor(v)\Bigr)\leq \zeta \enspace.
\]
\fi

\end{defin}

\begin{remark}

The reader may postulate that with the above definition, verifying whether a subset is loosely connected needs calculating complex  probabilities. This is correct, \emph{however}, note that knowledge of $\zeta$ or $c$ is not required by \algo{}; it is merely used to characterize performance.

\end{remark}

The following lemma 
whose proof in  \cref{sec:missing5} shows that loosely connecteness implies small relaxed trace variance.

\begin{restatable}
{lemma}{loosely}
\label{lemm:loosely}
Let $0<\zeta<1$, $c>1$ and $V'\subseteq V$ and $G'$ be as defined in \cref{def:looselyconnected}. Let $\M$ and $\M'$ be the Glauber dynamics chains on $G$ and $G'$ and $\TRel(\M)$ and $\TRel(\M')$ respectively their relaxation times.
Assume $f:\Gamma(G,k)\rightarrow [0,1]$ such that for a coloring $\Pcolor$, $f(\Pcolor)=f(\Pcolor(V'))$.
We have $$\ITVn{\TRel({\cal M})}({\cal M})\leq \frac{2\TRel({\cal M}')}{\TRel(\cal M)}+\frac{\log c}{c} ~.$$
%Consider a trace of length $T$ of $\M$, $\vec{X}:X_1, X_2,\dots X_T$.    we have: $v_{T}\leq \TRel(\M_j)/2T+12\zeta(1+\frac{9\log(T)}{T})$. Note that since we assume  $T\geq \TRel(\M)$, we have $\ITVar\leq \TRel(\M_j)/2\TRel(\M)+12\zeta(1+\frac{9\log(T)}{T})$.
\end{restatable}

The last step of the proof will be to show that having $k> \nicefrac{qn^4}{r^2}$, in the planted partition model each $C_i$ is with high probability loosely connected to $\M_G$. This is proved in the following section.

\subsubsection{Detailed Proofs from  \ref{sec:plantedap}}\label{sec:missing5}
%\shahrzad{Must be changed }
We now restate the definition of the planted partition graph model (\cref{def:plantedpartitions}), and prove the result related to it.

\defplantedpartitions*

\planted*

%\shahrzad{we need it to be $1/(n^2\log c)$ loosely connected }
\begin{proof}[Proof of Theorem \ref{thm:plantedpartition}]
Let $G\sim {\mathcal{P}(n,r,p,q)}$.
We first show that with probability at least $1-re^{-\frac{2n^2q}{r^2}}$, and taking $\zeta \doteq \frac{2n^2q}{kr^2}$, any partitions $C_j$, is $\zeta-$loosely connected to the rest of the graph $G$.  
\smallskip

For any $j$, and $v\in C_j$, let $d_{out}(v)$ be number of edges having one end outside $C_j$. Remember definition of $B(C_j)$ from Definition \ref{def:looselyconnected}, we have:

\[
B(C_j)= \sum_{v \in C_{j}} \# d_{out}(v) \enspace,
\]
Taking
$
B_{\max} \doteq \max_{j \in 1, \dots r} B(C_j) 
$, note that by symmetry and the fact that for each two end points of an edge in $ B(C_j)$ we can recolor one end point to make their colors different, we will have that \[
\Prob_{\gamma}\Bigl(\bigvee_{(u,v)\in B(C_j)} \gamma(u) = \gamma(v)\Bigr)\leq  \frac{B_{\max}}{k} \enspace.
\]

\smallskip 

Now, note that $\#C_{i} = \frac{n}{r}$, and furthermore, for each $v\in C_j$,
\[
d_{out}(v) \sim \BinomD(n - \frac{n}{r}, \frac{q}{r-1}) = \BinomD(\frac{n(r-1)}{r}, \frac{q}{r-1}) \enspace,
\]
where $\BinomD(n, p)$ denotes the binomial distribution. Thus by independence, we have
\[
B(C_j) \sim \BinomD(\frac{n^{2}(r-1)}{r^{2}}, \frac{q}{r-1}) \enspace, \Expect[B(C_j)] = \frac{n^{2}q}{r^{2}} \enspace, \Var[B(C_j)] = \frac{n^{2}q}{r^{2}}(1 - \frac{q}{r-1}) < \frac{n^{2}q}{r^{2}} \enspace.
\]

By the Gaussian CLT Chernoff bound, we have approximately
\[
\Prob\left(B(C_{j}) > \frac{n^{2}q}{r^{2}} + \sqrt{\frac{2n^{2}q}{r^{2}}\ln(\frac{1}{\delta})} \right)
    %\approx \Prob(\mathcal{N}(0, 1) \geq ?)
    %= \Prob\left(B_{i} > \frac{n^{2}q}{r^{2}} + \frac{n}{r}\sqrt{q\ln(\frac{1}{\delta})} \right) 
    \leq \delta \enspace.
\]

Furthermore, by the union bound, we have approximately
\[
\Prob\left(B_{\max} > \frac{n^{2}q}{r^{2}} + \sqrt{\frac{2n^{2}q}{r^{2}}\ln(\frac{r}{\delta})} \right)
    %= \Prob\left(B_{i} > \frac{n^{2}q}{r^{2}} + \frac{n}{r}\sqrt{q\ln(\frac{1}{\delta})} \right) 
    \leq \delta \enspace.
\]

Take $\delta = \exp(\frac{-n^{2}q\ln(\frac{1}{r})}{2r^{2}})$, and we get
\[
\Prob\left( B_{\max} > \frac{2n^{2}q}{r^{2}} \right) \leq \exp\left(\frac{-n^{2}q\ln(\frac{1}{r})}{2r^{2}}\right) \enspace.
\]

Taking $
\zeta \leq \frac{2n^{2}q}{kr^{2}}$  with probability at least $1 - r\exp(\frac{-n^{2}q}{r^{2}})$,   we have that all $C_j$s are $\zeta-$loosely connected to the rest of the graph. 

Thus, 
by taking $k\geq \nicefrac{ q n^4}{r^2}\cdot \log c$ we will have that $\zeta\leq 1/(n^2\log c)$.
\medskip

We now show that with probability at least $1-e^{-(np)/(8r)}$ in a graph generated from $G\sim {\mathcal{P}}(n,r,p,q)$ we have: $\Dmax (G) \geq {(np)}/{(2r)}$.  This is easily derived from the Chernoff bound and having $\Expect[d(v)]\geq np/r$, where $d(v)$ is the degree of each vertex.
%\cyrus{Don't we need an upper-bound on $d_{\max}$ to guarantee $k <  11/6 * d_{\max}$?  Should just need to flip the Chernoff bound for $d_{\max(G)} \leq \frac{2np}{r}$, but also need to account for the (negligible) contribution of $q$?}

 \medskip 
 
 Having this bound, if $k> \frac{(11/6)np}{2r}$ Vigoda's result will be applicable to  $G\sim {\mathcal P}(n,r,p,q)$ with probability at least $1-o(1)$. 
 Thus, we can conclude that for $k\geq \nicefrac{qn^4}{r^2}\cdot \log c$ the relaxation time of $\cal M$ it bounded by
 $n^2$. 
 
 %Note that, for such $k$ and assuming $q/p\leq 1/n$ we have that all  $C_j$s are $\frac{1}{r}-$loosely connected to the rest of the graph. 
We now employ \Cref{lemm:loosely}, to $G$ and $C_j$s, we have $\ITRelVar\leq \TRel(\M_j)/2\TRel(\M)+\log(c)/c$. 
Employing Vigoda's bound on the relaxation time we know $\TRel(\M)\leq \TMix\leq  n^2$, and $\M_j$ is Jerrum's chain with a slowdown by a factor of $\frac{r}{n}$. Thus, $\TRel(\M_j)\leq \frac{n}{r} (\# C_j)^2 \leq nr$. Thus, we will  have: $\ITRelVar\leq 1/r+\frac{log c}{c}\leq 2/r$.

\medskip

Note that the above analysis holds for all $G_i$s in intermediate steps. We approximate  each $f_i$  with precision  $\varepsilon/\# E$, setting $\frange\leq 1$ the sample complexity of each intermediate step will be: 

\[T\left(
 36 + 2\ln(\frac{3 \log_{2}(\frac{\# E}{36\varepsilon}) }{\delta/\# E}) \left(\frac{50}{\varepsilon} + \frac{ (\frac{252}{3})^{2/3}(\# E)^{4/3}}{\varepsilon^{4/3}} + \frac{18(\# E)^2(\frac{1}{r^2}+\frac{22}{r^2})}{\varepsilon^{2}} \right) \right)
\]

Setting $T=n^2$ by Vigoda's result, letting $\delta$, $\varepsilon$ be constants, and ignoring vanishing terms, the complexity of Jerrum's algorithm equipped with \algo{} is
\[
\mathcal{O}\left( n^2
\log(n\log(n))\left((\# E)^{2}+\frac{(\# E)^3}{r}\right)\right) \enspace.
\]

\end{proof}

Having this, the last piece of puzzle is to prove \cref{lemm:loosely}, so we restate and prove it.

\loosely*

\begin{proof}[Proof of \cref{lemm:loosely}]

 We bound the trace variance of ${\cal M}$ by relating it to the trace variance of ${\cal M}'$

For a positive $c$, assume $1/\zeta\geq \TRel({\cal M})\log(c)$. let $\tau= 1/2\zeta$ and $n_\tau$ be the number of steps in a trace of length $\tau$ in which  some edge in $B(V')$ has both end points colored the same. Note that $\Expect_{U(\Gamma(G,k))}[n_{\tau}]=\tau \zeta = 1/2$,  applying the concentration bound of \cref{thm:hoeffding}, we will have

\begin{align*}
    \mathbb{P} (n_\tau\geq 1) \leq  \mathbb{P} (\vert n_\tau - 1/2\vert \geq 1/\sqrt{2} ) 
    \leq \exp\left({-2\tau(1/\sqrt{2})^2/{\TRel({\cal M})} }\right)  \leq \exp\left({-\tau/{\TRel({\cal M})} }\right)\leq \frac{1}{c}.
\end{align*}

Let ${\cal E}(\gamma)$ be the event that  $\gamma \sim \UnifD(\Gamma(G',k))$  colors none of the two end points of an edge in $B(V')$ to the same color.  
 Let $\tilde{\mu}$ be the expectation of $f(\gamma)$ when $\gamma \sim \UnifD(\Gamma(G',k))$.
We have $\mu=\Expect[f(\gamma)\vert {\cal E}(\gamma)]$, $\gamma\sim \Gamma(G',k)$. 
Furthermore, we can rewrite $\tilde{\mu}$ as
 $\tilde{\mu}= \Expect[f(\gamma)\vert {\cal E}(\gamma) ]\Prob({\cal E}(\gamma))+ \Expect[f(\gamma)\vert \neg{\cal E}(\gamma) ](1-\Prob({\cal E}(\gamma)))$. Thus, 
$
 {\mu}(1-\zeta)\leq\tilde{\mu}\leq {\mu}(1-\zeta)+\zeta  \text{, hence, } \vert\mu-\tilde{\mu}\vert \leq  \zeta
$. 

\smallskip 

Furthermore,  under the condition that $n_\tau<1$ the trace in $\cal M$ and the traces in ${\cal M}'$ are distributed identically.
The following inequalities are thus concluded:
\begin{align*}
    \Expect_{\vec{X}}\left[(\frac{S_X}{\tau}-\mu)^2\right] 
    &=    \Expect_{\vec{X}}\left[(\frac{S_X}{\tau}-\mu)^2\middle\vert n_{\tau}< 1\right]\mathbb{P}(n_{\tau}< 1)+\Expect_{\vec{X}}\left[(\frac{S_X}{\tau}-\mu)^2\middle\vert  n_{\tau}\geq 1\right]\left(1-\mathbb{P}(n_{\tau}< 1)\right)\\
    & =\Expect_{\vec{X}}\left[\left(\frac{S_X}{\tau}-\tilde{\mu}-(\mu-\tilde{\mu})\right)^2\middle\vert n_{\tau}< 1\right]\mathbb{P}(n_{\tau}< 1)+R^2\cdot \left(1-\mathbb{P}(n_{\tau}< 1)\right)\\
        & \leq \left(\Expect_{\vec{X}}\left[\left(\frac{S_X}{\tau}-\tilde{\mu}\right)^2-2\zeta\left(\frac{S_X}{\tau}-\tilde{\mu}\right)\middle\vert n_{\tau}< 1\right]+\zeta^2\right)\mathbb{P}(n_{\tau}< 1)+\left(1-\mathbb{P}(n_{\tau}< 1)\right).\\
         & \leq \left(\frac{2\TRel({\cal M}')}{\tau}+2\zeta\right)\mathbb{P}(n_{\tau}< 1)+\left(1-\mathbb{P}(n_{\tau}< 1)\right).\\
            & \leq 2\left(\frac{\TRel({\cal M}')}{\tau}+\zeta\right)(1-\frac{1}{c})+\frac{1}{c} ~.
    \end{align*}
    
Since $\tau=1/\zeta$ we have $(1/\tau) \TRel({\cal M}')=\zeta \TRel({\cal M}')\geq \zeta $. Thus for $\tau= \TRel({\cal M})\log c$ we have $\ITVn{\tau}({\cal M})\leq \frac{2\TRel({\cal M}')}{\tau}+\frac{1}{c}$. Using \cref{lemma:tvar-prop}, we have 

$$\ITVn{\TRel({\cal M})}({\cal M})\leq \log(c)\cdot \left(\frac{2\TRel({\cal M}')}{\tau}+\frac{1}{c}\right)=\frac{2\TRel({\cal M}')}{\TRel(\cal M)}+\frac{\log c}{c} ~.$$

\end{proof}

\section{A Compendium of  theorems and definitions used from the literature}\label{sec:thms}

\subsection{MCMC concentration bounds}\label{sec:classicbounds}

All of these bounds are \emph{static}; they receive a fixed $m$ as input and run the chain for $m$ steps to generate the trace $\vec{X}_{1:m}:X_1,X_2,\dots, X_m$. 
%These algorithms then returns the empirical mean $\hat{\mu}$ as the output.
 
 \medskip

In the following theorems $\vec{X}_{1:m}:X_1,X_2,\dots, X_m$, is a length $m$ stationary trace of $\M$, (i.e., $\vec{X}_{1:m} \sim \pi^{(T)}$), with mixing time $\TMix$, relaxation time $\TRel$, and second largest eigenvalue $\lambda$.
%Based on each of these theorems we obtain the sample complexity (sufficient $m$) of a variance-aware or variance-agnostic algorithm. Remember that $f$ is a bounded function with range $[a,b]$, we define $\frange \doteq b-a$ as the range of $f$. 

\begin{theorem}[Hoeffding-Type Bounds for Mixing Processes, {\cite[Thm.~2.1]{fan2018hoeffding}}]\label{thm:hoeffding}
For any $\delta \in (0, 1)$, 
we have
\begin{equation}
\label{eq:Hoeffding}
\Prob\!\left( \lvert \hat{\mu} - \mu \rvert \geq \sqrt{\frac{2(1+\EigTwo)(\mathsmaller{\frac{\frange^{2}}{4}})\ln(\frac{2}{\delta})}{(1-\lambda)m}} \right) \leq \delta \enspace.
\end{equation}

\if 0
\cyrus{Stationarity: need 2.3 of this paper: \url{https://arxiv.org/pdf/1802.00211.pdf}}

Can we bound
\[
\max_{\omega \in \Omega} \frac{\Prob(\nu = \omega)}{\Prob(\pi = \omega)} ?
\]
without invoking $\pi_{\min}$?
\fi
\end{theorem}

%Largely due to their simplicity and convenience, Hoeffding-type bounds enjoy ubiquity in computer science.Unfortunately, they are generally much looser than variance-aware bounds, %}(i.e., Gaussian tails and the lower bound ???), as the dependence on variance $v$ is replaced by \emph{variance proxy} $(\frac{\frange}{2})^2$, %(by Popoviciu's inequality, this is the \emph{maximum possible variance} of a range-$\frange$ random variable)(which is maximal assuming boundedness).While it is not generally possible to remove dependence on $\frange$ entirely, Bernstein-%and Bennetttype inequalities greatly reduce this dependence, often yielding \emph{asymptotic improvement} to sample complexity \cite{audibert2007tuning,mnih2008empirical,audibert2009exploration,maurer2009empirical}.
\if 0
Thus, in some sense, these Bernstein bounds bridge the gap between the \emph{variance-based} asymptotic (central limit theorem) Gaussian bounds %and mean-estimation lower-bounds
and the \emph{range-based} Hoeffding bounds.
%\cyrus{Do we talk about these lower bounds anywhere?}
\fi

\begin{theorem}[Bernstein-Type Bound for Mixing Process {\cite[Thm.~1.2]{jiang2018bernstein}}]
\label{thm:bernstein}For any $\delta \in (0, 1)$, we have

\begin{equation}\label{eq:Bernstein}
 \Prob\left(\vert \hat{\mu}-\mu\vert\geq \frac{10\frange\ln(\frac{2}{\delta})}{(1-\lambda)m} + \sqrt{\frac{2(1+\EigTwo)\SVar\ln(\frac{2}{\delta})}{(1-\lambda)m}} \right) \leq \delta \enspace.
\end{equation}
\end{theorem}
\noindent
\medskip

To compare the sample complexity of algorithms derived from the above bounds note that:

\smallskip 

\paragraph{Sample complexity of the static variance-agnostic algorithm:}
%Equation \ref{eq:Hoeffding} implies the following sample complexity bound:
This implies sample complexity
\[
m_{H}(\M, f, \varepsilon, \delta)
  = \frac{1+\Lambda}{1-\Lambda} \ln(\mathsmaller{\frac{2}{\delta}}) \frac{\frange^{2}}{2\varepsilon^2}
  \in \Theta\Bigl(\TRelBound\ln(\mathsmaller{\frac{1}{\delta}})\frac{\frange^{2}}{\varepsilon^{2}}\Bigr)
  \enspace.
\]

\paragraph{Sample complexity of the static variance-aware algorithm:}
%Equation \ref{eq:Bernstein} implies the following sample complexity bound:
This implies sample complexity
\[
m_{B}(\M, f, v, \varepsilon, \delta)
  = \frac{2}{1-\Lambda}\ln(\mathsmaller{\frac{2}{\delta}})\Bigl(\frac{5\frange}{\varepsilon} + \frac{(1+\Lambda)V_{\pi}}{\varepsilon^2}\Bigr)
  \in \Theta\Bigl(\TRelBound\ln(\mathsmaller{\frac{1}{\delta}})\Bigl(\frac{\frange}{\varepsilon} + \frac{V_{\pi}}{\varepsilon^{2}}\Bigr)\Bigr) \enspace.
\]

%\cyrus{It's not so much ``naive'' as ``aware'' of $\SVar$; there's a variant in terms of $v_{asy}$, so it's more about existing bounds assuming too much knowledge than having inefficiency.}

\if 0
\cyrus{Problem: does not handle nonstationary case!} % Maybe we need the bound from the Paulin paper?}

\url{https://arxiv.org/pdf/1805.10721.pdf}
\fi
%

\begin{theorem}[McDiarmid inequality for Markov chains   \cite{paulin2015}]
\label{thm:mcdiarmid-mixing}
Let $\M$ be a Markov chain on state space $\Omega$ and mixing time $\TMix$.
Consider $F:\Omega^{m} \to \R$ a $c$-Lipschitz function.
Then for any trace $\vec{X}_{1:m}:X_1,X_2,\dots, X_m$ of $\M$ and  any $\varepsilon > 0$ we have that
\begin{equation}
\label{eq:mcdiamard}
    \Prob_{\vec{X} \sim \pi^{(T)}}\left( F(\vec{X}) \geq \Expect_{\pi^{(T)}}[F] + \varepsilon\right) \leq \exp\left({\frac{-2   \varepsilon^2}{9c^2m\TMix}}\right) \enspace.
\end{equation}
Equivalently, for any $\delta \in (0, 1)$, we have that
\begin{equation}
\label{eq:mcdiamard-2}
    \Prob_{\vec{X} \sim \pi^{(T)}}\left( F(\vec{X}) \geq \Expect_{\pi^{(T)}}[F] + \sqrt{\frac{m\TMix c^{2}(\frac{7}{3})^{2} \ln(\frac{1}{\delta})}{2}} \right) \leq \delta \enspace.
\end{equation}
\end{theorem}

\subsection{Background on k-coloring problem}\label{sec:jvvred}

Let $S$ be a set whose cardinality $\# S$ is unknown .
Assume there is a \emph{rapidly mixing Markov chain} $\M$ whose stationary distribution is the uniform distribution of $S$ denoted by $\UnifD(S)$.
Now, note that with $\M$, we can generate \emph{approximately uniform} samples from $S$. 
Jerrum, Valient and Vazirani introduced the first FPRAS  (henceforth denoted by JVV) for counting any \emph{self-reducible} $S$, using series of MCMC-mean estimations (employing $\M$).
 JVV's reduction was applied to different counting problems \cite{countingleWinklerBrightwell,jurrumSamplingiscounting, permanent}, here we show it for counting $k$-colorings.

%Let $S$ be a set whose size is unknown. 
 Consider an arbitrary graph $G=(V,E)$, let ${\cal E}$ ve the number of edges, and  $e_1,e_2,\dots , e_{\cal E}$ be an ordering of its edges. Based on this ordering we define 
  $G_{\cal E}, G_{{\cal E}-1},\dots , G_0$  a sequence of $G$'s subgraphs with the same vertex set $V$, such that for each $i$, $G_{i-1}$ is obtained from $G_i$ by removing $e_i$, $G_{\cal E} = G = (V, E)$, and $G_0 = (V, \emptyset)$.
  
 For each $1\leq i\leq \cal E$ we define $f_{e_i}:\Gamma(G_{i})\rightarrow \{0,1\}$ as $f_{e_i}(\Pcolor)\doteq 1$ if $\Pcolor(u)\neq\Pcolor(v)$, and $f_{e_i}\doteq 0$ otherwise, where $u$ and $v$ are endpoints of $e_i$. 
 
 Note that $$\Expect[f_{e_i}]=\smash{\frac{\#\Gamma(G_i,k)}{\#\Gamma(G_{i-1},k)}}$$

 Thus, to estimate $\#S$ we can use the following telescoping sum, for a graph with $n$ vertices:

%To estimate $\#S$, where $S$ is, for example, the set of \emph{proper colorings} of a graph, \emph{linear extensions} of a \emph{partially ordered set} (POSET), \emph{perfect matchings} of a \emph{bipartite graph}, etc.
%JVV reduction estimates $\#S$, in, say, $\I$ intermediate phases as follows: Take $S_{0} = S$, and for each $1\leq i\leq \I$, let $S_{i-1}\subseteq S_i$ (or $S_i\subseteq S_{i-1}$), and define $f_i:S_i\rightarrow \{0,1\}$ as follows: $f_i(x) \doteq 1$, if $x\in S_{i-1}$, $f_{i}(x)=0$ otherwise (or vice versa).  Note that $\Expect[f_i]=\#S_{i-1}/\#S_{i}$, thus if a  polynomial number of samples from $\UnifD(S_{i})$ is \emph{sufficient} to estimate $\Expect[f_i]$, e.g., via a Markov chain, then one can estimate $\#S_{i-1}/\#S_{i}$ in polynomial time.
%The number and structure of intermediate subsets depends on the problem in question; perfect matchings and colorings both require $\I = \#E$  phases (each of which adds one edge) \cite{permanent,jurrumSamplingiscounting}, and some problems like linear extensions of a POSET on $n$ items require more sophisticated intermediate phases, where $\I = n\log(n)$ (see  \cite{countingleWinklerBrightwell}).

%Given these intermediate steps, size of $S$ can be estimated with $\varepsilon$ precision, by estimating all $\Expect[f_i]$ with precision $\nicefrac{\varepsilon}{\I}$, and taking the \emph{telescoping product}

\begin{equation}
\#S=\#\Gamma(G_0,k)\ \smash{\prod_{i=1}^{{\cal E}-1}} \Expect[f_{e_i}] =\left(\smash{\prod_{i=1}^{{\cal E}-1}} {\frac{\#\Gamma(G_i,k)}{\#\Gamma(G_{i-1},k)}}\right)\#\Gamma(G_0,k)=\left(\smash{\prod_{i=1}^{{\cal E}-1}} {\frac{\#\Gamma(G_i,k)}{\#\Gamma(G_{i-1},k)}}\right) {n}^k\enspace.
\end{equation}
%Jerrum later tailored this reduction to count the number of proper colorings of a graph \cite{jurrumSamplingiscounting}.

Given these intermediate steps, size of $S$ can be estimated with $\varepsilon$ precision, by estimating all $\Expect[f_{e_i}]$s with precision $\nicefrac{\varepsilon}{\cal E}$. Using the classic static approaches listed in \cref{sec:classicbounds}, the complexity of abstaining a $\varepsilon, \delta$ approximation for a graph with $\cal E$ edges  will be  
$\mathcal{O}({\cal E}\cdot \TRelBound\ \frac{{\cal E}^2}{\varepsilon^2}\log(\frac{\cal E}{\delta}))=\tilde{O}\left( {\cal E}^3 \TRelBound\right)$.

 \tableofcontents
{
\bibliographystyle{alpha}
\bibliography{biblio}
}

\end{document}